\documentclass[11pt]{article}
\usepackage{fullpage}

\usepackage{graphics}
\usepackage[dvips]{epsfig}
\usepackage{comment}
\usepackage{xcolor}
\usepackage{amsmath}
\usepackage{amssymb}
\usepackage{amsfonts}
\usepackage{algorithm, algpseudocode}
\usepackage{graphicx, xspace}
\usepackage{multicol}

\usepackage[small,compact]{titlesec}
\usepackage{times}

\newtheorem{theorem}{Theorem}[section]
\newtheorem{lemma}{Lemma}[section]

\newtheorem{claim}{Claim}[section]
\newtheorem{proposition}{Proposition}[section]

\makeatletter
\newcommand*{\bdiv}{%
  \nonscript\mskip-\medmuskip\mkern5mu%
  \mathbin{\operator@font div}\penalty900\mkern5mu%
  \nonscript\mskip-\medmuskip
}
\makeatother

\newcommand{\qed}{\hfill $\Box$ \bigbreak}
\newenvironment{proof}{\noindent {\bf Proof.}}{\qed}

\newenvironment{proofclaim}{\noindent{\bf Proof of the claim.}}{\hfill$\star$}

\newcommand{\ie}{{i.e.,}\xspace}

\algrenewcommand\algorithmicindent{0.8em}%
\algnewcommand{\True}{\textt{true}}
\algnewcommand{\False}{\textt{false}}
\algblockdefx{Begin}{End}{\textbf{begin}}{\textbf{end}}
\algblockdefx{ForEach}{EndFor}
[1]{\algorithmicfor{} each #1 \algorithmicdo{}}{\algorithmicend{} \algorithmicfor}
\algblockdefx{Loop}{EndLoop}
{\algorithmicrepeat}{\algorithmicend{} \algorithmicrepeat}

\title{{\bf Want to Gather? No Need to Chatter!}}

\author{
S\'{e}bastien Bouchard\thanks{Sorbonne Universit\'e, CNRS, INRIA, LIP6, F-75005 Paris, France, E-mail: sebastien.bouchard@lip6.fr}
\and
Yoann Dieudonn\'{e}\thanks{
MIS Lab., Universit\'{e} de Picardie Jules Verne, France, E-mail: yoann.dieudonne@u-picardie.fr}
\and
Andrzej Pelc\thanks{D\'{e}partement d'informatique, Universit\'{e} du Qu\'{e}bec en Outaouais,
Gatineau, Qu\'{e}bec J8X 3X7,
Canada. E-mail: pelc@uqo.ca.
Supported in part by NSERC discovery grant 2018-03899
and by the Research Chair in Distributed Computing of
the Universit\'{e} du Qu\'{e}bec en Outaouais.}
}
\date{ }

\begin{document}

	\maketitle

	\begin{abstract}

		A team of mobile agents, starting from different nodes of an unknown network, possibly at different times, have to meet at the same node and declare that they have all met. Agents have different labels which are positive integers, and move in synchronous rounds along links of the network. The above task is known as {\em gathering} and was traditionally considered under the assumption that when some agents are at the same node then they can talk, i.e., exchange currently available information. In this paper we ask the question of whether this ability of talking is needed for gathering. The answer turns out to be no.

		Our main contribution are two deterministic algorithms that always accomplish gathering in a much weaker model. We only assume that at any time an agent knows how many agents are at the node that it currently occupies but agents do not see the labels of other co-located agents and cannot exchange any information with them. They also do not see other nodes than the current one.
Our first algorithm works under the assumption that agents know {\em a priori} some upper bound $N$ on the size of the network, and it works in time polynomial in $N$ and in the length $\ell$ of the smallest label. Our second algorithm does not assume any {\em a priori} knowledge about the network but  its complexity is exponential in the size of the network and in the labels of agents.
Its purpose is to show feasibility of gathering under this harsher scenario.

		As a by-product of our  techniques we obtain, in the same weak model, the solution of the fundamental problem of leader election among agents: One agent is elected a {\em leader} and all agents learn its identity. As an application of our  result we also solve, in the same model, the well-known {\em gossiping} problem: if each agent has a message at the beginning, we show how to make all messages known to all agents, even without any {\em a priori} knowledge about the network. If agents know an upper bound $N$ on the size of the network then our gossiping algorithm works in time polynomial in  $N$, in the length of the smallest label and in the length of the largest message.

		\vspace{2ex}

	\end{abstract}

	\vfill

	\vfill

	\thispagestyle{empty}
	\setcounter{page}{0}
	\pagebreak

	%%%%%%%%%%%%%%%%%%%%%%%%%%%%%%%%%%%%%%%%%%%%%%%%%%%%%%%%%%%
	\section{Introduction}
	%%%%%%%%%%%%%%%%%%%%%%%%%%%%%%%%%%%%%%%%%%%%%%%%%%%%%%%%%%%

		\subsection{The background}

			A team of at least two mobile agents, starting from different nodes of an unknown network, possibly at different times, have to meet at the same node and declare that they have all met.
			This basic task,  known  as {\em gathering}, has been extensively studied in the literature (cf., e.g., surveys  \cite{alpern02b,Pe}). Gathering often has to be solved in real life when some people have to meet in a city whose streets form a network. In computer science applications,	mobile agents usually represent software agents navigating in computer networks. In robotics applications, agents may  represent mobile robots  moving along corridors of a contaminated mine that is not accessible to humans. The reason to meet may be to gather in one place samples previously collected by the agents, or to coordinate some future task, such as network maintenance or finding a map of the network.

		\subsection{The model and the problem} \label{sec:model}

			The network is modeled as an undirected connected graph. In the sequel, we refer to it simply as a graph. The number of nodes in the graph is denoted by $n$ and is called the {\em size} of the graph. As it is usually done in the literature concerning rendezvous of mobile agents  (cf. \cite{Pe}), we seek gathering algorithms that do not	rely on the knowledge of node labels, and can work in anonymous graphs as well. The importance of designing such algorithms	is motivated by the fact that, even when nodes are equipped with distinct labels, agents may be unable to perceive them	because of limited perception capabilities,	or nodes may be reluctant to reveal their labels, e.g., due to security or privacy reasons.	Also note that if nodes had distinct labels, then agents might explore the graph and meet in the smallest node, hence gathering would reduce to the well-studied task of exploration.

			On the other hand, we assume that edges incident to a node $v$ have distinct labels in $\{0,\dots,d-1\}$, where $d$ is the degree of $v$. Thus every undirected edge $\{u,v\}$ has two labels, which are called its {\em port numbers} at $u$ and at $v$. Port numbering is {\em local}, i.e., there is no relation between port numbers at $u$ and at $v$. Note that in the absence of port numbers, edges incident to a node would be undistinguishable for agents and thus gathering would be often impossible, as the adversary could prevent an agent from taking some edge incident to the current node.

			We consider a team of at least two agents that start from different nodes of the graph (hence there are at most $n$ agents) and traverse its edges in synchronous rounds.
			%In each round, an agent can either stay at the current round or move to some neighboring node.
			Agents cannot mark visited nodes or traversed edges in any way. The adversary wakes up some of the agents in possibly different rounds. A dormant agent, not woken up by the adversary, is woken up by the first agent that visits its starting node, if such an agent exists. Agents have different labels which are positive integers. Each agent knows its label but it does not know the labels of the other agents. Agents execute the same deterministic algorithm with a parameter which is the label of the agent. Note that in the absence of different labels, deterministic gathering is not always possible, as witnessed by the example of two identical agents starting simultaneously in a ring all of whose edges have port numbers 0 and 1. Every agent starts executing the algorithm in the round of its wake-up. In every round an agent may perform some local computations and {then it executes a move instruction: it either moves to an adjacent node by a chosen port $p$ by executing the instruction {\tt take port $p$}, or stays idle at the current node by executing the instruction {\tt wait}. Once the instruction {\tt take port $p$} (resp. {\tt wait}), performed by an agent at node $v$ in round $r$ is completed, the agent is in round $r+1$ at a node $u\ne v$ (resp. at node $v$).}

			When an agent enters a node, it learns its degree and the port of entry. Agents that cross each other on an edge, traversing it simultaneously in different directions, do not notice this fact. We assume that the memory of the agents is unlimited: from the computational point of view they are modeled as Turing machines. The {time complexity} of an algorithm is the worst-case number of rounds between the wake-up of the earliest agent until the completion of the algorithm.

			The above described assumptions are standard in the literature concerning mobile agents gathering and were used, e.g., in \cite{BDD,DFKP,DPP,TSZ,YY}. However, in all these papers one additional crucial assumption was made. When some agents are at the same node then they can talk, i.e., exchange currently available information. In particular, they see the labels of the co-located agents. This ability of information exchange among co-located agents was very strongly used in previous gathering algorithms. A subset of agents that met at some node could, e.g.,  choose the agent with smallest label among them, and subsequently move together as this smallest agent would move, thus successively decreasing the number of moving groups, eventually meeting at the same node.

			In this paper we ask the question of whether this ability of talking among co-located agents is needed for gathering. The answer turns out to be no. Indeed, we replace the assumption about the ability of talking among agents within a node with the following much weaker assumption:
			\begin{itemize}
				\item In any round, an agent knows how many agents are at the node that it currently occupies but agents do not see the labels of other co-located agents and cannot exchange any information with them.
			\end{itemize}

			This assumption can be implemented by equipping nodes with counters recording the current number of agents at a node. Agents need only be able to read these counters but can be devoid of any transmitting devices.

		%--------------------------------------------------
		\subsection{Our results} \label{subsec:ourresults}

			Our main contribution are two deterministic algorithms that always accomplish gathering in the above described model,  much weaker than the traditional one. Our first algorithm works under the assumption that agents know {\em a priori} some upper bound $N$ on the size of the network, and it works in time polynomial in $N$ and in the length $\ell$ of the smallest label. Our second algorithm does not assume any {\em a priori} knowledge about the network but  its complexity is exponential in the size of the network and in the labels of agents. Its purpose is to show feasibility of gathering under this harsher scenario.

 We believe that accomplishing gathering in the weak model considered in this paper significantly increases the applicability of the solution compared to gathering algorithms working in the traditional model, as it permits us to solve gathering in scenarios where agents are deprived of direct means of communication.

			As a by-product of our  techniques we obtain, in the same weak model, the solution of the fundamental problem of leader election (cf. \cite{Ly}) among agents: One agent is elected a {\em leader} and all agents learn its identity. As an application of our  result we also solve, in the same model, the well-known {\em gossiping} problem: if each agent has a message at the beginning, we show how to make all messages known to all agents, even without any {\em a priori} knowledge about the network. If agents know an upper bound $N$ on the size of the network then our gossiping algorithm works in time polynomial in  $N$, in the length of the smallest label and in the length of the largest message. This result about gossiping is perhaps our most surprising finding: agents devoid of any transmitting devices can solve the most general information exchange problem, as long as they can read counters at visited nodes.

			In the absence of direct communication, a natural idea that comes to mind to solve the gathering problem, is to emulate the unavailable mechanism  of communication using moves of agents. Indeed, this is the basic approach that we adopt. However, this idea, while natural, turns out to be very delicate to put in use. Indeed, without special care, one gets soon to a dangerous situation where communication movements of one group of agents can interfere with communication movements of another closely located group. On top of this difficulty we have another one: movements of agents must serve to accomplish two different goals, one is to communicate with other agents, and the other is to travel in order to meet. Hence we face the danger of ``travelling'' movements interfering with ``communication'' movements. Moreover, it should be stressed that an agent does not, in fact, ``see'' another agent entering or leaving its node: it can only see the cardinality of the set of agents occupying its current node, and, e.g., notice changes in it while waiting at a node. Hence, for example,  an agent will not  notice any change, if one other agent leaves its node trying to communicate something, and another agent enters its node simply navigating in the graph. Of course, all the above challenges were entirely absent in the traditional model. Overcoming these difficulties in the design of our  gathering and gossiping algorithms is the main technical contribution of this paper.

		%--------------------------------------------------

		%--------------------------------------------------
		\subsection{Related work} \label{subsec:relatwork}
		%--------------------------------------------------

			Gathering has been studied both for two mobile agents, when it is usually called rendezvous, and for larger teams. An extensive survey of  randomized rendezvous in various scenarios  can be found in \cite{alpern02b}, cf. also  \cite{alpern95a,alpern02a,baston98}. Deterministic rendezvous in networks has been surveyed in \cite{Pe}. In several papers, the geometric scenario was considered (rendezvous in an interval of the real line, see, e.g., \cite{baston98,baston01,gal99}, or in the plane, see, e.g., \cite{anderson98a,anderson98b}). Gathering more than two agents has been studied, e.g., in \cite{DP,lim96,thomas92,YY}. In~\cite{YY} the authors considered	gathering many agents with unique labels, and gathering many labeled agents in the presence of Byzantine agents was studied in \cite{BDD,DPP}.	The problem was also investigated in the context of multiple robot systems, cf. \cite{CP05,FPSW}, and fault tolerant gathering of robots in the plane was studied, e.g., in \cite{AP06,CP08}.

			For the deterministic setting a lot of effort has been dedicated to the study of the feasibility of rendezvous, and to the time required to achieve this task, when feasible. For instance, deterministic rendezvous with agents equipped with tokens used to mark nodes was considered, e.g., in~\cite{KKSS}. Deterministic rendezvous of two agents that cannot mark nodes but have unique labels was discussed in \cite{DFKP,TSZ}.	These papers are concerned with the time of rendezvous in arbitrary graphs. In \cite{DFKP} the authors show a rendezvous algorithm polynomial in the size of the graph, in the length of the shorter label and in the delay between the starting time of the agents. In \cite{TSZ} rendezvous time is polynomial in the first two of these parameters and independent of the delay.

			Memory required by two anonymous agents to achieve deterministic rendezvous has been studied in \cite{FP} for trees and in  \cite{CKP} for general graphs. Memory needed for randomized rendezvous in the ring is discussed, e.g., in~\cite{KKPM08}.

			Apart from the synchronous model used in this paper, several authors have investigated asynchronous gathering in the plane \cite{CFPS,FPSW} and in network environments \cite{BCGIL,CLP,DGKKP,DPV}. In the latter scenario the agent chooses the edge which it decides to traverse but the adversary controls the speed of the agent. Under this assumption rendezvous in a node cannot be guaranteed even in very simple graphs and hence the rendezvous requirement is relaxed to permit the agents to meet inside an edge.

			A different asynchronous model for gathering in ring networks was considered in \cite{DNN,KMP}. In this model, agents were memoryless but they could perform look operations which gave them a snapshot of the entire network with the positions of all agents in it.

			In \cite{DP1}, the authors considered the problem of network exploration by many agents that could not communicate between them. However, the information available to an agent in each round was much different than in the present paper. Indeed, in \cite{DP1}, agents were getting {\em local traffic reports} consisting of answers to three questions: ``Am I alone in the node?'', ``Did any agent enter this node in this round?'', ``Did any agent leave this node in this round?''. To see that this feedback cannot be derived from our present assumption of knowing the number of agents co-located with an agent in a given round, consider the situation when an agent $a$ stays at a node, and in a given round one other agent leaves the node and another agent enters it. In our present model, agent $a$ does not notice any change, while in the model from \cite{DP1} it gets reports about somebody leaving the node and somebody entering it.

			In \cite{DDPS}, the problem of conveying bits of information using movements of robots was considered in a context much different from ours. Mobile robots were moving in the plane and they could periodically get snapshots of the entire configuration of robots.

	\section{Preliminaries}\label{prelim}

		In this section we introduce some conventions, definitions and procedures that will be used to describe and analyze our algorithms.

		Let us start with some conventions. We say that the execution $E$ of a sequence of instructions lasts $T$ rounds iff during $E$ the agent executes exactly $T$ move instructions of type {\tt wait} or {\tt take port p}. Moreover, if $E$ starts in round $t$, we say that $E$ \emph{is completed in} (resp. \emph{is completed by}) round $t+T$ iff $E$ lasts exactly $T$ rounds (resp. at most $T$ rounds). In our pseudocodes, we often use the shortcut {\tt wait x rounds}: this instruction is equivalent to a sequence of $x$ consecutive instructions {\tt wait}. Sometimes, we also use the instruction {\tt wait until event}: this is equivalent to execute the instruction {\tt wait} in each round, until reaching a round in which the event occurs.

		As mentioned earlier, in order to remedy the lack of direct means of communication, the agents will be required to use movements as a vector for the transmission of information. Hence in our algorithms some of the instructions are dedicated to handling messages via strings. Throughout the paper, we consider only binary strings over the alphabet $\{0,1\}$. The empty string will be denoted by $\epsilon$. $\{0,1\}^+$ denotes the set of non-empty strings. The length of a string $s$,  i.e.,  its number of bits, will be denoted by $|s|$ and its i$th$ bit will be referred to as $s[i]$. We will sometimes use the notation $s[i,j]$ to indicate the substring from $s[i]$ to $s[j]$ ($s[i]$ and $s[j]$ included). If $i>j$ or $i$ (resp. $j$) does not belong to $\{1,\ldots,|s|\}$, we consider that $s[i,j]=\epsilon$. We will also use the functions ${\tt code}$ and ${\tt decode}$ that are borrowed from~\cite{DFKP}. Given a string $s$, ${\tt code}(s)=01$ if $s=\epsilon$, ${\tt code}(s)=s[1]s[1]s[2]s[2]\ldots s[|s|]s[|s|]01$ otherwise. The function ${\tt decode}$ is the inverse function of ${\tt code}$, i.e. ${\tt decode}({\tt code}(s))=s$. The next proposition gives some properties of function ${\tt code}$.

		\begin{proposition}
			\label{pro:code}
			Let $s_1$ be a string belonging to $\{0,1\}^+$. We have the following properties:
			\begin{itemize}
				\item $|{\tt code}(s_1)|$ is even.
				\item ${\tt code}(s_1)[k,k+1]=01$ and $k$ is odd iff $k+1=|{\tt code}(s_1)|$.
				\item Given a string $s_2\ne s_1$ belonging to $\{0,1\}^+$, ${\tt code}(s_1)$ cannot be the prefix of ${\tt code}(s_2)$.
			\end{itemize}
		\end{proposition}

		The concatenation of two strings $s_1$ and $s_2$ will be denoted by $s_1s_2$, and for every non negative integer $i$, the word $w^i=\epsilon$ if $i=0$, and $w^i=ww^{i-1}$ otherwise.

		 We say that a sequence of $i$ integers $(x_1,x_2,\ldots,x_i)$ is a {\em path} from a node $u$ in the graph iff (1) $i=0$ or (2) there exists an edge $e$ between node $u$ and a node $w$ such that the port number of edge $e$ at node $u$ is $x_1$ and $(x_2,\ldots,x_i)$ is a path from node $w$ in the graph. Given a path $\sigma$ from node $u$, we denote by $\mathcal{N}(\sigma,u)$ the set of nodes that are visited by following path $\sigma$ from node $u$, including the starting node and the terminal node: if this terminal node is $v$, we will say that $\sigma$ is a path from node $u$ to node $v$. As for binary strings, the length of a path $p$, corresponding to its number of elements, will be denoted by $|p|$, and its i$th$ element will be referred to as $p[i]$.

		We now recall three basic procedures, known from the literature (cf. \cite{CKP,TSZ}), that will be used to design our gathering algorithm.

		The first procedure, due to Ta-Shma and Zwick \cite{TSZ}, aims at gathering two agents in a graph of unknown size. We will call this procedure ${\tt TZ}(L)$, where $L$ is a non negative integer given as input parameter. In \cite{TSZ}, the authors give a polynomial $\mathcal{P}$ in two variables, increasing in each of the variables, such that, if two agents start executing procedure ${\tt TZ}$ in possibly different round, one with an input parameter $L_1$ and the other with an input parameter $L_2\ne L_1$, then they will meet after at most $\mathcal{P}(n,l)$ rounds since the start of the later agent, where $l$ is the length of the binary representation of the smaller integer between $L_1$ and $L_2$. %Also, if an agent with label $L_i$  performs ${\tt TZ}(L_i)$ for $\mathcal{P}(n,\lceil \log L_i\rceil)$ rounds and the other agent is inert during this time, the meeting is guaranteed.

		The second procedure, explicitly described in \cite{BDD} (and derived from a proof given in \cite{CDK}), allows an agent to perform a graph exploration (i.e, visiting all nodes of the graph) without knowing any upper bound on the size of the graph, using a fixed token placed at the starting node of the agent. (In our applications, the role of the token will be played by a group of agents co-located at the same node). The execution time of the procedure is polynomial in the size of the graph, the exact value of which is learned by the agent when the execution is completed. We call this procedure ${\tt EST}$, for {\em exploration with a stationary token}. The maximum time of execution of the procedure ${\tt EST}$ in a graph of size at most $n\geq2$ is bounded by the value ${\tt T}({\tt EST}(n))=n^5$.

		The third and final procedure borrowed from the literature is based on universal exploration sequences (UXS) and is a corollary of the result of Reingold \cite{Re}. Given any positive integer $N\geq2$, the procedure allows the agent to visit all nodes of any graph of size at most $N$, starting from any node of this graph and coming back to it, using a number of edge traversals that is polynomial in $N$. In the first half of the procedure, which will be called the \emph{effective part}, after entering a node of degree $d$ by some port $p$, the agent computes the port $q$ by which it has to exit; more precisely $q=(p+x_i)\mod d$, where $x_i$ is the corresponding term of the UXS. In the second half of the procedure, which will be called the \emph{backtrack part}, the agent backtracks to its starting node by traversing in the reverse order the entire sequence of edges it has traversed during the first half (some edges may be traversed several times). We denote by ${\tt T}({\tt EXPLO}(N))$ the execution time of procedure ${\tt EXPLO}$ with parameter $N$: this time is polynomial in $N$.

		Throughout the paper, an agent will often need to use the number of agents (including itself) that occupy its current node: this value will be denoted by ${\tt CurCard}$ (which stands for ``current cardinality'').

	\section{Known upper bound on the size of the graph}

		This section is dedicated to the presentation and the analysis of our algorithm ${\tt GatherKnownUpperBound}$ that allows the agents to solve gathering in our model, provided they initially know a common upper bound on the graph size. Throughout this section, this known upper bound is denoted by $N$.

		\subsection{Intuition} \label{intuitionnconnu}

			In order to better describe the high-level idea of our algorithm, let us assume an ideal situation in which all the agents have labels of the same length $\mu$. Let us also assume that all agents are initially woken up at the same time by the adversary.

			In such an ideal situation, we can solve the gathering problem via a strategy made up of consecutive steps $1,2,3,$ etc. The agents start each step $i$ simultaneously. At the beginning of step $i$ they are distributed over $k_i$ distinct nodes. The intended goal is either to get the agents to declare gathering at the same time once step $i$ is completed if $k_i=1$, or otherwise to get all agents to start simultaneously step $i+1$ from at most $k_{i+1}\leq\lfloor \frac{k_i}{2} \rfloor$ distinct nodes. Since the beginning of step $1$ coincides with the round when the adversary wakes up all agents, we can consider a step $i\geq 1$ initiated at the same time by all agents from $k_i$ distinct nodes, and explain how to reach the intended goal mentioned above. This is the purpose of the following paragraphs in which we will use the notion of \emph{invisibility} that can be intuitively defined as follows: two agents (or two groups of agents), executing the same sequence of instructions $X$ starting at the same round but from two distinct nodes, are said to be invisible to each other if they do not meet when executing $X$.

			At the beginning of step $i$, an agent first applies the simple procedure described in Algorithm~\ref{alg:1} that is based on the graph traversal routine ${\tt EXPLO}$ introduced in Section 2. Recall that this procedure consists of two successive parts of equal durations: the effective part, in which each node is visited at least once, and the backtrack part in which the agent executes in reverse order all edge traversals made during the effective part.

			\begin{algorithm}
				\caption{{}}
				\label{alg:1}
				\begin{footnotesize}
				\begin{algorithmic}[1]
					\State c:= ${\tt CurCard}$
					\State execute ${\tt EXPLO}(N)$ and interrupt it as soon as there are more than $c$ agents in my node\label{alg1l2}
					\State wait until the time spent executing Algorithm~\ref{alg:1} is precisely ${\tt T}({\tt EXPLO}(N))$ rounds
				\end{algorithmic}
				\end{footnotesize}
			\end{algorithm}

			It should be noted that the agents that are initially together remain so during the execution of Algorithm~\ref{alg:1}. Hence, after having applied this algorithm, an agent is either (1) with more agents than at the beginning of the step or (2) with exactly the same number of agents.

			Let us first focus on the former situation. This situation necessarily implies that two groups of agents starting step $i$ from two distinct nodes are not  invisible to each other during the execution of the procedure ${\tt EXPLO}(N)$, and thus even not invisible to each other during the effective part of ${\tt EXPLO}(N)$, due to symmetry arguments. Therefore, as soon as the first $\frac{{\tt T}({\tt EXPLO}(N))}{2}$ rounds of the execution of Algorithm 1 are elapsed, we get at most two kinds of groups: the old ones (if any) that have not met any group yet and the new ones (at least one exists) that result from the merge of at least two old groups. In view of the fact that the old groups, which have not merged yet, were invisible to each other when executing the effective part of ${\tt EXPLO}(N)$ and the fact that every new group remains idle during the last $\frac{{\tt T}({\tt EXPLO}(N))}{2}$ rounds, we have the guarantee that each remaining old group meets a new one when executing the backtrack part of ${\tt EXPLO}(N)$. Thus, the execution by every agent of Algorithm~\ref{alg:1} lasts exactly ${\tt T}({\tt EXPLO}(N))$ rounds: when it is completed the agents are all situated in at most $\lfloor \frac{k_i}{2} \rfloor$ distinct nodes. Note that all agents know this, and know that every agent knows this because if an agent ends up sharing its node with more agents than at the beginning of the step, it follows from the above explanations that this is the case for the other agents as well. Hence, we can fulfill our intended goal by just requiring an agent to start step $i+1$ if, after having applied Algorithm~\ref{alg:1}, it is with more agents than at the beginning of the step.

			Now, let us focus on the latter situation in which, after having applied Algorithm~\ref{alg:1}, an agent is exactly with the same number of agents as at the beginning of step $i$. For an agent experiencing this situation, it could be tempting to think that everyone is together as, after all, it has not met any new agent when executing Algorithm~\ref{alg:1}. However, at this stage, this would be premature and thus dangerous. In fact, it can be shown that the current situation implies that: either $k_i=1$ (i.e. all agents are indeed together), or $k_i\geq2$ but the $k_i$ groups are pairwise invisible to each other when executing Algorithm~\ref{alg:1} (the procedure ${\tt EXPLO}(N)$ does not guarantee rendezvous of two agents starting at different nodes). However, these possible invisibilities that could appear detrimental at first glance, we turn them to our advantage. Indeed, we are actually in a convenient situation to allow each agent, using movements, to communicate with the agents of its group, without being disturbed by the agents of the other groups. Those communications aim at ensuring that in each group the agents end up knowing the label of one of them.

			To achieve this, still in step $i$, the agents will act in phases $1,2,3\ldots,\mu$, each lasting $2\cdot{\tt T}({\tt EXPLO}(N))$ rounds. At the beginning of phase $k$, we have the following property $\mathcal{P}(k)$: in every group $G$, there is an agent with label $L$ such that all agents of $G$ know the prefix $p_{k-1}$ of length $k-1$ of the binary representation of $L$ (note that $\mathcal{P}(1)$ is trivially satisfied at the beginning of the first phase). Let us see how these agents proceed to have the property $\mathcal{P}(k+1)$ satisfied when phase $k$ is completed.

			During the first (resp. last) ${\tt T}({\tt EXPLO}(N))$ rounds of phase $k$, the agents having a label whose prefix is $p_{k-1}0$ execute ${\tt EXPLO}(N)$ (resp. remain idle) while the others remain idle (resp. execute ${\tt EXPLO}(N)$). The respective invisibilities of the groups come into the picture as they imply the following crucial property: two agents belonging to two distinct groups and executing ${\tt EXPLO}(N)$ in the first (resp. last) ${\tt T}({\tt EXPLO}(N))$ rounds cannot meet each other within phase $k$. This is crucial because it means that when a set of agents belonging to the same group move together by executing ${\tt EXPLO}(N)$, they visit a node that contains only agents belonging to this set. Hence, by comparing the number of agents sharing its node at the beginning of phase $k$ to the minimum number of agents with which it was at some node when executing ${\tt EXPLO}(N)$, every agent can determine the number of agents of its group for which $p_{k-1}0$ is a prefix. Once phase $k$ is completed, if an agent concludes that there is at least one agent in its group that has a label prefixed by $p_{k-1}0$, then all agents of the group conclude the same, and then $p_k$ is set to $p_{k-1}0$, otherwise $p_k$ is set to $p_{k-1}1$.

			After phase $\mu$, in each group $G$ all agents know the same label $p_\mu$ of an agent belonging to $G$: this label can now be used to break the invisibility of $G$ via procedure ${\tt TZ}$ introduced in the preliminaries section. Indeed, by requiring each agent, once its execution of phase $\mu$ is completed, to execute Algorithm~\ref{alg:2} that relies on procedure ${\tt TZ}$, we can show, using similar arguments as before, that we reach a configuration where: either (1) the number of groups is at most $\lfloor \frac{k_i}{2} \rfloor$ and the cardinality of each of them has increased, or (2) there is only one group and its cardinality has remained unchanged since the beginning of the phase. Note that every agent can detect in which of these two situations it is, just by looking at the cardinality of its group. If the agents are in the first situation, then they start step $i+1$, otherwise they declare that gathering is over: whichever is the case, they can do it at the same time in view of line~\ref{alg2lastline} of Algorithm 2. Hence, in our ideal scenario, we can prove that gathering is declared after at most $\lceil\log N\rceil$ steps, leading to a time complexity polynomial in $N$ and $\mu$.

			Of course, things get more complicated when we are in a scenario that is not necessarily ideal. However, it is through those conceptual principles, together with extra algorithmic ingredients (to circumvent the possible desynchronizations between the wake-ups of the agents as well as the possible different label lengths) that we finally obtain a gathering algorithm working in the general case with a time complexity polynomial in $N$ and in the length of the smallest label among the agents.

			\begin{algorithm}
				\caption{{}}
				\label{alg:2}
				\begin{footnotesize}
				\begin{algorithmic}[1]
					\State $c\gets {\tt CurCard}$
					\State $p_\mu\gets$ the label learned when phase $\mu$ is completed
					\State execute ${\tt TZ}(p_\mu)$ and interrupt it as soon as there are more than $c$ agents in my node or the execution has lasted $\mathcal{P}(N,\mu)$ rounds
					\If{there are $c$ agents in my node}
						\State execute ${\tt EXPLO}(N)$ and interrupt it as soon as there are more than $c$ agents in my node
					\EndIf
					\State wait until the time spent executing Algorithm~\ref{alg:2} is precisely $\mathcal{P}(N,\mu)+{\tt T}({\tt EXPLO}(N))$ rounds \label{alg2lastline}
				\end{algorithmic}
				\end{footnotesize}
			\end{algorithm}

		\subsection{Algorithm}

			Algorithm~\ref{alg:GKUB} gives the pseudocode of ${\tt GatherKnownUpperBound}$. Some of its instructions contain durations that are specified using the value $\mathcal{D}_k$: for every non negative integer $k$, $\mathcal{D}_k$ is precisely equal to $\mathcal{P}(N,k)+3(k+2){\tt T}({\tt EXPLO}(N))$. The execution by an agent of Algorithm~\ref{alg:GKUB} terminates when it declares the gathering is achieved. We show in the proof of correctness that all agents end up making such a declaration in the same round and in the same node. We also show that there is a label $L$ belonging to an agent such that at the end of the execution of Algorithm~\ref{alg:GKUB} by any agent, the variable $\lambda$ is equal to $L$. This shows the promised solution of leader election, as a by-product.

			Algorithm ${\tt GatherKnownUpperBound}$ relies on an important subroutine, which is the function ${\tt Communicate}$. It is this function that will enable breaking possible invisibilities mentioned earlier, by allowing at some point some agent to make its label known. In Algorithm~\ref{alg:GKUB}, each call to function ${\tt Communicate}$ is given three input parameters: an integer $i$ corresponding to the number of bits that will be transmitted/received during the call, a string $s={\tt code}(x)$ where $x$ is the binary representation of the agent's label and a boolean $bool$ indicating whether the agent will attempt to transmit its parameter $s$ or not during the call. Note that the third parameter can appear futile in our current context as in each call to ${\tt Communicate}$ in Algorithm~\ref{alg:GKUB} the input boolean is always true: however it will be useful in our gossiping algorithm (see Section~\ref{sec:elecgoss}) which also relies on this function.

			The output of function ${\tt Communicate}$ is a couple $(l,k)$ such that $l$ is a binary string and $k$ an integer. Similarly as for the input parameter $bool$ of function ${\tt Communicate}$, the second element $k$ will be useful only in our gossiping algorithm. Assume that a group $G$ of agents start executing function ${\tt Communicate}(i,s,bool)$ with the same parameter $i$. Provided some conditions are fulfilled (see Lemma~\ref{lem:com}), $l$ will correspond to one of the strings $s$ given as input parameter by an agent of $G$ for which $bool={\tt true}$, and $k$ will correspond to the number of agents in $G$ for which the input parameter $s$ (resp. $bool$) is equal to $l$ (resp. ${\tt true}$).

			\begin{algorithm}
				\begin{footnotesize}
				\caption{Algorithm~${\tt GatherKnownUpperBound}$ executed by an agent labeled $L$\label{alg:GKUB}}
				\begin{algorithmic}[1]
					\Begin
						\State execute ${\tt EXPLO}(N)$ \label{wakeup}
						\State wait ${\tt T}({\tt EXPLO}(N))$ rounds
						\State $i\gets 1$
						\Repeat
							\State $c\gets$ ${\tt CurCard}$ \label{FirstLineR}
							\State $\lambda\gets0$
							\State execute the following begin-end block and interrupt it before its completion as soon as ${\tt CurCard}>c$
								\Begin \label{alg:fbegin}
									\State wait $\mathcal{D}_{i}$ rounds\label{FirstWait}
									\State execute ${\tt EXPLO}(N)$\label{firstEXphase}
									\State wait ${\tt T}({\tt EXPLO}(N))$ rounds\label{anotherwait}
									\State execute ${\tt EXPLO}(N)$\label{secondEX}
									%\State wait ${\tt T}({\tt EXPLO}(N))$ rounds\label{secondWait}
								\End \label{alg:ebegin}
							\If {${\tt CurCard}>c$}\label{line:firstc}
								\State wait until having seen $\mathcal{D}_{i+1}$ consecutive rounds without any variation of ${\tt CurCard}$ since its latest change
									(the current round and the round of its latest change included).\label{waitaftermeet}
							\Else
								\State $x_L\gets$ the binary representation of $L$ \label{getbin}
								\State $(l,k)\gets {\tt Communicate}(i,{\tt code}(x_L),{\tt true})$ \label{calltocom}
								\If {there exists an odd integer $z<|l|$ such that $l[z,z+1]=01$} \label{conditionlabel}
									\State $\lambda\gets$ the integer for which ${\tt decode}(l[1,z+1])$ is the binary representation \label{getlabel}
								\EndIf
								\State execute the following begin-end block and interrupt it before its completion as soon as ${\tt CurCard}>c$ \label{debsecondb}
									\Begin \label{BC1}
										\State wait ${\tt T}({\tt EXPLO}(N))$ rounds \label{encoreWait}
										\State execute ${\tt TZ}(\lambda)$ for $\mathcal{D}_{i}$ consecutive rounds \label{alg:TZ}
										\State wait ${\tt T}({\tt EXPLO}(N))$ rounds \label{anotherwaitB}
										\State execute ${\tt EXPLO}(N)$ \label{secondEXB}
									\End \label{BC2}
								\If {${\tt CurCard}>c$}\label{line:secondc}
									\State wait until having seen $\mathcal{D}_{i+1}$ consecutive rounds without any variation of ${\tt CurCard}$ since its latest change
										(the current round and the round of its latest change included). \label{waitaftermeet2}
								\EndIf
							\EndIf
							\State wait $\mathcal{D}_{i+1}$ rounds \label{lastwaiting}
							\If {${\tt CurCard}=c$ {\bf and} $\lambda\ne0$}\label{conditiondeclare}
								\State declare the gathering is achieved\label{finsecondb}
							\EndIf
							\State $i\gets i+1$
						\Until{the gathering has been declared achieved}\label{alg:stop}
					\End
				\end{algorithmic}
				\end{footnotesize}
			\end{algorithm}

			\begin{algorithm}
				\begin{footnotesize}
				\caption{Algorithm~${\tt Communicate}(i,s,bool)$\label{alg:com}}
				\begin{algorithmic}[1]
					\Begin
						\State $c\gets {\tt CurCard}$
						\State $k\gets 1$
						\State $l\gets\epsilon$
						\If{$|s|\leq i$ {\bf and} $bool={\tt true}$}\label{comsetdeb}
							\State $participate\gets{\tt true}$
						\Else
							\State $participate\gets{\tt false}$
						\EndIf\label{comsetfin}
						\For {$j \leftarrow 1$ {\bf to} $i$}
							\If {$participate={\tt true}$ {\bf and} $j\leq|s|$ {\bf and} $s[j]=0$}\label{condition}
								\State wait ${\tt T}({\tt EXPLO}(N))$ rounds\label{firstlinealpha}
								\State execute ${\tt EXPLO}(N)$\label{COM:process0}
								\State wait $3{\tt T}({\tt EXPLO}(N))$ rounds\label{lastlinealpha}
								\State $l[j]\gets0$
								\If {$c>1$}
									\State $k\gets$ the smallest value reached by ${\tt CurCard}$ during the latest execution of line~\ref{COM:process0}
								\EndIf
							\Else \label{elseblock}
								\State /* $participate={\tt false}$ or $s[j]=1$ or $j>|s|$ */
								\State wait $3{\tt T}({\tt EXPLO}(N))$ rounds\label{firstlinebeta}
								\State execute ${\tt EXPLO}(N)$\label{COM:process1}
								\State wait ${\tt T}({\tt EXPLO}(N))$ rounds\label{lastlinebeta}
								\State $c'\gets$ the smallest value reached by ${\tt CurCard}$ during the latest execution of line~\ref{COM:process1}\label{receive}
								\If {$c=1$ {\bf or} $c'=c$}\label{comtogether}
									\State $l[j]\gets1$
								\Else
									\State $l[j]\gets0$\label{else2}
									\State $participate\gets{\tt false}$
									\State $k\gets c-c'$\label{fincomtogether}
								\EndIf
							\EndIf
						\EndFor
						\State {\bf return} $(l,k)$
					\End
				\end{algorithmic}
				\end{footnotesize}
			\end{algorithm}

		\subsection{Correctness and complexity analysis}

			We start our analysis with the following technical lemma about function ${\tt Communicate}$ (cf. Algorithm~\ref{alg:com}). Every execution by an agent of the for loop within function ${\tt Communicate}(k,*,*)$ will be viewed as a series of consecutive steps $j=1,2,\ldots,k$, where step $j$ is the part of its execution corresponding to the $j$th iteration of the for loop. The part of the execution by an agent which is before step $1$ will be called step $0$.

			\begin{lemma}
				\label{lem:com}
				Let $G$ be the set of all agents located at a given node $v$ in a given round $t$. Assume that all agents of $G$ start executing function ${\tt Communicate}(i,s,bool)$ in round $t$. Also assume that in all these executions the following three conditions are satisfied:
				\begin{itemize}
					\item The input parameter $i$ is a positive integer that is the same for all agents.
					\item The input parameter $bool$ is a boolean and the input parameter $s={\tt code}(x)$ for some binary string $x$.
					\item There is only one agent in $G$, or during each call to procedure ${\tt EXPLO}(N)$ by any agent $A$ of $G$, there is a round when agent $A$ is in a node $u\ne v$ with no agent that does not belong to $G$.
				\end{itemize}
				Let $\overline{G}$ be the set of agents of $G$ for which $bool=true$ and $|s|\leq i$ in their respective call to function ${\tt Communicate}(i,s,bool)$. The execution of function ${\tt Communicate}(i,s,bool)$ by each agent of $G$ is completed at node $v$, in round $t+5i{\tt T}({\tt EXPLO}(N))$, and its return value is a couple $(l,k)$ having the following properties.
				\begin{itemize}
					\item If $\overline{G}\ne\emptyset$, then $l=\sigma 1^{i-|\sigma|}$, where $\sigma$ is the lexicographically smallest input parameter $s$ used by an agent of $\overline{G}$, and $k$ is the number of agents belonging to $\overline{G}$ for which the input parameter $s=\sigma$.
					\item Otherwise, $l=1^i$ and $k=1$.
				\end{itemize}
			\end{lemma}

			\begin{proof}
				In this proof, each time we refer to a specific line, it is one of Algorithm~\ref{alg:com}, and thus we omit to mention it, in order to facilitate the reading.
				In view of Algorithm~\ref{alg:com} and the first two conditions of the lemma, we know that for each $1\leq j\leq i$ the execution by each agent of $G$ of the $j$th step of function ${\tt Communicate}$ is started (resp. completed) at node $v$ in round $t+5(j-1){\tt T}({\tt EXPLO}(N))$ (resp. in round $t+5j{\tt T}({\tt EXPLO}(N))$). Note that in the case where $\overline{G}=\emptyset$, all agents always stay together from round $t$ to $t+5i{\tt T}({\tt EXPLO}(N))$. Indeed, in each step it executes, every agent applies the instructions of lines~\ref{firstlinebeta} to~\ref{lastlinebeta} because its variable $participate$ is necessarily set to false (cf. lines~\ref{comsetdeb} to~\ref{comsetfin}) given that $\overline{G}=\emptyset$. By the third condition of the lemma, this implies that in the execution of each step $1\leq j\leq i$ by every agent, the condition of line~\ref{comtogether} evaluates to true, $l[j]$ is assigned the bit $1$, while $k$ remains unchanged. Hence, when the execution of step $i$ is completed by an agent in round $t+5i{\tt T}({\tt EXPLO}(N))$, its variable $l$ is equal to $1^i$ and its variable $k$ has still the same value as in round $t$ i.e., $1$. Therefore the lemma is true when $\overline{G}=\emptyset$.

				Now, consider the complementary case where $\overline{G}\ne\emptyset$. For any binary string $w$, denote by $\overline{G}_w$ the set of all agents of $\overline{G}$ for which the input parameter $s$ is prefixed by $w$. To analyse properly the current case, we introduce the property $\Psi(p)$ which involves $\overline{G}_w$. We say that $\Psi(p)$ holds iff for each agent the following three conditions are met when its execution of step $p$ is completed in round $t+5p{\tt T}({\tt EXPLO}(N))$: $(1)$ its variable $l$ is equal to $\sigma[1,\min\{p,|\sigma|\}]1^{\max\{0,p-|\sigma|\}}$, $(2)$ if $1\leq p\leq |\sigma|$ and $\sigma[p]=0$, then its variable $k$ is equal to the number of agents belonging to $\overline{G}_{\sigma[1,\min\{p,|\sigma|\}]}$, and $(3)$ its variable $participate$ is equal to $true$ iff the agent belongs to $\overline{G}_{\sigma[1,\min\{p,|\sigma|\}]}$.

				To proceed further, we need to prove the following claim.

				\begin{claim}
					\label{claimstrong}
					Property $\Psi(p)$ holds for every $0\leq p\leq i$.
				\end{claim}

				\begin{proofclaim}
					We prove the claim by induction on $p$. Note that $\Psi(0)$ is satisfied. Thus, let us assume that there exists a non negative integer $p<i$ such that $\Psi(p)$ is true and let us prove that $\Psi(p+1)$ is true as well.

					First assume that $|\sigma|\leq p$ or $\sigma[p+1]=1$. By the inductive hypothesis, we know that during the execution by an agent $A$ of step $p+1$, the condition of line~\ref{condition} evaluates to true iff agent $A$ belongs to $\overline{G}_{\sigma[1,\min\{p,|\sigma|\}]}$ while its input parameter $s$ has length at least $p+1$ and has bit $0$ in position $p+1$. Note that if $|\sigma|\leq p$, an agent cannot belong to $\overline{G}_{\sigma[1,\min\{p,|\sigma|\}]}=\overline{G}_{\sigma}$ while having an input parameter $s$ of length at least $p+1$ because, according to Proposition~\ref{pro:code}, no input parameter $s$ different from $\sigma$ can be prefixed by $\sigma$ (by the second condition of the lemma, each input parameter $s$ is the image of a binary string under function ${\tt code}$). Also note that if $\sigma[p+1]=1$, then the condition of line~\ref{condition} evaluates to true in an execution of step $p+1$ by an agent $A$ only if $A$ belongs to $\overline{G}_{\sigma[1,p]0}$: however this set is empty, as otherwise we would get a contradiction with the definition of $\sigma$. Hence, the condition of line~\ref{condition} evaluates to false during the execution by every agent of step $p+1$. From round $t+5p{\tt T}({\tt EXPLO}(N))$ to round $t+5(p+1){\tt T}({\tt EXPLO}(N))$, all the agents execute the else block starting at line~\ref{firstlinebeta}: in particular they are always together when executing lines~\ref{firstlinebeta} to~\ref{lastlinebeta}. By the third condition of the lemma, it follows that the condition of line~\ref{comtogether} evaluates to true during the execution by every agent of step $p+1$. Thus, in view of the first condition of property $\Psi(p)$, if $p+1>|\sigma|$ (resp. $p+1\leq|\sigma|$) then the variable $l$ of each agent is equal to $\sigma[1,|\sigma|]1^{(p-|\sigma|)+1}$ (resp. $\sigma[1,p]1=\sigma[1,p+1]$) when its execution of step $p+1$ is completed in round $t+5(p+1){\tt T}({\tt EXPLO}(N))$. This implies that the variable $l$ of each agent is, at that point, equal to $\sigma[1,\min\{p+1,|\sigma|\}]1^{\max\{0,p+1-|\sigma|\}}$, and thus the first condition of property $\Psi(p+1)$ holds. Moreover, the second condition of property $\Psi(p+1)$ directly follows from the assumption that $|\sigma|\leq p$ or $\sigma[p+1]=1$. Finally, we know that if $|\sigma|\leq p$, then $\overline{G}_{\sigma[1,\min\{p,|\sigma|\}]}=\overline{G}_{\sigma[1,\min\{p+1,|\sigma|\}]}=\overline{G}_{\sigma}$. We also know that if $\sigma[p+1]=1$, then $\overline{G}_{\sigma[1,p]}=\overline{G}_{\sigma[1,p]1}=\overline{G}_{\sigma[1,p+1]}$: indeed no agent can have a parameter $s$ that is a prefix of $\sigma$ by Proposition~\ref{pro:code}, and according to the explanations given above $\overline{G}_{\sigma[1,p]0}=\emptyset$ if $\sigma[p+1]=1$. As a result, $\overline{G}_{\sigma[1,\min\{p,|\sigma|\}]}=\overline{G}_{\sigma[1,\min\{p+1,|\sigma|\}]}$. Hence, in view of the fact that no agent changes its variable $participate$ during its execution of step $p+1$, the third condition of property $\Psi(p+1)$ holds as well. This closes the analysis when $|\sigma|\leq p$ or $\sigma[p+1]=1$.

					Now assume that we are in the complementary case in which $\sigma[p+1]=0$ and $|\sigma|\geq p+1$. As mentioned earlier, during the execution by an agent $A$ of step $p+1$, the condition of line~\ref{condition} evaluates to true iff agent $A$ belongs to $\overline{G}_{\sigma[1,\min\{p,|\sigma|\}]}$ while its input parameter $s$ has length at least $p+1$ and has bit $0$ in position $p+1$. As $\sigma[p+1]=0$ and $|\sigma|\geq p+1$, we know that $\sigma[1,\min\{p,|\sigma|\}]0=\sigma[1,p+1]$, and thus the condition of line~\ref{condition} evaluates to true during the execution by an agent $A$ of step $p+1$ iff agent $A$ belongs to $\overline{G}_{\sigma[1,p+1]}$. (Note that since  $\sigma[p+1]=0$ and $|\sigma|\geq p+1$, there is at least one agent in $\overline{G}_{\sigma[1,p+1]}$). Hence, from round $t+5p{\tt T}({\tt EXPLO}(N))$ to round $t+5(p+1){\tt T}({\tt EXPLO}(N))$, all the agents of $\overline{G}_{\sigma[1,p+1]}$ (resp. $G\setminus\overline{G}_{\sigma[1,p+1]}$ if any) execute the then block starting at line~\ref{firstlinealpha} (resp. the else block starting at line~\ref{firstlinebeta}): in particular they are always together when executing lines~\ref{firstlinealpha} to~\ref{lastlinealpha} (resp. lines~\ref{firstlinebeta} to~\ref{lastlinebeta}).
					From the waiting periods of lines~\ref{firstlinealpha},~\ref{lastlinealpha},~\ref{firstlinebeta} and~\ref{lastlinebeta}, it follows that during the entire execution of procedure ${\tt EXPLO}(N)$ at line~\ref{COM:process0} (resp. at line~\ref{COM:process1}) by an agent of $\overline{G}_{\sigma[1,p+1]}$ (resp. $G\setminus\overline{G}_{\sigma[1,p+1]}$) in step $p+1$, the agents of $G\setminus\overline{G}_{\sigma[1,p+1]}$ (resp. $\overline{G}_{\sigma[1,p+1]}$) are idle at node $v$. All of this, the inductive hypothesis and the third condition of the lemma imply that when the execution of step $p+1$ is completed by an agent of $\overline{G}_{\sigma[1,p+1]}$ in round $t+5(p+1){\tt T}({\tt EXPLO}(N))$, its variable $k$, its variable $l$, and its variable $participate$ are respectively equal to the number of agents in $\overline{G}_{\sigma[1,p+1]}$, to $l[1,p]0=\sigma[1,p+1]$ and to the boolean true (its variable $participate$ remains unchanged in step $p+1$). It also follows that the variable $c'$ is assigned the number of agents in $G\setminus\overline{G}_{\sigma[1,p+1]}$ in the execution by every agent of $G\setminus\overline{G}_{\sigma[1,p+1]}$ of line~\ref{receive} in step $p+1$. Note that since there is at least one agent in $\overline{G}_{\sigma[1,p+1]}$, we have $c\geq2$ and $c'\ne c$ for each agent of $G\setminus\overline{G}_{\sigma[1,p+1]}$ when it executes line~\ref{comtogether} in step $p+1$ (the condition of this line then evaluates to false). Hence, by lines~\ref{else2} to~\ref{fincomtogether}, we know that when the execution of step $p+1$ is completed by an agent of $G\setminus\overline{G}_{\sigma[1,p+1]}$ in round $t+5(p+1){\tt T}({\tt EXPLO}(N))$, its variable $k$, its variable $l$, and its variable $participate$ are respectively equal to the number of agents in $\overline{G}_{\sigma[1,p+1]}$, to $l[1,p]0=\sigma[1,p+1]$ and to the boolean false. As a result, when $\sigma[p+1]=0$ and $|\sigma|\geq p+1$, $\Psi(p+1)$ holds. This ends the proof of the claim by induction.
				\end{proofclaim}

				By Claim~\ref{claimstrong}, property $\Psi(i)$ holds. This means that in the execution of function ${\tt Communicate}$ by every agent we have $l=\sigma[1,\min\{i,|\sigma|\}]1^{max\{0,i-|\sigma|\}}$ when the function is completed. Since $\overline{G}\ne\emptyset$, we know that $i\geq |\sigma|$, and thus, for each agent, $l=\sigma 1^{i-|\sigma|}$ at that time. Consequently, all that remains to be shown is that when the execution of step $p+1$ is completed by an agent in round $t+5i{\tt T}({\tt EXPLO}(N))$, the variable $k$ of the agent is equal to the number of agents belonging to $\overline{G}$ for which the input parameter $s=\sigma$.

				By the third conditions of properties $\Psi(|\sigma|+1),\ldots,\Psi(i)$ and the condition of line~\ref{condition}, the agents stay together from the beginning of their execution of step $|\sigma|+1$ till the end of their execution of step $i$, and thus during this period the variable $k$ of each agent remains unchanged. Moreover, we know by Proposition~\ref{pro:code} that every input parameter $s$ is of even length, which means that there is no agent having an input parameter $s$ that is equal to $\sigma[1,|\sigma|-1]$. We also know that there is no agent of $\overline{G}$ having a parameter $s$ prefixed by $\sigma[1,|\sigma|-1]0$, as otherwise we get a contradiction with the definition of $\sigma$ (its last bit is $1$). As a result, we have $\overline{G}_{\sigma[1,|\sigma|-1]}=\overline{G}_{\sigma}$, and since $\sigma[|\sigma|]=1$, we know in view of the third condition of property $\Psi(|\sigma|)$ and the condition of line~\ref{condition}, that all the agents stay together during step $|\sigma|$. Hence, the variable $k$ of each agent remains unchanged from the beginning of its execution of step $|\sigma|$ till the end of its execution of step $i$. Note that in view of the definition of  function ${\tt code}$, we have $\sigma[|\sigma|-1]=0$. Also note that  $\Psi(|\sigma|-1)$ holds by Claim~\ref{claimstrong}. It follows that when the execution of step $|\sigma|-1$ is completed by an agent $A$ in round $t+5(|\sigma|-1){\tt T}({\tt EXPLO}(N))$ (and, by extension, when the execution of step $i$ is completed by agent $A$ in round $t+5i{\tt T}({\tt EXPLO}(N))$) the variable $k$ of agent $A$ is equal to the number of agents in $\overline{G}_{\sigma[1,|\sigma|-1]}$. However, we have shown above that $\overline{G}_{\sigma[1,|\sigma|-1]}=\overline{G}_{\sigma}$. Since, by Proposition~\ref{pro:code}, no parameter $s$ distinct from $\sigma$ can be prefixed by $\sigma$, $\overline{G}_{\sigma}$ is precisely the set of the agents for which the input parameter $s$ is equal to $\sigma$. Hence, at the end of its execution of step $i$, the variable $k$ of an agent is equal to the number of agents of $\overline{G}$ for which the input parameter $s$ is equal to $\sigma$, which concludes the proof of the lemma.
			\end{proof}

			Now we can turn attention to Algorithm ${\tt GatherKnownUpperBound}$. All further results of this section are stated assuming that the algorithm that is executed by an agent when it wakes up is ${\tt GatherKnownUpperBound}$ (cf. Algorithm~\ref{alg:GKUB}). Every execution by an agent of the repeat loop of Algorithm~\ref{alg:GKUB} will be viewed as a series of consecutive phases $i=1,2,3,\ldots$, where phase $i$ is the part of its execution corresponding to the $i$th iteration of the repeat loop. The part of the execution by an agent which is before phase $1$ will be called phase $0$. Given an agent $A$, we denote by $t_{A,i}$ the round, if any, when agent $A$ starts executing phase $i$.

			Lemma~\ref{lem:invisible} will allow us to use Lemma~\ref{lem:com} in the proof of Lemma~\ref{lem:crucial}. At a high level, Lemma~\ref{lem:invisible} exploits the possible ``invisibilities'' presented in the intuitive explanations of Section~\ref{intuitionnconnu}, and will contribute to showing that function ${\tt Communicate}$ is used properly in every execution of Algorithm ${\tt GatherKnownUpperBound}$.

			\begin{lemma}
				\label{lem:invisible}
				Consider the set $G$ of all agents located at a given node $v$ in a given round $t$ and assume that the following three conditions are satisfied:
				\begin{itemize}
					\item All agents of $G$ start executing phase $i\geq1$ in round $t$.
					\item For every couple of agents $X$ and $Y$ (not necessarily in $G$), $|t_{X,i}-t_{Y,i}|\leq\mathcal{D}_{i}$
					\item The condition of line~\ref{line:firstc} in Algorithm~\ref{alg:GKUB} evaluates to false in the execution by every agent (in $G$ or outside of $G$) of phase $i$.
				\end{itemize}
				We have the following two properties.
				\begin{itemize}
					\item Property~1: For every agent $X$, $|t-t_{X,i}|\leq\frac{{\tt T}({\tt EXPLO}(N))}{2}$.
					\item Property~2: There is only one agent in $G$, or during each call to procedure ${\tt EXPLO}(N)$ made by every agent $A$ of $G$ when executing function ${\tt Communicate}$ in phase $i$, there is a round when agent $A$ is in a node $u\ne v$ with no agent that does not belong to $G$.
				\end{itemize}
			\end{lemma}

			\begin{proof}
				Let $A$ (resp. $X$) be an agent of $G$ (resp. an agent that does not belong to $G$).
				We first prove property~1. Assume, by contradiction, that the three conditions of the statement are satisfied and $\frac{{\tt T}({\tt EXPLO}(N))}{2}<|t-t_{X,i}|\leq\mathcal{D}_{i}$. It follows that either agent $A$ has time to execute entirely the effective part of the first execution of ${\tt EXPLO}(N)$ of phase $i$  together with the other agents of $G$ (cf. line~\ref{firstEXphase} of Algorithm~\ref{alg:GKUB}) while $X$ is waiting in phase $i$ (cf. line~\ref{FirstWait} of Algorithm~\ref{alg:GKUB}), or agent $X$ has time to execute entirely the effective part of the first execution of ${\tt EXPLO}(N)$ of phase $i$ while $A$ is waiting with the other agent of $G$ in phase $i$. Hence, the condition of line~\ref{line:firstc} in Algorithm~\ref{alg:GKUB} evaluates to true in the execution by agent $A$ of phase $i$. This is a contradiction, which proves the first property of this lemma.

				Now, we prove property 2. Note that this property is necessarily true if $G$ contains at most one agent. Suppose by contradiction that the three conditions of the statement are satisfied, $G$ contains at least two agents, and there exists a call to procedure ${\tt EXPLO}(N)$ made by agent $A$ (when executing function ${\tt Communicate}$ in phase $i$), during which agent $A$ is with an agent that does not belong to $G$, each time it is at a node $u\ne v$.

				Since $G$ contains at least two agents, there exists a node $x$ in the graph from which no agent starts executing phase $i$ of Algorithm~${\tt GatherKnownUpperBound}$ because the number of agents is at most the size of the graph. Moreover, in view of the fact that the third condition of the statement is satisfied, we know that each agent starts function ${\tt Communicate}$ in phase $i$ and each of its steps at the node (different from $x$) from which it started phase $i$. This implies that there is a positive integer $k$ at most equal to $i$ such that agent $A$ meets an agent $B$, not belonging to $G$, at node $x$ in some round $t_x$ when agent $A$ is executing ${\tt EXPLO}(N)$ in the $k$th step of ${\tt Communicate}$. Denote by $s_A$ (resp. $s_B$) the round when agent $A$ (resp. agent $B$) starts the $k$th step of ${\tt Communicate}$, and denote by $v'$ the node from which agent $B$ starts phase $i$ as well as the $i$ steps of ${\tt Communicate}$ called in this phase. Note that since each step of ${\tt Communicate}$ lasts exactly $5{\tt T}({\tt EXPLO}(N))$ rounds (cf. lines~\ref{firstlinealpha} to~\ref{lastlinealpha} and lines~\ref{firstlinebeta} to~\ref{lastlinebeta} of Algorithm~\ref{alg:com}) and the condition of line~\ref{line:firstc} in Algorithm~\ref{alg:GKUB} evaluates to false in the execution of phase $i$ by every agent (in $G$ or outside of $G$), we have the following claim.

				\begin{claim}
					\label{N:claim1}
					$s_B-s_A=t_{B,i}-t_{A,i}$.
				\end{claim}

				Recall that agent $A$ is executing ${\tt EXPLO}(N)$ in the $k$th step of ${\tt Communicate}$ when it meets agent $B$ at node $x$ in round $t_x$. According to Claim~\ref{N:claim1}, by the first property and the fact that within each step of ${\tt Communicate}$, an execution of ${\tt EXPLO}(N)$ is directly preceded  and followed by a waiting period of length at least ${\tt T}({\tt EXPLO}(N))$, we know that agent $B$ is also processing the $k$th step of ${\tt Communicate}$ in round $t_x$. Note that in view of the definition of node $x$, the node $v'$ from which agent $B$ started phase $i$ of Algorithm~${\tt GatherKnownUpperBound}$ and the $k$th step of function ${\tt Communicate}$ is different from node $x$. This implies that agent $B$ is executing ${\tt EXPLO}(N)$ in round $t_x$, and thus either each of agents $A$ and $B$ is executing ${\tt EXPLO}(N)$ at line~\ref{COM:process0} of Algorithm~\ref{alg:com} in the $k$th step of function ${\tt Communicate}$ in round $t_x$, or each of agents $A$ and $B$ is executing ${\tt EXPLO}(N)$ at line~\ref{COM:process1} of Algorithm~\ref{alg:com} in the $k$th step of function ${\tt Communicate}$ in round $t_x$. Indeed, by Claim~\ref{N:claim1} and by the first property, we have $|s_B-s_A|\leq \frac{{\tt T}({\tt EXPLO}(N))}{2}$.  Since within every step a possible execution of line~\ref{COM:process0} (resp. line~\ref{COM:process1}) in Algorithm~\ref{alg:com} occurs just after a waiting period of length exactly ${\tt T}({\tt EXPLO}(N))$ (resp. $3{\tt T}({\tt EXPLO}(N))$) starting at the beginning of the step, in round $t_x$ we cannot have an agent that is executing ${\tt EXPLO}(N)$ at line~\ref{COM:process0} of Algorithm~\ref{alg:com} in the $k$th step, while the other agent is executing ${\tt EXPLO}(N)$ at line~\ref{COM:process1} of Algorithm~\ref{alg:com} in the $k$th step.

				Consider only the case where each of agents $A$ and $B$ is executing ${\tt EXPLO}(N)$ at line~\ref{COM:process0} of Algorithm~\ref{alg:com} in the $k$th step in round $t_x$, as the complementary case can be analyzed in a similar way.

				When agent $A$ (resp. $B$) starts procedure ${\tt EXPLO}(N)$ from node $v$ (resp. node $v'$) in round $s_A+{\tt T}({\tt EXPLO}(N))$ (resp. $s_B+{\tt T}({\tt EXPLO}(N))$), agents $A$ and $B$ meet at node $x$ in round $t_x=s_A+{\tt T}({\tt EXPLO}(N))+\zeta$, where $\zeta$ is some non negative integer less than ${\tt T}({\tt EXPLO}(N))$. However, by Claim~\ref{N:claim1}, $s_B-s_A=t_{B,i}-t_{A,i}$. This implies that when each agent of $G$ (resp. agent $B$) starts the first execution of procedure ${\tt EXPLO}(N)$ of phase $i$ from node $v$ in round $t_{A,i}+\mathcal{D}_{i}$ (resp. from node $v'$ in round $t_{B,i}+\mathcal{D}_{i}$), agent $B$ and all the agents of $G$ are together at node $x$ in round $t_{A,i}+\mathcal{D}_{i}+\zeta<t_{A,i}+\mathcal{D}_{i}+{\tt T}({\tt EXPLO}(N))$. Hence, the third condition of the statement of the theorem is not satisfied, which is a contradiction. This proves property 2 and concludes the proof of the lemma.
			\end{proof}

			The next lemma proves several properties that will serve directly to establish correctness and analyze the complexity of Algorithm ${\tt GatherKnownUpperBound}$. The lemma
			makes use of the notation $\ell$ which refers to the length of the binary representation of the smallest label among the agents. It also makes use of the set $\Phi_i$ which is defined as follows. Given a positive integer $i$, $\Phi_i$ is the set of couples $(v,t)$ such that $(v,t)\in \Phi_i$ iff there exists an agent that starts executing phase $i$ from node $v$ in round $t$. The cardinality of $\Phi_i$ is denoted by $|\Phi_i|$.

			\begin{lemma}
				\label{lem:crucial}
				Let $i$ be a non negative integer such that no agent declares that gathering is achieved before executing phase $i$. For every couple of agents $A$ and $B$, we have the following two properties:
				\begin{itemize}
					\item $P_1(i)$: $|t_{A,i}-t_{B,i}|\leq\mathcal{D}_{i}$.
					\item $P_2(i)$: if $t_{A,i+1}$ and $t_{B,i+1}$ exist, and agents $A$ and $B$ are at the same node in round $t_{A,i+1}$, then $t_{A,i+1}=t_{B,i+1}$.
				\end{itemize}
				Moreover, denote by $F$ the earliest agent (or one of the earliest agents) that starts executing phase $i$. At least one of the following three properties holds.
				\begin{itemize}
					\item $P_3(i)$: $|\Phi_{i+1}|\leq \lfloor\frac{|\Phi_i|}{2}\rfloor$. Furthermore, either every agent starts phase $i+1$ during the interval $\{t_{F,i}+\mathcal{D}_i+2\mathcal{D}_{i+1},\ldots,t_{F,i}+2\mathcal{D}_i+2\mathcal{D}_{i+1}+3{\tt T}({\tt EXPLO}(N))\}$, or every agent starts phase $i+1$ during the interval $\{t_{F,i}+2\mathcal{D}_{i+1}+\mathcal{D}_i+(5i+4){\tt T}({\tt EXPLO}(N)),\ldots,t_{F,i}+2\mathcal{D}_{i+1}+2\mathcal{D}_i+(5i+6){\tt T}({\tt EXPLO}(N))\}$.
					\item $P_4(i)$: All agents declare that gathering is achieved at the same node in round $t_{F,i}+\mathcal{D}_{i+1}+2\mathcal{D}_{i}+(5i+6){\tt T}({\tt EXPLO}(N))$. Furthermore, there exists a label $L$ belonging to an agent of the team such that the variable $\lambda$ of every agent in Algorithm~\ref{alg:GKUB} is equal to $L$ in round $t_{F,i}+\mathcal{D}_{i+1}+2\mathcal{D}_{i}+(5i+6){\tt T}({\tt EXPLO}(N))$.
					\item $P_5(i)$: $i<2\ell+2$ and every agent $A$ starts executing phase $i+1$ after having spent exactly $2{\tt T}({\tt EXPLO}(N))$ rounds if $i=0$ (resp. exactly $\mathcal{D}_{i+1}+2\mathcal{D}_{i}+(5i+6){\tt T}({\tt EXPLO}(N))$ rounds if $i>0$) in phase $i$.
				\end{itemize}
			\end{lemma}

			\begin{proof}
				Throughout this proof, each time we refer to a specific line, it is one of Algorithm~\ref{alg:GKUB}, and thus we omit to mention it, in order to facilitate the reading.

				The proof is by induction on $i$. Let us first consider the base case $i=0$.
				According to line~\ref{wakeup}, we know that the delay between the starting times of phase $0$ of any two agents is at most $\frac{{\tt T}({\tt EXPLO}(N))}{2}$ rounds. Indeed, as soon as the execution of the effective part of ${\tt EXPLO}(N)$ in phase $0$ is completed by the first woken up agent, all nodes have been visited at least once and there cannot remain any dormant agent in the network. Hence, property $P_1(0)$ is true.

				Concerning property $P_2(0)$, note that every agent starts phase $1$ from its initial node in view of the backtrack part of ${\tt EXPLO}(N)$. Moreover, since the delay between the starting times of phase $0$ of any two agents is at most $\frac{{\tt T}({\tt EXPLO}(N))}{2}$ rounds, we know that when an agent $A$ starts phase $1$, every other agent $B$ is waiting in its initial node ($B$ has at least started the waiting period in phase $0$ but it cannot have started any move in phase $1$ because in this phase, the first move is necessarily preceded by a waiting period of length at least $\mathcal{D}_1>\frac{{\tt T}({\tt EXPLO}(N))}{2}$). Thus, $P_2(0)$ is also true because the initial nodes of the agents are pairwise distinct. To conclude the base case, note that every agent $A$ starts phase $1$ in round $t_{A,0}+2 {\tt T}({\tt EXPLO}(N))$ in view of line~\ref{wakeup}. Hence, since $0<\ell$ we also know that $P_5(0)$ is satisfied. As a result, the theorem holds when $i=0$.

				Now, let us assume that there exists an integer $k\geq0$ such that the lemma is true when $i=k$, and let us prove that the lemma is still true when $i=k+1$ and no agent has declared the gathering achieved before executing phase $k+1$. We have the following claim.

				\begin{claim}
					\label{cP1}
					Property $P_1(k+1)$ is satisfied.
				\end{claim}

				\begin{proofclaim}
					Since, by assumption no agent declares the gathering achieved before executing phase $k+1$, we know that property $P_4(k)$ is not satisfied and thus property $P_3(k)$ or $P_5(k)$ is true. If property $P_5(k)$ is true, then agents $A$ and $B$ spend exactly the same time to execute phase $k$, and thus $|t_{A,k+1}-t_{B,k+1}|\leq\mathcal{D}_{k}\leq\mathcal{D}_{k+1}$ in view of property $P_1(k)$. If $P_5(k)$ is false and $P_3(k)$ is true, then $|t_{A,k+1}-t_{B,k+1}|\leq \mathcal{D}_k+3{\tt T}({\tt EXPLO}(N))\leq \mathcal{D}_{k+1}$. Hence $P_1(k+1)$ is true, which concludes the proof of this claim.
				\end{proofclaim}

				In the light of Claim~\ref{cP1}, it remains to show that property $P_2(k+1)$ and at least one property among $P_3(k+1)$, $P_4(k+1)$ and $P_5(k+1)$ hold. To this end, we introduce the property $\Gamma(x)$ that is satisfied iff there exists at least one execution of phase $x$ by some agent in which the condition of line~\ref{line:firstc} or line~\ref{line:secondc} evaluates to true. The rest of this proof is conducted by discussing the cases where $\Gamma(k+1)$ holds and where it does not. For each $(v,t)$ belonging to $\Phi_{k+1}$, denote by $G_{(v,t)}$ the group of agents that start executing phase $k+1$ from node $v$ in round $t$. Note that in view of property $P_2(k)$, all the agents that are at node $v$ in round $t$ belongs to $G_{(v,t)}$, and then all agents of $G_{(v,t)}$ start together phase $k+1$ by assigning to the variable $c$ the number of agents in $G_{(v,t)}$. In each round when the agents of $G_{(v,t)}$ are together in phase $k+1$, they can then detect whether there is a co-located agent that does not belong to $G_{(v,t)}$ by comparing the variable $c$ to ${\tt CurCard}$. In the rest of this proof, when we speak of a group, we always mean a group $G_{(v,t)}$ such that $(v,t)\in\Phi_{k+1}$.

				\begin{claim}
					\label{cP4}
					Assume that property $\Gamma(k+1)$ is not satisfied. Then property $P_2(k+1)$, as well as property $P_4(k+1)$ or property $P_5(k+1)$ hold.
				\end{claim}

				\begin{proofclaim}
					Consider any agent $A$ and denote by $v$ the node from which agent $A$ starts phase $k+1$. The group $G_{(v,t_{A,k+1})}$ to which agent $A$ belongs is shortly denoted by $G$. Given an agent labeled $L$, we will denote by $x_L$ the binary representation of its label, and the binary string ${\tt code}(x_L)$ will be called the code of this agent.

					According to Claim~\ref{cP1}, we have $|t_{X,k+1}-t_{Y,k+1}|\leq\mathcal{D}_{k+1}$ for every couple of agents $(X,Y)$, and, since $\Gamma(k+1)$ is not satisfied, the condition of line~\ref{line:firstc} evaluates to false in all the execution of phase $k+1$. Hence, in view of the second property of Lemma~\ref{lem:invisible}, we know that there is only agent $A$ in $G$, or during each call to procedure ${\tt EXPLO}(N)$ made by any agent $B$ of $G$ (including agent $A$) when executing function ${\tt Communicate}$ in phase $i$, there is a round when agent $B$ is at a node $u\ne v$ with no agent that does not belong to $G$. As the condition of line~\ref{line:firstc} evaluates to false in all the executions of phase $k+1$ by the agents of $G$, we also know in view of lines~\ref{FirstWait} to~\ref{calltocom} that each agent of $G$ start executing function ${\tt Communicate}(k+1,{\tt code}(x_L),true)$ (where $L$ is its label) from node $v$ in round $t_{A,k+1}+\mathcal{D}_{k+1}+3{\tt T}({\tt EXPLO}(N))$. As a result, Lemma~\ref{lem:com} implies that the execution of function ${\tt Communicate}$ by each agent of $G$ is completed at node $v$, in round $t_{A,k+1}+\mathcal{D}_{k+1}+3{\tt T}({\tt EXPLO}(N))+5(k+1){\tt T}({\tt EXPLO}(N))$, and its return value is a couple $(l,*)$,  where $l$ is as follows. Denote by $\overline{G}$ the set of all agents for which  the length of the code is at most $k+1$. If $\overline{G}\ne\emptyset$, then $l={\tt code}(x)1^{k+1-|{\tt code}(x)|}$ where ${\tt code}(x)$ is the lexicographically smallest code of an agent belonging to $\overline{G}$. Otherwise $l=1^{k+1}$.

					Let us first consider the case where $k+1<2\ell+2$. In view of the definition of function ${\tt code}$ we know that $\overline{G}=\emptyset$, and thus $l=1^{k+1}$. Therefore, we know that in the execution of phase $k+1$ by every agent of $G$, the condition of line~\ref{conditionlabel}, and thus the condition of line~\ref{conditiondeclare}, both evaluate to false. It follows that all agents of $G$ are always together from round $t_{A,k+1}+\mathcal{D}_{k+1}+3{\tt T}({\tt EXPLO}(N))+5(k+1){\tt T}({\tt EXPLO}(N))$ till the round $t_{A,k+1}+\mathcal{D}_{k+2}+2\mathcal{D}_{k+1}+(5(k+1)+6){\tt T}({\tt EXPLO}(N))$ in which they start phase $k+2$. Hence, each agent starts executing phase $k+2$ after having spent exactly $\mathcal{D}_{k+2}+2\mathcal{D}_{k+1}+(5(k+1)+6){\tt T}({\tt EXPLO}(N))$ rounds in phase $k+1$. This proves property $P_5(k+1)$.

					To prove property $P_2(k+1)$, it is just enough to show there is no agent outside of $G$ at the same node $w$ as agent $A$ in round $t_{A,k+2}$. Suppose by contradiction that such an agent exists and call it $X$. Denote by $G_X$ the group to which agent $X$ belongs at the beginning of phase $k+1$. Using the same arguments as those used above for group $G$, we know that all agents of $G_X$ are always together from round $t_{X,k+1}+\mathcal{D}_{k+1}+4{\tt T}({\tt EXPLO}(N))+5(k+1){\tt T}({\tt EXPLO}(N))$ till the round $t_{X,k+1}+\mathcal{D}_{k+2}+2\mathcal{D}_{k+1}+(5(k+1)+9){\tt T}({\tt EXPLO}(N))$ in which they start phase $k+2$. In particular, similarly as the agents of $G$, the agents of $G_X$ execute together the block made of the lines~\ref{BC1} to~\ref{BC2}. According to the first property of Lemma~\ref{lem:invisible}, and in view of property $P_5(k+1)$ and of the assumption that property $\Gamma(k+1)$ is not satisfied, we know that $|t_{X,k+2}-t_{A,k+2}|\leq\frac{{\tt T}({\tt EXPLO}(N))}{2}$. This implies that agent $X$ cannot be at node $w$ in round $t_{A,k+2}$ after a move of phase $k+2$, due to the fact that in phase $k+2$ the first move is preceded by a waiting period of length larger than $\frac{{\tt T}({\tt EXPLO}(N))}{2}$. Hence agent $X$, as well as $A$, has been at node $w$ since its execution of line~\ref{secondEXB} in phase $k+1$ is completed. If the execution of line~\ref{secondEXB} in phase $k+1$ is completed in the same round by agents $A$ and $X$, then the agents of $G_X$ and $G$ are at node $w$ when they evaluate the condition of line~\ref{line:secondc}. Hence, this condition evaluates to true for all these agents in their execution of phase $k+1$: this contradicts the assumption that property $\Gamma(k+1)$ is not satisfied. Otherwise, in view of the delay $|t_{X,k+2}-t_{A,k+2}|\leq\frac{{\tt T}({\tt EXPLO}(N))}{2}$, one of the two groups is executing the waiting period of $\mathcal{D}_{k+2}$ rounds at line~\ref{lastwaiting} at node $w$, when the execution of line~\ref{secondEXB} is completed by the other group at node $w$. This implies that $G$ and $G_X$ are together when the agents of one of these groups evaluate the condition of line~\ref{line:secondc}. Then, we again get a contradiction with the assumption that property $\Gamma(k+1)$ is not satisfied, which proves property $P_2(k+1)$.

					Let us now consider the complementary case where $k+1\geq2\ell+2$. Denote by $\gamma$ the label of the agent having the lexicographically smallest code among the agents of the team having a code of length at most $k+1$. Note that, in the considered case, $\gamma$ necessarily exists in view of the definition of function ${\tt code}$.

					If $|\Phi_{k+1}|=1$ then all the agents belong to $G$, and $\overline{G}\ne\emptyset$. This means that the first element $l$ of the return value of function ${\tt Communicate}$ in phase $k+1$ is ${\tt code}(x_{\gamma})1^{k+1-|{\tt code}(x_{\gamma})|}$ for all agents in the team. In view of Proposition~\ref{pro:code}, in every execution of phase $k+1$, the condition of line~\ref{conditionlabel} evaluates to true and the variable $\lambda$ is assigned the label $\gamma$. From this and the fact that property $\Gamma(k+1)$ is not satisfied, the condition of line~\ref{conditiondeclare} evaluates to true in each of the executions of phase $k+1$. It follows that all agents of $G$ (and thus all agents in the graph) are always together from round $t_{A,k+1}+\mathcal{D}_{k+1}+3{\tt T}({\tt EXPLO}(N))+5(k+1){\tt T}({\tt EXPLO}(N))$ till the round $t_{A,k+1}+\mathcal{D}_{k+2}+2\mathcal{D}_{k+1}+(5(k+1)+6){\tt T}({\tt EXPLO}(N))$ in which they declare that gathering is achieved (cf. lines~\ref{debsecondb} to~\ref{finsecondb}). Property $P_4(k+1)$ is then true, and so is property $P_2(k+1)$ because no agent starts phase $k+2$.

					If $|\Phi_{k+1}|\geq2$ then let us choose an agent $B\ne A$ respecting the following conditions: $(1)$ the group $G_{(v',t_{B,k+1})}$ to which agent $B$ belongs is such that $(v,t_{A,k+1})\ne(v',t_{B,k+1})$ and $(2)$ if the agent labeled $\gamma$ is not in $G$ then $B$ has label $\gamma$. The set $G_{(v',t_{B,k+1})}$ is shortly denoted by $G'$. Via similar arguments to those used before, we can show that the agents of $G$ (resp. $G'$) start executing procedure $TZ(l)$ (resp. $TZ(l')$) from node $v$ in round $t_{A,k+1}+\mathcal{D}_{k+1}+(5(k+1)+4){\tt T}({\tt EXPLO}(N))$ (resp. from node $v'$ in round $t_{B,k+1}+\mathcal{D}_{k+1}+(5(k+1)+4){\tt T}({\tt EXPLO}(N))$) for some integers $l$ and $l'$. Actually, due to the definition of $B$, Proposition~\ref{pro:code}, as well as lines~\ref{conditionlabel} and~\ref{getlabel}, we know that $l\ne l'$ and $\gamma\in\{l,l'\}$. Both these executions of procedure $TZ$ are not interrupted prematurely because property $\Gamma(k+1)$ is not satisfied, and thus each of them lasts exactly $\mathcal{D}_{k+1}\geq\mathcal{P}(N,k+1)+\frac{{\tt T}({\tt EXPLO}(N))}{2}$ rounds. By Claim~\ref{cP1}, we know that the delay between the starting times of phase $k+1$ for any two agents  is at most $\mathcal{D}_{k+1}$. Hence, in view of the definition of $G$ (resp. $G'$) and the fact that property $\Gamma(k+1)$ is not satisfied, we can apply the first property of Lemma~\ref{lem:invisible} to state that $|t_{A,k+1}-t_{B,k+1}|\leq\frac{{\tt T}({\tt EXPLO}(N))}{2}$. Therefore, there exists a time interval $\mathcal{T}$ of at least $\mathcal{P}(N,k+1)$ rounds during which the agents of $G$ and the agents of $G'$ execute line~\ref{alg:TZ}. In particular, during interval $\mathcal{T}$ the agents of one group execute together procedure $TZ(l)$, while the agents of the other group execute together procedure $TZ(l')$. Since $\gamma\in\{l,l'\}$ and the length of the binary representation of $\gamma$ is at most $k+1$, the agents of $G$ and $G'$ share the same node in some round during interval $\mathcal{T}$, due to the properties of procedure $TZ$ (cf. Section~\ref{prelim}). However this is a contradiction with the fact that property $\Gamma(k+1)$ is not satisfied.

					To sum up, we have proved that if property $\Gamma(k+1)$ is not satisfied, then property $P_2(k+1)$ is true and so is property $P_4(k+1)$ or $P_5(k+1)$. This concludes the proof of this claim.
				\end{proofclaim}

				The following claim is complementary to the previous one.

				\begin{claim}
					\label{cP5}
					Assume that property $\Gamma(k+1)$ is satisfied. Then properties $P_2(k+1)$ and $P_3(k+1)$ hold.
				\end{claim}

				\begin{proofclaim}
					We first consider the situation where there exists at least one execution of phase $k+1$ by some agent in which the condition of line~\ref{line:firstc} evaluates to true (the complementary situation is adressed thereafter). Let $(u,r)$ be the couple (or one of the couples) of $\Phi(k+1)$ such that for all $(u',r')$ of $\Phi(k+1)$, $r'\geq r$. In view of the waiting periods of $\mathcal{D}_{k+1}$ rounds at the beginning of phase $k+1$ and at the end of phase $k$, as well as the properties $P_1(k+1)$ and $P_2(k)$, no group can meet an agent not belonging to the group from round $r$ to round $r+\mathcal{D}_{k+1}$. Moreover, every group starts phase $k+1$ by round $r+\mathcal{D}_{k+1}$ and cannot start executing line~\ref{secondEX} before round $r+\mathcal{D}_{k+1}+2{\tt T}({\tt EXPLO}(N))$. We have two cases to analyse: the first case is when the first meeting between two groups $G_1$ and $G_2$ occurs in a round of $\{r+\mathcal{D}_{k+1},\ldots,r+\mathcal{D}_{k+1}+2{\tt T}({\tt EXPLO}(N))\}$ and the second case is when it does not.

					Let us start with the first case. When $G_1$ and $G_2$ meet, they then wait at least $\mathcal{D}_{k+2}>\mathcal{D}_{k+1}+3{\tt T}({\tt EXPLO}(N))$ rounds (cf. line~\ref{waitaftermeet}). This means that groups $G_1$ and $G_2$ wait in each round of the set  $\{r+\mathcal{D}_{k+1}+2{\tt T}({\tt EXPLO}(N)), \ldots, r+\mathcal{D}_{k+1}+\mathcal{D}_{k+2}\}$. Moreover, in view of Claim~\ref{cP1}, every group which has not yet encountered another group, starts executing line~\ref{secondEX} in a round of $\{r+\mathcal{D}_{k+1}+2{\tt T}({\tt EXPLO}(N)), \ldots, r+2\mathcal{D}_{k+1}+2{\tt T}({\tt EXPLO}(N))\}$. Hence, every group which has not yet encountered another group, visits the whole graph between rounds $r+\mathcal{D}_{k+1}+2{\tt T}({\tt EXPLO}(N))$ and $r+2\mathcal{D}_{k+1}+3{\tt T}({\tt EXPLO}(N))$ \ie while $G_1$ and $G_2$ wait, and thus, unless it meets another group before, it encounters $G_1$ and $G_2$. By line~\ref{waitaftermeet}, we can state that when any two groups meet, they then have to wait until they see together $\mathcal{D}_{k+2}$ consecutive rounds without any variation of ${\tt CurCard}$ since its latest change. This implies that the execution of a waiting period of $\mathcal{D}_{k+2}$ rounds without any variation of ${\tt CurCard}$ is jointly completed in a round of $\{r+\mathcal{D}_{k+1}+\mathcal{D}_{k+2},\ldots,r+\mathcal{D}_{k+2}+2\mathcal{D}_{k+1}+3{\tt T}({\tt EXPLO}(N))\}$ by any two groups that have met during the interval $\{r+\mathcal{D}_{k+1},\ldots,r+2\mathcal{D}_{k+1}+3{\tt T}({\tt EXPLO}(N))\}$ (every group meets another one in this interval). Thus, by line~\ref{lastwaiting}, any two groups that have met during the interval $\{r+\mathcal{D}_{k+1},\ldots,r+2\mathcal{D}_{k+1}+3{\tt T}({\tt EXPLO}(N))\}$ start together phase $k+2$ in a round of $\{r+\mathcal{D}_{k+1}+2\mathcal{D}_{k+2},\ldots,r+2\mathcal{D}_{k+2}+2\mathcal{D}_{k+1}+3{\tt T}({\tt EXPLO}(N))\}$. This proves property $P_3(k+1)$.

					To prove property $P_2(k+1)$, it is enough to consider any agent $X$ that is at the node $w$ occupied by an agent $A\ne X$ in round $t_{A,k+2}$ and to prove that $t_{A,k+2}=t_{X,k+2}$. If $X$ and $A$ belong to the same group $G$ at the beginning of phase $k+1$, we immediately have $t_{A,k+2}=t_{X,k+2}$. Thus, assume that the groups to which $A$ and $X$ belong when they start phase $k+1$ are different. Since every agent starts phase $k+2$ in a round of $\{r+\mathcal{D}_{k+1}+2\mathcal{D}_{k+2},\ldots,r+2\mathcal{D}_{k+2}+2\mathcal{D}_{k+1}+3{\tt T}({\tt EXPLO}(N))\}$ (cf. the above paragraph), we know that $|t_{A,k+2}-t_{X,k+2}|<\mathcal{D}_{k+2}$. This means that agent $X$ cannot be at node $w$ in round $t_{A,k+2}$ after a move of phase $k+2$, due to the fact that in phase $k+2$ the first move is preceded by a waiting period of at least $\mathcal{D}_{k+2}$ rounds. Hence agent $X$, as well as agent $A$, have been at node $w$ since it evaluated the condition of line~\ref{line:firstc} in phase $k+1$: whether for $A$ or $X$, this evaluation occurs at the latest in round $r+2\mathcal{D}_{k+1}+3{\tt T}({\tt EXPLO}(N))$. As a result, the group of agent $X$ and the group of agent $A$ meet at node $w$ in a round of the interval $\{r+\mathcal{D}_{k+1},\ldots,r+2\mathcal{D}_{k+1}+3{\tt T}({\tt EXPLO}(N))\}$. We proved in the above paragraph that such a meeting between two groups in this interval implies that the agents of the two groups start together phase $k+2$. Hence $t_{A,k+2}=t_{X,k+2}$, which proves property $P_2(k+1)$.

					Let us now consider the case where no group meets another one in a round of $\{r+\mathcal{D}_{k+1},\ldots,r+\mathcal{D}_{k+1}+2{\tt T}({\tt EXPLO}(N))\}$. This implies that when the effective part of procedure ${\tt EXPLO}(N)$ of line~\ref{firstEXphase} in phase $k+1$ is completed by $G_{(u,r)}$, every agent has at least started the effective part of procedure ${\tt EXPLO}(N)$ of line~\ref{firstEXphase} in phase $k+1$ (otherwise using Claim~\ref{cP1} we can get a contradiction with the fact that no group meets another one during the interval $\{r+\mathcal{D}_{k+1},\ldots,r+\mathcal{D}_{k+1}+2{\tt T}({\tt EXPLO}(N))\}$). This means that the delay between the starting times of phase $k+1$ by any two agents is at most $\frac{{\tt T}({\tt EXPLO}(N))}{2}$. It follows that the first round in which the condition of line~\ref{line:firstc} evaluates to true in the execution of phase $k+1$ (which also corresponds to a round in which two groups $G_1$ and $G_2$ meet) is equal to $r+x$ with $\mathcal{D}_{k+1}+2{\tt T}({\tt EXPLO}(N))+1\leq x\leq\mathcal{D}_{k+1}+\frac{7{\tt T}({\tt EXPLO}(N))}{2}$. In view of the delay between the starting times of phase $k+1$ by any two agents which is at most $\frac{{\tt T}({\tt EXPLO}(N))}{2}$ and of the definition of round $r+x$, we know that in each round of $\{r+\mathcal{D}_{k+1}+2{\tt T}({\tt EXPLO}(N))+1,\ldots,r+x\}$, each agent either executes an instruction of line~\ref{anotherwait} or~\ref{secondEX}, or has just started executing the waiting period of line~\ref{waitaftermeet}, or waits by processing the first waiting period of ${\tt T}({\tt EXPLO}(N))$ rounds in function ${\tt Communicate}$. We also know that the agents that are executing function ${\tt Communicate}$ in some round of $\{r+\mathcal{D}_{k+1}+2{\tt T}({\tt EXPLO}(N))+1,\ldots,r+x\}$ cannot have executed line~\ref{waitaftermeet} in phase $k+1$, or otherwise we get a contradiction with the definition of round $r+x$. Hence, since the last ${\tt T}({\tt EXPLO}(N))$ rounds of line~\ref{FirstWait}, line~\ref{firstEXphase}, and line~\ref{anotherwait} consist of the same instructions as line~\ref{anotherwait}, line~\ref{secondEX}, and either the first ${\tt T}({\tt EXPLO}(N))$ rounds of line~\ref{waitaftermeet} or the first waiting period of ${\tt T}({\tt EXPLO}(N))$ rounds in function ${\tt Communicate}$, respectively, the groups $G_1$ and $G_2$ that are at the same node $w$ in round $r+x$, are also at node $w$ in round $r+x-2{\tt T}({\tt EXPLO}(N))$. The condition of line~\ref{line:firstc} then evaluates to true in round $r+x-2{\tt T}({\tt EXPLO}(N))$ in the execution of phase $k+1$ by every agent of $G_1$ or $G_2$. We then get a contradiction with the definition of round $r+x$.

					As a result, in the situation where there exists at least one execution of phase $k+1$ by some agent in which the condition of line~\ref{line:firstc} evaluates to true, we know that properties $P_3(k+1)$ and $P_2(k+1)$ are satisfied. Let us now consider the complementary situation  assuming that property $\Gamma(k+1)$ is satisfied, namely: there exists at least one execution of phase $k+1$ by some agent in which the condition of line~\ref{line:secondc} evaluates to true, and the condition of line~\ref{line:firstc} evaluates to false in the execution of phase $k+1$ by every agent.

					By Lemma~\ref{lem:invisible} and Claim~\ref{cP1}, we know that the maximum delay in phase $k+1$ between the rounds in which there is an execution of function ${\tt Communicate}$ that is completed is at most $\frac{{\tt T}({\tt EXPLO}(N))}{2}$ rounds. We also know that in each group $G_{(v,t)}$ the execution of function ${\tt Communicate}$ is completed by the agents of $G_{(v,t)}$ in the same round at node $v$, and thus the agents of $G_{(v,t)}$ start together the next instruction of phase $k+1$ from node $v$. Note that in the round when the execution of ${\tt Communicate}$ is completed by the agents of $G_{(v,t)}$, call it $z$, all the agents that are at node $v$ belong to $G_{(v,t)}$. Indeed, if this were not the case, in view of the aforementioned delay of $\frac{{\tt T}({\tt EXPLO}(N))}{2}$ rounds and the fact that the end of the execution of function ${\tt Communicate}$ by each group in phase $k+1$ is directly preceded and followed by waiting periods of at least ${\tt T}({\tt EXPLO}(N))$ rounds, we would have the following: all the agents that are at node $v$ in round $z$, are also at node $v$ in round $z-5(k+1){\tt T}({\tt EXPLO}(N))-\frac{3{\tt T}({\tt EXPLO}(N))}{2}$ in which the agents of $G_{(v,t)}$ start the second half of the waiting period of length ${\tt T}({\tt EXPLO}(N))$ at line~\ref{anotherwait}. This would contradict the assumption that the condition of line~\ref{line:firstc} evaluates to false in the execution of phase $k+1$ by every agent. Let $r'=r+\mathcal{D}_{k+1}+(5(k+1)+3){\tt T}({\tt EXPLO}(N))$: round $r'$ is the one in which the execution of function ${\tt Communicate}$ in phase $k+1$ is completed by the agents of $G_{(u,r)}$ at node $u$. In view of the above explanations, we know that no group meets another one from round $r'$ to round $r'+{\tt T}({\tt EXPLO}(N))$. Moreover, every group starts executing the second begin-end block of phase $k+1$ by round $r'+{\tt T}({\tt EXPLO}(N))$ and cannot have started executing line~\ref{secondEXB} before round $r'+\mathcal{D}_{k+1}+2{\tt T}({\tt EXPLO}(N))$. From there on, the proof is similar to that for the first situation. However, for completeness, we state it in detail.

					Let us first assume that a meeting between two groups occurs in some round of $\{r'+{\tt T}({\tt EXPLO}(N)), \ldots, r'+\mathcal{D}_{k+1}+2{\tt T}({\tt EXPLO}(N))\}$, and denote by $G_1$ and $G_2$ the two groups involved in the first meeting in this interval. When $G_1$ and $G_2$ meet, they then wait at least $\mathcal{D}_{k+2}>\mathcal{D}_{k+1}+3{\tt T}({\tt EXPLO}(N))$ rounds (cf. line~\ref{waitaftermeet2}). This means that groups $G_1$ and $G_2$ wait in each round of $\{r'+\mathcal{D}_{k+1}+2{\tt T}({\tt EXPLO}(N)), \ldots, r'+\mathcal{D}_{k+2}+{\tt T}({\tt EXPLO}(N))\}$. Moreover, every group which has not yet encountered another one, starts executing line~\ref{secondEXB} in a round of $\{r'+\mathcal{D}_{k+1}+2{\tt T}({\tt EXPLO}(N)), \ldots, r'+\mathcal{D}_{k+1}+\frac{5{\tt T}({\tt EXPLO}(N))}{2}\}$, in view of the maximum delay of  $\frac{{\tt T}({\tt EXPLO}(N))}{2}$ rounds in phase $k+1$ between the rounds in which there is an execution of function ${\tt Communicate}$ that is completed. Hence, every group which has not yet encountered another one, executes the effective part of line~\ref{secondEXB} and thus visits the whole graph between rounds $r'+\mathcal{D}_{k+1}+2{\tt T}({\tt EXPLO}(N))$ and $r'+\mathcal{D}_{k+1}+3{\tt T}({\tt EXPLO}(N))$ \ie while $G_1$ and $G_2$ wait and thus, unless it meets another group before, it encounters $G_1$ and $G_2$. By line~\ref{waitaftermeet2}, we can state that when any two groups meet, they then have to wait until they see together $\mathcal{D}_{k+2}$ consecutive rounds without any variation of ${\tt CurCard}$ since its latest change. This implies that the execution of a waiting period of length $\mathcal{D}_{k+2}$ without any variation of ${\tt CurCard}$ is jointly completed in a round of $\{r'+\mathcal{D}_{k+2}+{\tt T}({\tt EXPLO}(N)), \ldots, r'+\mathcal{D}_{k+1}+\mathcal{D}_{k+2}+3{\tt T}({\tt EXPLO}(N))\}$ by any two groups that have met during the interval $\{r'+{\tt T}({\tt EXPLO}(N)), \ldots, r'+\mathcal{D}_{k+1}+3{\tt T}({\tt EXPLO}(N))\}$ (every group meets another one in this interval). Thus, by line~\ref{lastwaiting}, any two groups that have met during the interval $\{r'+{\tt T}({\tt EXPLO}(N)), \ldots, r'+\mathcal{D}_{k+1}+3{\tt T}({\tt EXPLO}(N))\}$ start together phase $k+2$ in a round of $\{r+\mathcal{D}_{k+1}+2\mathcal{D}_{k+2}+(5(k+1)+4){\tt T}({\tt EXPLO}(N)), \ldots, r+2\mathcal{D}_{k+1}+2\mathcal{D}_{k+2}+(5(k+1)+6){\tt T}({\tt EXPLO}(N))\}$. This proves property $P_3(k+1)$.

					To prove property $P_2(k+1)$, it is enough to consider any agent $X$ that is at the node $w$ occupied by an agent $A\ne X$ in round $t_{A,k+2}$ and to prove that $t_{A,k+2}=t_{X,k+2}$. If $X$ and $A$ belong to the same group $G$ at the beginning of phase $k+1$, we immediately have $t_{A,k+2}=t_{X,k+2}$. Thus, assume that the groups to which $A$ and $X$ belong when they start phase $k+1$ are different. Since every agent starts phase $k+2$ in a round of $\{r+\mathcal{D}_{k+1}+2\mathcal{D}_{k+2}+(5(k+1)+4){\tt T}({\tt EXPLO}(N)), \ldots, r+2\mathcal{D}_{k+1}+2\mathcal{D}_{k+2}+(5(k+1)+6){\tt T}({\tt EXPLO}(N))\}$ (cf. the above paragraph), we know that $|t_{A,k+2}-t_{X,k+2}|<\mathcal{D}_{k+2}$. This means that agent $X$ cannot be at node $w$ in round $t_{A,k+2}$ after a move of phase $k+2$, due to the fact that in phase $k+2$ the first move is preceded by a waiting period of at least $\mathcal{D}_{k+2}$ rounds. Hence agent $X$, as well as agent $A$, has been at node $w$ since it evaluated the condition of line~\ref{line:secondc} in phase $k+1$: whether for $A$ or $X$, this evaluation occurs at the latest in round $r'+\mathcal{D}_{k+1}+3{\tt T}({\tt EXPLO}(N))$. As a result, the group of agent $X$ and the group of agent $A$ meet at node $w$ in a round of the interval $\{r'+{\tt T}({\tt EXPLO}(N)), \ldots, r'+\mathcal{D}_{k+1}+3{\tt T}({\tt EXPLO}(N))\}$. We proved in the above paragraph that such a meeting between two groups in this interval implies that the agents of the two groups start together phase $k+2$. Hence $t_{A,k+2}=t_{X,k+2}$, which proves property $P_2(k+1)$.

					Consider the case where no group meets another one in a round of $\{r'+{\tt T}({\tt EXPLO}(N)), \ldots, r'+\mathcal{D}_{k+1}+2{\tt T}({\tt EXPLO}(N))\}$. In this case, the first round in which the condition of line~\ref{line:secondc} evaluates to true in the execution of phase $k+1$ by some agent is at least $r'+\mathcal{D}_{k+1}+2{\tt T}({\tt EXPLO}(N))+1$. The first round in which the condition of line~\ref{line:secondc} evaluates to true in the execution of phase $k+1$ by some agent is also at most $r'+\mathcal{D}_{k+1}+\frac{7{\tt T}({\tt EXPLO}(N))}{2}$ in view of the maximum delay of $\frac{{\tt T}({\tt EXPLO}(N))}{2}$ rounds in phase $k+1$ between the rounds in which there is an execution of function ${\tt Communicate}$ that is completed. Consequently, the first round in which the condition of line~\ref{line:secondc} evaluates to true in the execution of phase $k+1$ by some agent (which also corresponds to a round in which two groups $G_1$ and $G_2$ meet) is equal to  $r'+x$ with $\mathcal{D}_{k+1}+2{\tt T}({\tt EXPLO}(N))+1\leq x\leq\mathcal{D}_{k+1}+\frac{7{\tt T}({\tt EXPLO}(N))}{2}$. In view of the aforementioned delay and of  the definition of round $r'+x$, we know that in each round of $\{r'+\mathcal{D}_{k+1}+2{\tt T}({\tt EXPLO}(N))+1, \ldots, r'+x\}$, each agent either executes an instruction of line~\ref{anotherwaitB} or~\ref{secondEXB}, or has just started the waiting period period of line~\ref{waitaftermeet2}, or waits by processing the first ${\tt T}({\tt EXPLO}(N))$ rounds of the waiting period at line~\ref{lastwaiting}. We also know that the agents that are processing line~\ref{lastwaiting} in some round of $\{r'+\mathcal{D}_{k+1}+2{\tt T}({\tt EXPLO}(N))+1, \ldots, r'+x\}$ cannot have executed line~\ref{waitaftermeet2} in phase $k+1$, or otherwise we get a contradiction with the definition of round $r'+x$. Hence, since the last ${\tt T}({\tt EXPLO}(N))$ rounds of line~\ref{FirstWait}, line~\ref{firstEXphase}, and line~\ref{anotherwait} consist of the same instructions as line~\ref{anotherwaitB}, line~\ref{secondEXB}, and the first ${\tt T}({\tt EXPLO}(N))$ rounds of either line~\ref{waitaftermeet2} or line~\ref{lastwaiting}, respectively, the groups $G_1$ and $G_2$ that are at the same node $w$ in round $r'+x$, are also at node $w$ in round $r'+x-\mathcal{D}_{k+1}-5(k+2){\tt T}({\tt EXPLO}(N))$. The condition of line~\ref{line:firstc} then evaluates to true in round $r'+x-\mathcal{D}_{k+1}-5(k+2){\tt T}({\tt EXPLO}(N))$ in the execution of phase $k+1$ by every agent of $G_1$ or $G_2$, which is a contradiction.

					To summarize, we proved that properties $P_2(k+1)$ and $P_3(k+1)$ always hold, which concludes the proof of this claim.
				\end{proofclaim}

				According to Claims~\ref{cP1},~\ref{cP4} and~\ref{cP5}, we know that properties $P_1(k+1)$ and $P_2(k+1)$ hold. We also know that at least one property among $P_3(k+1)$, $P_4(k+1)$ and $P_5(k+1)$ holds. This concludes the inductive proof of the lemma.
			\end{proof}

			\begin{theorem}
				\label{theo:known}
				Algorithm ${\tt GatherKnownUpperBound}$ solves the gathering problem and the leader election problem, and has a time complexity that is polynomial in the known upper bound $N$ on the size of the network and in the length $\ell$ of the smallest label among the agents.
			\end{theorem}

			\begin{proof}
				Suppose that an agent, call it $A$, declares that gathering is achieved in some round $\tau$ when processing some phase $i$ and suppose, without loss of generality, $i$ is the smallest integer for which this occurs. According to Algorithm~\ref{alg:GKUB}, the execution of ${\tt GatherKnownUpperBound}$ is completed by agent $A$ in round $\tau$, and the agent will not start phase $i+1$. Hence, in view of the minimality of $i$, property $P_4(i)$ of Lemma~\ref{lem:crucial} holds. It follows that the gathering problem and the leader election problem are solved in round $\tau$ in which the execution of ${\tt GatherKnownUpperBound}$ by each agent is completed.
				It also follows that $\tau=t_{F,i}+\mathcal{D}_{i+1}+2\mathcal{D}_{i}+(5i+6){\tt T}({\tt EXPLO}(N))$ where $F$ is the earliest agent (or one of the earliest agents) to start phase $i$. By properties $P_3(k)$ and $P_5(k)$ of Lemma~\ref{lem:crucial}, agent $F$ cannot have spent more than $4\mathcal{D}_{k+1}+(5k+6){\tt T}({\tt EXPLO}(N))$ rounds executing any phase $k<i$. By property $P_1(0)$ of Lemma~\ref{lem:crucial}, $|t_{F,0}-t_{F',0}|<\mathcal{D}_0$ where $F'$ is the earliest agent (or one of the earliest agents) to start executing ${\tt GatherKnownUpperBound}$ and thus to start executing phase $0$. Therefore, if $i\leq \lfloor \log N \rfloor +2\ell+2$, then $|\tau-t_{F',0}|\leq (\lfloor \log N \rfloor +2\ell+4)(4\mathcal{D}_{\lfloor \log N \rfloor +2\ell+3}+(5(\lfloor \log N \rfloor +2\ell+2)+6){\tt T}({\tt EXPLO}(N)))$, which is polynomial in $N$ and $\ell$.

				As a result, to conclude the proof, it is enough to show that $\tau$ exists and $i\leq \lfloor \log N \rfloor +2\ell+2$. Suppose by contradiction, that it is not the case. By Lemma~\ref{lem:crucial}, every agent ends up executing phase $\lfloor \log N \rfloor +2\ell+3$, and thus $|\Phi_{\lfloor \log N \rfloor +2\ell+3}|\geq1$. Still by Lemma~\ref{lem:crucial}, we know that $|\Phi_{\lfloor \log N \rfloor +2\ell+3}|\leq\frac{|\Phi_{2\ell+2}|}{2^{\lfloor \log N \rfloor+1}}$. However, $|\Phi_{2\ell+2}|\leq N$ and $2^{\lfloor \log N \rfloor+1}>N$. Hence, $1\leq|\Phi_{\lfloor \log N \rfloor +2\ell+3}|<1$, which is a contradiction and proves the theorem.
			\end{proof}

	\section{Unknown upper bound on the size of the graph}

		This section is dedicated to the presentation and the analysis of our algorithm ${\tt GatherUnknownUpperBound}$ that allows the agents to solve gathering without direct means of communication in the case where they are not initially given any upper bound on the graph size.

		\subsection{Intuition} \label{subsec:intuition}

			A preliminary question that may come to mind when considering this harsher scenario, is how to guarantee gathering even if in each round an agent had the capacity to exchange all information available to it with the other agents sharing its node. Actually, this can be done by implementing a general mechanism explained below that stems from the following simple observation. If the agents were woken up at the same time by the adversary and were initially given the description of the initial configuration $\phi$ (i.e., the complete map of the underlying graph, with all port numbers, in which a node $v$ is labeled $L$ iff $v$ is the starting node of the agent labeled $L$), then the problem could be solved by applying a simple rule: upon its wake-up, each agent moves, using the map, towards the node containing the smallest label, and then declares that gathering is over when all agents are in that node.

			In an attempt to recreate similar favorable conditions to those given in the above observation, a general mechanism to solve gathering with communication has been described in \cite{BDD}: as it will be useful for our purpose, let us briefly explain how it works at a high level. Let $\Omega=(\phi_1,\phi_2,\phi_3,\cdots)$ be a fixed enumeration of the (recursively enumerable) set of all initial configurations. An agent proceeds in consecutive phases $i=1,2,3,\cdots$, where in each phase $i$, an agent tries to achieve gathering by making a hypothesis that the initial configuration $\phi$ is $\phi_i$.
			This assumption will be called hypothesis $i$. When making hypothesis $i$, the agent first executes ${\tt EXPLO}(n_i)$ where $n_i$ is the supposed graph size (in the hope of waking up the remaining dormant agents, if any), and then, if its label belongs to $\phi_i$, tries to go to the node that supposedly corresponds to the starting node $v_i$ of the agent having the smallest label in $\rho_i$. Once an agent reaches a node that supposedly corresponds to $v_i$, it waits some amount of time sufficient to allow the other agents to join it if the hypothesis $i$ is good (i.e., correct). Note that in the case where the hypothesis $i$ is not good, some agents may notice it more or less rapidly, for instance if $\phi_i$ does not contain their label or if the path they have to follow to reach $v_i$ simply does not exist in the real network. Hence, the process of a phase $i$ we have described so far can lead to one of the following three situations: (1) the hypothesis $i$ is good and all agents think rightly that gathering is over, (2) the hypothesis $i$ is not good and all agents know it, or (3) the hypothesis $i$ is not good but some agents do not know it because for some of them everything appears to have gone very well (in particular they are at the node supposedly corresponding to $v_i$ with agents having the labels they should have w.r.t $\phi_i$). The third situation may especially occur if the number of supposed agents or the supposed size of the graph in $\phi_i$ is lower than it actually is. At this point, an ``optimistic'' agent may think that it is in the first situation, while really it is in the third. This is why each time a phase is completed, the optimistic agents, if any, execute a checking protocol based on a simulation of the exploration protocol ${\tt EST}$, in which the role of the token is played by some of them. If after having executed the checking protocol, the optimistic agents notice they have constructed a map of the graph corresponding to that of $\phi_i$ and they have not encountered agents thinking that hypothesis $i$ is wrong (before switching to the next phase, these agents stay idle enough time in order to be detected by possible optimistic agents), they can be sure to have accomplished gathering. Otherwise, the agents that were optimistic join the camp of the agents knowing that hypothesis $i$ is wrong, and after a certain time, each agent goes back to its initial node in order to start phase $i+1$.

			The high level idea of our algorithm consists in emulating this general mechanism in a context in which the agents are devoid of direct means of communication. However, this turns out to be much easier said than done. Indeed, to settle properly such an emulation, we have to face numerous challenges. For instance, how can an agent become optimistic about a supposed configuration $\phi_i$, if it cannot see the labels of the agents sharing its node? Or, how can an agent recognize its token played by some agents during an execution of the checking protocol? In the next subsection, we bring algorithmic solutions to solve these problems. However, one particular problem was really more important than the others, in the sense that its resolution appeared as a \emph{sine qua non} condition to solve most of the other problems raised by our emulation. The problem is this: in order to emulate the mechanism outlined above in a correct way, an agent has to avoid confusing another agent acting under the same hypothesis as itself, with an agent acting under a different hypothesis. While this is not at all an issue when an agent can exchange arbitrary messages with the other agents sharing its current node (indeed, it is then enough to ignore the agents that indicate processing a different hypothesis), this becomes a real difficulty in our context. To tackle it, we put in place two fundamental schemes.

			The first scheme consists in slowing down every agent more and more as it progresses through successive hypotheses, so that an agent processing a hypothesis $h$ avoids being mislead by an agent processing a hypothesis $h'>h$. Actually, using various strategies of moves such as, for example, some ``dancing'' protocols (that consist in leaving and entering the same node in carefully chosen rounds), and by interleaving, between each step of every hypothesis $h>1$, waiting periods lasting the time of processing the hypotheses $1$ to $h-1$ in the worst case, every agent can detect an agent that processes a hypothesis larger than its own current hypothesis. However, this turns out to be ineffective to protect an agent from being confused by agents processing smaller hypotheses. This is why we introduce a second scheme that divides the processing of each hypothesis into two parts: the main part corresponding, strictly speaking, to the emulation of the general mechanism with, in particular, the previously mentioned slowdowns and dances, that is preceded by a preprocessing part, the heart of the second scheme, and which relies on the notions of \emph{kernel} and \emph{ball} defined below.

			A kernel $\mathcal{K}(u,h)$ is the set of all nodes that may be visited by an agent processing the main part of hypothesis $h$, starting from a node $u$. In our solution, the set $\mathcal{K}(u,h)$ does not depend on the label of the executing agent: in fact, the label can influence the way the nodes of $\mathcal{K}(u,h)$ are visited, but not the set of visited nodes. Given a kernel $\mathcal{K}$, its ball is the set of all its critical nodes. A node $v$ is said to be critical for a kernel $\mathcal{K}$ if $v$ belongs to $\mathcal{K}$, or if an agent, when initially located at node $v$, may visit a node of $\mathcal{K}$ before processing hypothesis $h$, either during a preprocessing part or during a main part of some hypothesis $h'<h$. These notions are used within the second scheme, in the following way. An agent executing the preprocessing part of a hypothesis $h$ performs the following actions, without paying attention to other agents: first, it visits all nodes of the ball of $\mathcal{K}(u,h)$, where $u$ is the node it occupies at the beginning of the preprocessing and of the main part of hypothesis $h$, then it goes back to $u$, and finally it waits the maximum time required for an agent, which would have just been woken up, to begin the execution of hypothesis $h$ (in our solution the time to reach hypothesis $h$ is bounded by some function depending only on the supposed graph sizes of the previous hypotheses). By doing so, we have the guarantee that an agent $a$ executing the main part of a hypothesis $h$, the kernel of which is $\mathcal{K}$, cannot encounter any agent executing the preprocessing part or the main part of a previous hypothesis. Indeed, otherwise this would imply that an agent, initially located at a node belonging to the ball of $\mathcal{K}$, has been woken up after the traversal by agent $a$ of the ball of $\mathcal{K}$, which would be a contradiction. Note that, visiting all nodes of the ball of $\mathcal{K}(u,h)$ can be made by following all the (not necessary simple) paths of length $d_1+d_2$ from node $u$, where $d_1$ (resp. $d_2$) is the maximum distance between two nodes visited during the main part of hypothesis $h$ (resp. during the processing of the hypotheses $1$ to $h-1$). Also note that the executing agent will abort the traversal of the ball of $\mathcal{K}(u,h)$ as soon as it occupies a node having a degree at least equal to the graph size of $\phi_h$. This precaution does not cause any problem because if such an event occurs the agent has the guarantee that $\phi_h\ne\phi$, and before starting phase $h+1$, it no longer has to worry about agents processing hypotheses different from its own.  This is nonetheless crucial as we need the agents to know in advance the worst-case time to process any given hypothesis (to settle consistently some of the aforementioned waiting periods).

			To get all of this to work, we must not forget to link the schemes to each other: in particular, the slowdowns of the first scheme have to be also added between the steps of the preprocessing parts, and every ball traversal of the second scheme has to take into account the dancing protocols of each main part resulting from the first scheme.

			In order to illustrate the importance of the joint application of the two presented schemes, let us close our intuitive explanations by considering one of the problems, mentioned earlier, encountered in putting in place our emulation: that of an exploration, during which an agent plays the role of an explorer while the others play the role of a token, performed to check whether a supposed hypothesis is good or not. Actually, our solution is designed so that for each hypothesis $h$, such explorations can be triggered by at most one group $G_h$ of agents located at the same node $u_h$, those thinking that hypothesis $h$ is good: the role of explorer is assigned in turn to all agents of $G_h$ using the order over the labels that are in $\phi_h$. When an agent becomes explorer, it executes protocol ${\tt EST}$ using as a token the other agents of $G_h$, which then stay idle at node $u_h$.
As soon as the execution of ${\tt EST}$ is completed or as soon as the explorer notices that the map under construction of the network does not match that of $\phi_h$, it goes back to its token to let the other agents, which have not yet done so, execute ${\tt EST}$ and reach the same conclusion about the validity of hypothesis $h$. The key thing for an explorer is to perform a ``clean'' exploration i.e., not to confuse the group of agents representing its token with another group, as otherwise it might reach a wrong conclusion. Such a confusion can be avoided by requiring the agents of $G_h$ to perform twice together a traversal of all the nodes located at distance at most $d+1$ from node $u_h$ before starting the simulations of ${\tt EST}$, where $d$ is the maximal distance that may separate an explorer from its token during a simulation of ${\tt EST}$ from $u_h$. If some agents are encountered during one of these two collective traversals, which can be easily detected through a rise of cardinality, then all the agents of $G_h$ have the guarantee that $\phi_h$ is not good and they do not even need to proceed further with the simulations of ${\tt EST}$. On the other hand, if no other agent is encountered during any of these traversals, then in view of the two schemes, we have the guarantee that each simulation of ${\tt EST}$ by an agent of $G_h$ will be clean. Indeed, the second scheme ensures that an explorer cannot meet an agent processing a hypothesis $h'<h$. Moreover, an explorer cannot meet an agent processing a hypothesis $h'>h$ because the first scheme ensures that such an agent is too slow to reach or to be already in the zone of the nodes at distance at most $d$ from $u_h$ (in which the successive simulations of ${\tt EST}$ are done) without having been detected during one of the previous two traversals. Of course, one might still argue that an explorer could be bothered by agents that also test hypothesis $h$ but that do not belong to $G_h$. However, this kind of situation will never occur. In fact, the agents that do not belong to $G_h$, will quickly notice that the hypothesis $h$ is not good, when processing the main part of hypothesis $h$, and thus by adding judicious waiting periods as we did in our algorithm, we will be able to prove they will be detected by all agents of $G_h$ during one of their two collective traversals.

		\subsection{Algorithm} \label{sec:ninconnu}

			In order to give a clear presentation of our solution, we first need to introduce some notation and definitions.

			In the sequel, we consider an arbitrarily fixed enumeration $\Omega=(\phi_1,\phi_2,\phi_3,\ldots)$ of all initial configurations $\phi_i$ represented as a graph of size at least $2$ with all port numbers, in which there are at least 2 labeled nodes: a node $v$ is labeled $L$ iff we assume that $v$ is the initial node of an agent labeled $L$. For every positive integer $h$, the values $n_h$ and $k_h$ respectively denote the number of nodes and the number of labeled nodes in configuration $\phi_h$. We also denote by $m_h$ the smallest integer such that for all $i \in \{1, \dots, h\}, n_i\leq m_h$.

			Gathering in our model, in the case where the agents are not initially given any upper bound on the graph size, can be done using Algorithm ${\tt GatherUnknownUpperBound}$ described in Algorithm~\ref{alg:mainnoupper}. This protocol mainly consists of successive executions of the function ${\tt Hypothesis}(h)$ whose pseudocode is given in Algorithm~\ref{alg:hyp}. Roughly speaking, when executing ${\tt Hypothesis}(h)$, an agent tries to solve gathering by assuming that configuration $\phi_h$ was the real initial configuration $\phi$. We will show that if there is a round $r$ when this function returns true during the execution by an agent of ${\tt GatherUnknownUpperBound}$, then the gathering is achieved and declared as such by all agents in round $r$: the variable $leader$ (resp. $size$) of each agent is, by then, equal to the smallest agent's label (resp. the size of the graph). We will also show that function ${\tt Hypothesis}$ returns true when the integer $h$ given as input is such that $\phi_h=\phi$.

			\begin{algorithm}
				\begin{footnotesize}
				\caption{Algorithm ${\tt GatherUnknownUpperBound}$\label{alg:mainnoupper}}
				\begin{algorithmic}[1]
					\Begin
						\State $h\gets 0$
						\Repeat
							\State $h\gets h+1$
							\State $b\gets {\tt Hypothesis}(h)$
						\Until{$b={\tt true}$}
						\State $\mathit{leader}\gets$ the smallest node's label in $\phi_h$ \label{line:gath:lead}
						\State $\mathit{size}\gets$ the graph size in $\phi_h$ \label{line:gath:size}
						\State declare the gathering is achieved
					\End
				\end{algorithmic}
				\end{footnotesize}
			\end{algorithm}

			To draw a parallel with the intuitive explanations given in Section~\ref{subsec:intuition}, lines~\ref{alg:prepro1} and~\ref{alg:prepro2} of Algorithm~\ref{alg:hyp} correspond to the preprocessing part of a hypothesis (i.e., the second scheme to protect an agent from confusion by agents processing previous hypotheses) while the lines~\ref{alg:main1} to~\ref{alg:main2} correspond to the main part of a hypothesis. For every positive integer $h$, the values $\mathcal{S}_h$ and $\mathcal{T}_h$ used in Algorithm~\ref{alg:hyp} are defined mutually recursively as follows:

			$\left\{
			\begin{array}{l}
				\mathcal{S}_h= {\tt T}({\tt BallTraversal}(h))+\sum_{i=1}^{h-1}\mathcal{T}_i  \\
				\mathcal{T}_h=8m_h^{2m_h^5}(3\mathcal{S}_h+2{\tt T}({\tt Ball\-Traversal}(h)))
			\end{array}
			\right.$

			The value ${\tt T}({\tt BallTraversal}(h))$ used in the above formulas is defined below. We will show that $\mathcal{S}_h$ (resp. $\mathcal{T}_h$) is an upper bound on the time to execute ${\tt Hypothesis}(1)$ to ${\tt Hypothesis}(h-1)$ plus the function ${\tt BallTraversal}(h)$ at the beginning of ${\tt Hypothesis}(h)$ (resp. an upper bound on the time to execute ${\tt Hypothesis}(h)$).

			\begin{algorithm}
				\begin{footnotesize}
				\caption{Algorithm ${\tt Hypothesis}(h)$\label{alg:hyp}}
				\begin{algorithmic}[1]
					\Begin \label{alg:debut}
						\State /* Beginning of the first part of the function */
						\If{${\tt BallTraversal}(h)$ }\label{alg:prepro1}
							\State wait $\mathcal{S}_h$ rounds \label{alg:prepro2}
							\If{${\tt MoveToCentralNode}(h)$} \label{alg:main1}
								\If{${\tt StarCheck}(h)$} \label{line:hyp:star}
									\If{${\tt EnsureCleanExploration}(h)$}\label{alg:check1}
										\If{${\tt GraphSizeCheck}(h)$}\label{alg:check2}
											\State {\bf return} ${\tt true}$
										\EndIf
									\EndIf
								\EndIf
							\EndIf
						\EndIf \label{alg:fin}
						\State /* Beginning of the second part of the function */
						\State let $e$ be the sequence of ports by which the agent entered during the first part (in the order of entrance); $i\gets |e|$ \label{alg:hyp:lineback1}
						\While{$i\geq 1$}
							\State wait $7m_h^{2m_h^5}$ rounds \label{slowdowns0}
							\State take port $e[i]$ \label{line:hyp:back}
							\State $i\gets i-1$
						\EndWhile \label{alg:hyp:lineback2}
						\State wait until having spent at least $\mathcal{T}_h$ rounds in the current execution of the function \label{alg:hyp:linewait2}
						\State {\bf return} ${\tt false}$\label{alg:main2}
					\End
				\end{algorithmic}
				\end{footnotesize}
			\end{algorithm}

			As the name suggests, the ball traversal of the preprocessing part is based on the execution of function ${\tt BallTraversal}(h)$ detailed in Algorithm~\ref{alg:btraversal}: this function returns false iff during its traversal the agent visits a node having a degree at least equal to the graph size of $\phi_h$. The lines~\ref{slowdown1} and~\ref{slowdown2} of Algorithm~\ref{alg:btraversal} correspond to the slowdowns of our first scheme. For ease of presentation, we explained in the intuition section that we could insert, between all steps of every hypothesis, waiting periods lasting enough time in order to help detect agents processing higher hypotheses. However, as it is the case here, we do not need so many extra waiting periods. Actually, the only other function where such slowdows are used for the same purpose is ${\tt Hypothesis}(h)$ (cf. line~\ref{slowdowns0} Algorithm~\ref{alg:hyp}). The value ${\tt T}({\tt BallTraversal}(h))=64hm_h^{7hm_h^5}$ is an upper bound on the execution time of ${\tt BallTraversal}(h)$.

			\begin{algorithm}
				\begin{footnotesize}
					\caption{Algorithm ${\tt BallTraversal}(h)$\label{alg:btraversal}}
					\begin{algorithmic}[1]
						\Begin
							\ForEach{path $x$ of length $4h m^5_h$ from the set $\{0,\dots,n_h-2\}$}\label{wordprocess}
								\State $b\gets{\tt true}$
								\State $e\gets\epsilon$
								\State $i\gets 1$
								\While{$b={\tt true}$ {\bf and} $i\leq 4hm^5_h$}
									\If{the degree of the current node is at least $n_h$}
										\State {\bf return} ${\tt false}$ \label{alg:ball:retfal}
									\EndIf
									\If{there is no port $x[i]$ at the current node}
										\State $b\gets{\tt false}$
									\Else
										\State wait $7m_h^{2m_h^5}$ rounds\label{slowdown1}
										\State take port $x[i]$ \label{line:btrav:take1}
										\State $e[i]\gets$ the port through which the agent has just entered the current node
										\State $i\gets i+1$
									\EndIf
								\EndWhile
								\While{$i>1$} \label{tr:debwhile}
									\State $i\gets i-1$
									\State wait $7m_h^{2m_h^5}$ rounds\label{slowdown2}
									\State take port $e[i]$ \label{line:btrav:take2}
								\EndWhile \label{tr:endwhile}
							\EndFor
							\State {\bf return} ${\tt true}$
						\End
					\end{algorithmic}
				\end{footnotesize}
			\end{algorithm}

			We also explained in the intuition section that the main part of hypothesis $h$ begins with a move to the node that supposedly corresponds to the node $v_h$ of $\phi_h$ having the smallest label. This move is executed via the function ${\tt MoveToCentralNode}(h)$ of Algorithm~\ref{alg:movecentral}: in the sequel, given any positive integer $h$, node $v_h$ will be called the central node of configuration $\phi_h$. In Algorithm~\ref{alg:movecentral}, given a node's label $L$ of configuration $\phi_h$, ${\tt path}_h(L)$ is a function that returns a sequence $p$ corresponding to the lexicographically smallest shortest path to follow in $\phi_h$ in order to reach the central node of $\phi_h$ from the node containing label $L$. If an agent labeled $L$ succeeds to follow entirely ${\tt path}_h(L)$ and notices it is with $k_h-1$ other agents after some prescribed period of time, its call to function ${\tt MoveToCentralNode}(h)$ will return true, otherwise it will return false.

			\begin{algorithm}
				\begin{footnotesize}
				\caption{Algorithm ${\tt MoveToCentralNode}(h)$ \label{alg:movecentral}}
				\begin{algorithmic}[1]
					\Begin
						\State $L\gets$ the label of the executing agent
						\If{there is a node labeled $L$ in $\phi_h$} \label{line:move:labeltest}
							\State $p\gets{\tt path}_h(L)$
							\For {$i \leftarrow 1$ {\bf to} $|p|$}
								\If{there is no port $p[i]$ at the current node}
									\State {\bf return} ${\tt false}$
								\EndIf
								\State exit by port $p[i]$ \label{line:move:move}
							\EndFor
							\State $j\gets 0$
							\While{${\tt CurCard}\ne k_h$ and $j<\mathcal{S}_h+n_h$} \label{line:move:whileb}
								\State wait \label{line:move:wait1}
								\State $j\gets j+1$
							\EndWhile \label{line:move:whilee}
							\If{${\tt CurCard}=k_h$}
								\State wait $\mathcal{S}_h+n_h$ rounds \label{line:move:wait}
								\If{${\tt CurCard}=k_h$} \label{line:move:card}
									\State {\bf return} ${\tt true}$
								\EndIf
							\EndIf
						\EndIf
						\State {\bf return} ${\tt false}$
					\End
				\end{algorithmic}
				\end{footnotesize}
			\end{algorithm}

			When ${\tt MoveToCentralNode}(h)$ returns true, the next step for the agent is to check whether it is exclusively with agents having succeeded their move to the central node of $\phi_h$ (indeed, the agent could for example be with agents processing higher hypotheses). This is the main purpose of function ${\tt StarCheck}(h)$ that is one of the components of the first scheme described earlier. In particular, when function ${\tt StarCheck}(h)$ returns true, the agent has the guarantee that it is in a group $G_h$ in which each agent has a label that is equal to a node's label in $\phi_h$. The function ${\tt rank}_h(L)$ used in the pseudocode of ${\tt StarCheck}(h)$, and also in another one below, returns the number of nodes having a label smaller than $L$ in configuration $\phi_h$.

			\begin{algorithm}
				\begin{footnotesize}
				\caption{Algorithm ${\tt StarCheck}(h)$} \label{alg:unk:star}
				\begin{algorithmic}[1]
					\Begin
						\State $L\gets$ the label of the executing agent
						\State $d\gets$ the degree of the current node
						\State $b\gets{\tt true}$
						\For {$t \leftarrow 1$ {\bf to} $2$}
							\For {$i \leftarrow 0$ {\bf to} $k_h-1$}
								\If{($i={\tt rank}_h(L)$ {\bf and} ($t=1$ {\bf or} ($t=2$ {\bf and} $b={\tt true}$)))} \label{line:star:test}
									\For {$j \leftarrow 0$ {\bf to} $d-1$} \label{line:star:starb}
										\State take port $j$ \label{line:star:move}
										\State $e\gets$ the port through which the agent has just entered the current node
										\If{ $t=1$ {\bf and} ${\tt CurCard}\ne 1$} \label{line:star:one}
											\State$b\gets{\tt false}$
										\EndIf
										\State take port $e$ \label{line:star:back}
										\If{${\tt CurCard}\ne k_h$} \label{line:star:card}
											\State$b\gets{\tt false}$
										\EndIf
									\EndFor \label{line:star:stare}
								\Else
									\For {$j \leftarrow 1$ {\bf to} $2d$} \label{line:star:tokb}
										\State wait \label{line:star:tokw}
										\If {($j\bmod 2=1$ {\bf and} ${\tt CurCard}\ne k_h-1$) {\bf or} ($j\bmod 2=0$ {\bf and} ${\tt CurCard}\ne k_h$)} \label{line:star:toktb}
												\State$b\gets{\tt false}$
										\EndIf \label{line:star:tokte}
									\EndFor \label{line:star:toke}
								\EndIf
							\EndFor
						\EndFor
						\State {\bf return} $b$
					\End
				\end{algorithmic}
				\end{footnotesize}
			\end{algorithm}

			\begin{algorithm}
				\begin{footnotesize}
				\caption{Algorithm ${\tt EnsureCleanExploration}(h)$\label{alg:ensure}}
				\begin{algorithmic}[1]
					\Begin
						\For {$t \leftarrow 1$ {\bf to} $2$} \label{line:clean:two}
							\ForEach{path $x$ of length $n^5_h+1$ from the set $\{0,\dots,n_h-2\}$} \label{line:clean:enum}
								\State $b\gets{\tt true}$; $i\gets 1$
								\While{$(b={\tt true }$ {\bf and} $i\leq n^5_h+1)$}
									\If{there is no port $x[i]$ at the current node}
										\State $b\gets{\tt false}$
									\Else
										\State take port $x[i]$
										\If{${\tt CurCard}\ne k_h$} \label{line:clean:card}
											\State {\bf return} ${\tt false}$
										\EndIf
										\State $i\gets i+1$
									\EndIf
								\EndWhile
								\State traverse in the reverse order the $i-1$ edges that have been traversed during the previous execution of the while loop \label{line:clean:back}
							\EndFor
						\EndFor
						\State {\bf return} ${\tt true}$
					\End
				\end{algorithmic}
				\end{footnotesize}
			\end{algorithm}

			At the beginning of the execution of line~\ref{alg:check1} in Algorithm~\ref{alg:hyp}, the gathering may be done. We will show it is the case if $\phi_h=\phi$. By chance, it could be also the case if $\phi_h\ne\phi$. However, at this point an agent cannot decide if the gathering is done or not. To enable such a decision, we make use of the functions ${\tt EnsureCleanExploration}(h)$ and ${\tt GraphSizeCheck}(h)$ described respectively in Algorithms~\ref{alg:ensure} and~\ref{alg:SizeCheck}. Function ${\tt GraphSizeCheck}(h)$ mainly relies on function ${\tt EST}^+(n_h)$ that derives from procedure ${\tt EST}$ introduced in the preliminaries section. Note that the execution of ${\tt EST}$ requires the existence of a genuine token in the network, which is not the case in our context: in ${\tt EST}^+(n_h)$ this issue is simply overcome by using some agents to play the role of a token. More precisely, function ${\tt EST}^+(n_h)$ consists of two consecutive parts. The first part is a simulation of ${\tt EST}$ in which the executing agent exploring the graph considers it is with its token in the rounds and only in the rounds in which ${\tt CurCard}>1$. The first part is completed as soon as the simulation is completed or as soon as the execution of the first part by the agent has lasted ${\tt T}({\tt EST}(n_h))$ rounds. In the second part, the agent backtracks by traversing in the reverse order the edges that have been traversed during the first part. At the end of the backtrack, the function returns true if the simulation of ${\tt EST}$ has been completed in time during the first part and if the graph size learned by the agent is exactly $n_h$, otherwise the function returns false.

			In the proof of correctness of Algorithm ${\tt GatherUnknownUpperBound}$, we will show that during each execution of function ${\tt EST^+}$ from a node $u_h$, the executing agent encounters another agent if and only if it is at node $u_h$: as mentioned in the intuition section, we then speak of clean exploration. This property comes from the fact that the call to ${\tt GraphSizeCheck}(h)$ in line~\ref{alg:check2} of Algorithm~\ref{alg:hyp} is made only if the call to function ${\tt EnsureCleanExploration}(h)$ in line~\ref{alg:check1} has returned true. If function ${\tt EnsureCleanExploration}(h)$ returns false, all agents of $G_h$ know that there are one or more agents that do not belong to $G_h$. On the other hand, if function ${\tt EnsureCleanExploration}(h)$ returns true, the agents of $G_h$ do not know whether there are one or more agents that do not belong to $G_h$, but all their simulations of ${\tt EST}$ in ${\tt GraphSizeCheck}(h)$ will be clean. Hence, if ${\tt EnsureCleanExploration}(h)$ returns true, all that an agent needs to conclude is to check whether the graph size of $\phi_h$ is equal to the size of the real network or not (checking whether the map of $\phi_h$ is the map of the real network is not necessary). When function ${\tt GraphSizeCheck}(h)$ is executed by the agents of $G_h$ in Algorithm~${\tt GatherUnknownUpperBound}$, the function returns true to all of them if $n=n_h$, and false to all of them otherwise. When the latter case occurs, each agent of $G_h$ knows that $\phi_h\neq\phi$ and will start ${\tt Hypothesis}(h+1)$. If the former case occurs, $\phi_h$ may or may not be $\phi$. However, whether $\phi_h=\phi$ or not, we will show that each agent of $G_h$ is guaranteed that during the execution of ${\tt EnsureCleanExploration}(h)$ it visited the whole graph and, since ${\tt EnsureCleanExploration}(h)$ returned true, there are no other agents than those of $G_h$, and thus the gathering is achieved.

			\begin{algorithm}
				\begin{footnotesize}
				\caption{Algorithm ${\tt GraphSizeCheck}(h)$\label{alg:SizeCheck}}
				\begin{algorithmic}[1]
					\Begin
						\State $L\gets$ the label of the executing agent
						\For {$i \leftarrow 1$ {\bf to} $k_h$} \label{line:size:for}
							\If{$i={\tt rank}_h(L)+1$}
								\State  $b\gets {\tt EST}^+(n_h)$
							\EndIf
							\State wait until having spent exactly $2i{\tt T}({\tt EST}(n_h))$ rounds in the current execution of this algorithm \label{line:size:wait}
						\EndFor
						\State {\bf return} $b$
					\End
				\end{algorithmic}
				\end{footnotesize}
			\end{algorithm}

		\subsection{Correctness}

			All the results given in this subsection are stated assuming that the algorithm that is executed by an agent when it wakes up is ${\tt GatherUnknownUpperBound}$ (cf. Algorithm~\ref{alg:mainnoupper}).

			We start with two basic propositions that follow directly from the formulation of Algorithms~\ref{alg:btraversal} to~\ref{alg:SizeCheck}. In the first proposition, we use the notation $\mathcal{F}(x)$. Given a positive integer $x$, $\mathcal{F}(x)$ is the set of the four routines that may be executed with input parameter $x$ during a call to ${\tt Hypothesis}(x)$ within Algorithm~${\tt GatherUnknownUpperBound}$ except ${\tt BallTraversal}(x)$.

			\begin{proposition}
				\label{proof:unk:distances}
				Let $x$ be a positive integer and let $F$ be a routine in $\mathcal{F}(x)\cup\{{\tt BallTraversal}(x)\}$ that is executed by an agent $A$ from a node $v$. We have the following properties.
				\begin{enumerate}
					\item If $F$ is ${\tt BallTraversal}(x)$, ${\tt StarCheck}(x)$ or ${\tt EnsureCleanExploration}(x)$, then during the execution of $F$, agent $A$ always remains at distance at most $4xm^5_x$, at most $1$, or at most $n^5_x+1$ from node $v$, respectively. If, in addition to being one of those three routines, $F$ returns {\tt true}, then agent $A$ is back at node $v$ when $F$ is completed.
					\item If $F$ is ${\tt MoveToCentralNode}(x)$, then during the execution of $F$, agent $A$ always remains at distance at most $n_x-1$ from node $v$.
					\item If $F$ is ${\tt GraphSizeCheck}(x)$, then during the execution of $F$, agent $A$ always remains at distance at most $n_x^5$ from node $v$. Moreover, agent $A$ is back at node $v$ when $F$ is completed.
				\end{enumerate}
			\end{proposition}

			\begin{proposition}
				\label{proof:unk:sizecheck}
				Let $x$ be a positive integer, and let  $A$ be an agent which calls function ${\tt Graph\-Size\-Check}(x)$. The execution by agent $A$ of ${\tt Graph\-Size\-Check}(x)$ lasts exactly $2k_x{\tt T}({\tt EST}(n_x))$ rounds.
			\end{proposition}

			Note that in view of Algorithms~\ref{alg:mainnoupper} and~\ref{alg:hyp}, for a given positive integer $x$, agent $A$ cannot call ${\tt Hypothesis}(x)$ or start executing line~\ref{alg:prepro2} of Algorithm~\ref{alg:hyp} during its execution of ${\tt Hypothesis}(x)$ more than once. In the rest of this subsection, given an agent $A$ we will denote by $v_A$ its initial node and by $t_A$ the round at which it wakes up. Moreover, we will denote by $s_{A,x}$ (resp. $w_{A,x}$) the round $t_A+\sum_{i=1}^{x-1}\mathcal{T}_i$ (resp. $s_{A,x}+{\tt T}({\tt Ball\-Traversal}(x))$).

			The following lemmas are related to the rounds at which the agents start their different calls to function ${\tt Hypothesis}$, and provide some properties about the beginning of any execution of this function. In particular, for every positive integer $x$, they aim at proving an upper bound on the time required to execute ${\tt Hypothesis}(x)$, and some synchronization property in the form of an upper bound on the delay between the rounds at which every pair of agents start executing ${\tt Hypothesis}(x)$. These two goals are achieved by Lemmas~\ref{proof:unk:timebound2} and~\ref{proof:unk:delay1}, respectively.

			\begin{lemma}
				\label{proof:unk:initnode}
				Let $A$ be an agent and let $x$ be a positive integer. If agent $A$ calls function ${\tt Hypothesis}(x)$, then it does so at node $v_A$.
			\end{lemma}

			\begin{proof}
				This proof is made by induction on $x$.

				First consider the base case $x=1$. Function ${\tt Hypothesis}(1)$ is called exactly once, at the beginning of the execution of ${\tt GatherUnknownUpperBound}$, before any move instruction. Hence, agent $A$ starts executing ${\tt Hypothesis}(1)$ at node $v_A$.

				Assume that there exists a positive integer $y$ such that agent $A$ starts executing ${\tt Hypothesis}(y)$ at node $v_{A}$. We aim at showing that, if agent $A$ starts executing ${\tt Hypothesis}(y+1)$, then it does so at node $v_{A}$. Assume that agent $A$ indeed starts executing ${\tt Hypothesis}(y+1)$. It does so only once and only if ${\tt Hypothesis}(y)$ returns {\tt false} (refer to Algorithm~\ref{alg:mainnoupper}). In view of lines~\ref{alg:hyp:lineback1} to~\ref{alg:hyp:lineback2} of Algorithm~\ref{alg:hyp} and the induction hypothesis, when ${\tt Hypothesis}(y)$ returns {\tt false}, agent $A$ is back at node $v_A$. Since there is no wait or move instruction between the end of ${\tt Hypothesis}(y)$ and the beginning of ${\tt Hypothesis}(y+1)$, node $v_A$ is the node from which agent $A$ starts ${\tt Hypothesis}(y+1)$. This concludes the inductive proof of this lemma.
			\end{proof}

			The following lemma is a rather technical one exhibiting the properties of function ${\tt BallTraversal}$. As explained in Section~\ref{subsec:intuition}, this function is used in particular to upper bound the durations of other routines such as ${\tt StarCheck}$, in the proof of Lemma~\ref{proof:unk:timebound1}. It is also crucial to establish the synchronization property of Lemma~\ref{proof:unk:delay1}, and the statement of the properties provided by the second scheme in Lemma~\ref{proof:unk:scheme}.

			\begin{lemma}
				\label{proof:unk:degree}
				Let $x$ be a positive integer. Let $A$ be an agent that calls ${\tt BallTraversal}(x)$. The execution by agent $A$ of ${\tt BallTraversal}(x)$ returns {\tt true} if and only if the degree of every node at distance at most $4xm_x^5$ from node $v_A$ is at most $n_x-1$. Moreover, if the execution of ${\tt BallTraversal}(x)$ by agent $A$ returns {\tt true}, then during this execution, agent $A$ visits every node at distance at most $4xm_x^5$ from node $v_A$.
			\end{lemma}

			\begin{proof}
				In view of Lemma~\ref{proof:unk:initnode}, the execution by agent $A$ of ${\tt Ball\-Traversal}(x)$ starts at node $v_A$. In view of lines~\ref{tr:debwhile} to~\ref{tr:endwhile} of Algorithm~\ref{alg:btraversal}, at the beginning of each iteration of the for loop of Algorithm~\ref{alg:btraversal}, agent $A$ is at node $v_A$. Each iteration of this loop consists of considering a path of length $4xm_x^5$ whose elements belong to $\{0, \dots, n_x-2\}$ and following it. Finding a node with degree at least $n_x$ leads function ${\tt Ball\-Traversal}(x)$ to return {\tt false}, and this is the single event which makes this function return {\tt false}. Since the execution by agent $A$ returns {\tt true}, for each path $\pi$ followed during one of the for loop iterations, agent $A$ visits all nodes of $\mathcal{N}(\pi, v_A)$ and each of them has degree at most $n_x-1$. Finally, in view of line~\ref{wordprocess} of Algorithm~\ref{alg:btraversal}, for each path of length $4xm_x^5$ whose elements belong to $\{0, \dots, n_x-2\}$, there is one iteration of the for loop which processes it. All paths of length $4xm_x^5$ from node $v_A$ are followed, which completes the proof.
			\end{proof}

			The three following lemmas state upper bounds on the number of rounds required to execute some of our routines. The two first ones are related to functions called by ${\tt Hypothesis}(x)$ (${\tt BallTraversal}(x)$, and some routines of $\mathcal{F}(x)$) while the third one concerns ${\tt Hypothesis}(x)$ itself. The second lemma is particularly important as, at a high level, it provides the duration of the waiting periods necessary to set up the first scheme presented in Section~\ref{subsec:intuition}. The three functions whose duration this lemma upper bounds (${\tt StarCheck}$, ${\tt EnsureCleanExploration}$, and ${\tt GraphSizeCheck}$) are the ones during the execution of which an agent must not be misled by agents executing other hypotheses: they are the sensitive part of function ${\tt Hypothesis}$.

			\begin{lemma}
				\label{proof:unk:timebound0}
				Let $x$ be a positive integer. Executing ${\tt BallTraversal}(x)$ requires at most ${\tt T}({\tt BallTraversal}(x))=64xm_x^{7xm_x^5}$ rounds.
			\end{lemma}

			\begin{proof}
				Function ${\tt Ball\-Traversal}(x)$ consists of processing $(n_x-1)^{4xm_x^5}$ distinct paths. The process of each of these paths requires at most $8xm_x^5(1+7m_x^{2m_x^5})$ rounds. Thus, the number of rounds needed by the execution of ${\tt Ball\-Traversal}(x)$ is at most $8xm_x^5n_x^{4xm_x^5}(1+7m_x^{2m_x^5})$. This is at most $64xm_x^{(4x+2)m_x^5+5}\leq 64xm_x^{7xm_x^5}$ rounds (which is referred to as ${\tt T}({\tt Ball\-Traversal}(x))$ in Section~\ref{sec:ninconnu}). The lemma is proved.
			\end{proof}

			\begin{lemma}
				\label{proof:unk:timebound1}
				Let $x$ be a positive integer. During every execution of ${\tt Hypothesis}(x)$, the total duration of the execution of lines~\ref{line:hyp:star} to~\ref{alg:check2} of Algorithm~\ref{alg:hyp} (\ie routines ${\tt StarCheck}(x)$, ${\tt EnsureCleanExploration}(x)$, and ${\tt GraphSizeCheck}(x)$) is at most $7n_x^{2n_x^5}$ rounds.
			\end{lemma}

			\begin{proof}
				Executing function ${\tt Star\-Check}(x)$ (refer to Algorithm~\ref{alg:unk:star}) takes exactly $4dk_x$ rounds, where $d$ is the degree of the node $v$ in which the execution starts. In view of Algorithm~\ref{alg:hyp}, every execution of function ${\tt StarCheck}(x)$ is preceded by executions of routines ${\tt Ball\-Traversal}(x)$ and ${\tt MoveToCentralNode}(x)$. In view of Proposition~\ref{proof:unk:distances} and Lemma~\ref{proof:unk:initnode}, node $v$ is at distance at most $n_x-1$ from node $v_A$. In view of Lemma~\ref{proof:unk:degree}, this means that its degree is at most $n_x-1$. Hence, the number of rounds taken by the execution of ${\tt StarCheck}(x)$ is at most $4n_x(n_x-1)$.

				Function ${\tt Ensure\-Clean\-Exploration}(x)$ (refer to Algorithm~\ref{alg:ensure}) consists of processing $2(n_x-1)^{n_x^5+1}$ paths. The process of each of these paths requires at most $2(n_x^5+1)$ edge traversals. Since $n_x\geq 2$, the number of rounds required by the execution of this function is thus at most $4n_x^{n_x^5+6}$.

				In view of Proposition~\ref{proof:unk:sizecheck}, every execution of function ${\tt Graph\-Size\-Check}(x)$ requires at most $2n_x^6$ rounds.

				Hence, the total duration of the executions of the three functions is at most $4n_x^{n_x^5+6}+2n_x^6+4n_x(n_x-1)$. Since $n_x\geq 2$, $4n_x(n_x-1)\leq n_x^6$. This implies that $4n_x^{n_x^5+6}+2n_x^6+4n_x(n_x-1)\leq 4n_x^{n_x^5+6}+3n_x^6$ which is at most $7n_x^{n_x^5+6}$ and thus at most $7n_x^{2n_x^5}$, since $n_x\geq2$.
			\end{proof}

			\begin{lemma}
				\label{proof:unk:timebound2}
				Let $x$ be a positive integer. Let $A$ be an agent. If $A$ calls ${\tt Hypothesis}(x)$, then it does so (resp. starts executing line~\ref{alg:prepro2} of Algorithm~\ref{alg:hyp} during its execution of ${\tt Hypothesis}(x)$) in round $s_{A,x}$ (resp. in round $w_{A,x}$ at the latest). If agent $A$ has not declared that the gathering is achieved by round $s_{A,x}+1$ (resp. $w_{A,x}+1$), then in round $s_{A,x}$ (resp. $w_{A,x}$), it starts executing ${\tt Hypothesis}(x)$ (resp. has started executing line~\ref{alg:prepro2} of Algorithm~\ref{alg:hyp} during its execution of ${\tt Hypothesis}(x)$).
			\end{lemma}

			\begin{proof}
				First note that if agent $A$ starts executing ${\tt Hypothesis}(x)$ at a round $r$, then in view of Lemma~\ref{proof:unk:timebound0}, it starts executing line~\ref{alg:prepro2} of Algorithm~\ref{alg:hyp} during its execution of ${\tt Hypothesis}(x)$, at the latest in round $r+{\tt T}({\tt Ball\-Traversal}(x))$. Thus, to prove the lemma, it is enough to show that if agent $A$ calls ${\tt Hypothesis}(x)$, then it does so in round $s_{A,x}$, and that if agent $A$ has not declared that the gathering is achieved by round $s_{A,x}+1$, then in round $s_{A,x}$ it starts executing ${\tt Hypothesis}(x)$.

				This proof is by induction on $x$. First consider the base case $x=1$. In view of Algorithm~\ref{alg:mainnoupper}, agent $A$ starts executing ${\tt Hypothesis}(1)$ in round $t_A=s_{A,1}$ when it wakes up, and in which round it cannot declare that the gathering is achieved. This completes the analysis of the case when $x=1$.

				We have two parts of the induction step. Assuming that there exists a positive integer $y$ such that if agent $A$ calls ${\tt Hypothesis}(y)$, then it does so in round $s_{A,y}$, we aim at showing that if agent $A$ calls ${\tt Hypothesis}(y+1)$, then it does so in round $s_{A,y+1}$. Moreover, assuming that there exists a positive integer $y$ such that if agent $A$ has not declared that the gathering is achieved by round $s_{A,y}+1$, then in round $s_{A,y}$ it starts executing ${\tt Hypothesis}(y)$, we aim at showing that if agent $A$ has not declared that the gathering is achieved by round $s_{A,y+1}+1$, then in round $s_{A,y+1}$ it starts executing ${\tt Hypothesis}(y+1)$.

				Let us first show that if agent $A$ starts executing ${\tt Hypothesis}(y)$, then it uses a time of at most $3\mathcal{S}_y+2{\tt T}({\tt Ball\-Traversal}(y))$ rounds executing the first part of ${\tt Hypothesis}(y)$ (lines~\ref{alg:prepro1} to~\ref{alg:fin} of Algorithm~\ref{alg:hyp}). Every execution of function ${\tt Move\-To\-Central\-Node}(y)$ consists, in the worst case, of performing $n_y-1$ edge traversals and waiting twice $\mathcal{S}_y+n_y$ rounds. Hence, in view of Lemmas~\ref{proof:unk:timebound0} and~\ref{proof:unk:timebound1}, and taking into account line~\ref{alg:prepro2} of Algorithm~\ref{alg:hyp}, agent $A$ spends at most $3\mathcal{S}_y+3n_y+7n_y^{2n_y^5}+{\tt T}({\tt Ball\-Traversal}(y))$ rounds executing the first part of ${\tt Hypothesis}(y)$. This is at most $3\mathcal{S}_y+2{\tt T}({\tt Ball\-Traversal}(y))$ rounds.

				In view of the induction hypotheses, this implies that if the execution of ${\tt Hypothesis}(y)$  by agent $A$ returns {\tt true}, then it does so (and thus agent $A$ declares that the gathering is achieved) in round $s_{A,y}+3\mathcal{S}_y+2{\tt T}({\tt Ball\-Traversal}(y))$ at the latest. If this happens, then agent $A$ does not start executing ${\tt Hypothesis}(y+1)$, and has declared that the gathering is achieved by round $s_{A,y+1}+1$.

				We now assume that agent $A$ calls ${\tt Hypothesis}(y+1)$ or that it has not declared that the gathering is achieved by round $s_{A,y+1}+1$. In view of the previous paragraph, the execution by agent $A$ of ${\tt Hypothesis}(y)$ cannot return {\tt true}. Note that since $s_{A,y+1}-s_{A,y}=\mathcal{T}_y$, showing that if agent $A$ starts executing ${\tt Hypothesis}(y)$, then it spends exactly $\mathcal{T}_y$ rounds executing this routine is enough to complete the proof of the lemma.

				At the end of the second part of ${\tt Hypothesis}(y)$, before returning {\tt false}, agent $A$ waits until having spent at least $\mathcal{T}_y$ rounds executing ${\tt Hypothesis}(y)$ (refer to line~\ref{alg:hyp:linewait2} of Algorithm~\ref{alg:hyp}). Unless the round in which the execution by agent $A$ of lines~\ref{alg:debut} to~\ref{alg:hyp:lineback2} of ${\tt Hypothesis}(y)$ is completed is strictly later than round $s_{A,y} + \mathcal{T}_y$, the execution by agent $A$ of ${\tt Hypothesis}(y)$ is completed in round $s_{A,y} + \mathcal{T}_y$.

				Since agent $A$ spends at most $3\mathcal{S}_y+2{\tt T}({\tt Ball\-Traversal}(y))$ rounds in the execution of the first part of ${\tt Hypothesis}(y)$, the execution of lines~\ref{alg:hyp:lineback1} to~\ref{alg:hyp:lineback2} during the second part of ${\tt Hypothesis}(y)$ requires at most $(1+7m_y^{2m_y^5})(3\mathcal{S}_y+2{\tt T}({\tt Ball\-Traversal}(y)))$. Hence the execution of ${\tt Hypothesis}(y)$ by agent $A$ takes exactly $\mathcal{T}_y$ rounds, which concludes the proof of the lemma.
			\end{proof}

			The next lemma states a synchronization property in the form of an upper bound on the delay between the rounds in which every pair of agents start executing ${\tt Hypothesis}(x)$, for every positive integer $x$.

			\begin{lemma}
				\label{proof:unk:delay1}
				Let $x$ be a positive integer. Let $A$ and $B$ be two agents. Assume that agent $A$ calls ${\tt Hypothesis}(x)$. If there exists a path $\pi$ of length at most $4xm_x^5$ from node $v_A$ to node $v_B$ and each node at distance at most $4xm_x^5$ from node $v_A$ has degree at most $n_x-1$, then round $s_{B,x}$ is at most $s_{A,x} + \mathcal{S}_x$ and round $w_{B,x}$ is at most $w_{A,x}+\mathcal{S}_x$.
			\end{lemma}

			\begin{proof}
				In view of Lemma~\ref{proof:unk:timebound2} and Algorithm~\ref{alg:hyp}, in round $s_{A,x}$, agent $A$ starts executing ${\tt Ball\-Traversal}(x)$. In view of Lemma~\ref{proof:unk:degree}, during this execution $A$ visits every node at distance at most $4xm_x^5$ from node $v_A$, which includes node $v_B$. For this reason, agent $B$ is woken up at the latest when the execution by agent $A$ of ${\tt BallTraversal}(x)$ is completed. In other words, the round $t_B$ in which agent $B$ is woken up, is at most $w_{A,x}$.

				Round $s_{B,x}$ is equal to $t_B+\sum_{y=1}^{x-1}(\mathcal{T}_{y})\leq w_{A,x}+\sum_{y=1}^{x-1}(\mathcal{T}_{y})$, which is equal to $s_{A,x}+\mathcal{S}_{x}$. Moreover, round $w_{B,x}$ is equal to $t_B+\sum_{y=1}^{x-1}(\mathcal{T}_{y})+{\tt T}({\tt Ball\-Traversal}(x))\leq w_{A,x}+\sum_{y=1}^{x-1}(\mathcal{T}_{y})+{\tt T}({\tt Ball\-Traversal}(x))$, which is equal to $w_{A,x}+\mathcal{S}_{x}$. This concludes the proof.
			\end{proof}

			The next lemma states the properties concerning the second scheme presented in Section~\ref{subsec:intuition}. This lemma is subsequently used in the proofs of Lemmas~\ref{proof:unk:star} and~\ref{proof:unk:ensure} to show that the agents executing other hypotheses can be identified by dancing protocols such as ${\tt StarCheck}$ and ${\tt EnsureCleanExploration}$ and thus do not mislead the agents executing ${\tt GraphSizeCheck}$.

			\begin{lemma}
				\label{proof:unk:scheme}
				Let $x$ be a positive integer. Let $A$ and $B$ be two agents. Let $r_1\leq r_2$ be two rounds int which agent $A$ executes move instructions belonging to $\mathcal{F}(x)\setminus\{{\tt MoveToCentralNode}(x)\}$. Let $v$ be a node and let $r_A$ be a round belonging to $\{r_1, \dots, r_2+1\}$ such that agent $A$ is at node $v$ in round $r_A$. For every round $r_B$ of $\{r_1, \dots, r_2\}$, if agent $B$ executes a move instruction $i_B$ in round $r_B$, and is at node $v$ in round $r_B$ or in round $r_B+1$, then instruction $i_B$ belongs:
				\begin{itemize}
					\item either to ${\tt Hypothesis}(y)$ with $y\geq x$ but to none of the routines of $\mathcal{F}(y)$.
					\item or to a waiting period of ${\tt MoveToCentralNode}(x)$ (lines~\ref{line:move:wait1} or~\ref{line:move:wait} of Algorithm~\ref{alg:movecentral}).
					\item or to a routine of $\mathcal{F}(x)\setminus\{{\tt MoveToCentralNode}(x)\}$.
				\end{itemize}
			\end{lemma}

			\begin{proof}
				Let $y$ be the positive integer such that instruction $i_B$ belongs to ${\tt Hypothesis}(y)$. Denote by $i_1$ and $i_2$ the move instructions executed by agent $A$ in rounds $r_1$ and $r_2$ respectively. Also denote by $i_A$ the move instruction executed by agent $A$ either in round $r_A$ if $r_A\leq r_2$, or in round $r_2$ otherwise. Since instruction $i_1$ belongs to a routine of $\mathcal{F}(x)\setminus\{{\tt MoveToCentralNode}(x)\}$, in round $r_1$ agent $A$ has completed the execution of line~\ref{line:move:wait} of Algorithm~\ref{alg:movecentral}, and thus has spent more than $\mathcal{S}_x$ rounds in its execution of ${\tt Hypothesis}(x)$. In view of Lemma~\ref{proof:unk:timebound2}, this means that round $r_1$ is larger than $s_{A,x}+\mathcal{S}_x$.

				This proof relies on the following technical claim.

				\begin{claim}
					\label{proof:unk:schemec}
					{Let $v$ be a node and $C$ (resp. $D$) be an agent which executes a move instruction $i_C$ (resp. $i_D$) of ${\tt Hypothesis}(p)$ (resp. ${\tt Hypothesis}(q)$), with $p\geq q$, in a round $r_C$ (resp. $r_D$). Assume that agent $C$ (resp. $D$) is at node $v$ in at least one round of $\{r_C, r_C+1\}$ (resp. $\{r_D, r_D+1\}$).} Also assume that instruction $i_C$ belongs to a routine $F$ of $\mathcal{F}(p)$, and either $p>q$ or instruction $i_D$ also belongs to a routine of $\mathcal{F}(p)$. The distance between nodes $v_C$ and $v_D$ is at most $4pm_p^5$ and every node at distance at most $4pm_p^5$ from node $v_C$ has degree at most $n_p-1$.
				\end{claim}

				\begin{proofclaim}
					In view of Algorithm~\ref{alg:hyp}, function ${\tt BallTraversal}(p)$ and every routine of $\mathcal{F}(p)$ whose execution has been completed by agent $C$ before it starts function $F$, has returned {\tt true}. In view of Lemma~\ref{proof:unk:degree}, the degree of every node at distance at most $4pm_p^5$ from node $v_C$ is at most $n_p-1$. As a result, it remains to show that node $v_D$ is at distance at most $4pm_p^5$ from node $v_C$.

					In view of Proposition~\ref{proof:unk:distances}, once its execution of ${\tt BallTraversal}(p)$ is completed, agent $C$ is back at node $v_C$, and in rounds $r_C$ and $r_C+1$, it is at distance at most $n_p+n_p^5$ from node $v_C$ (this corresponds to the case in which ${\tt MoveToCentralNode}(p)$ brings it to some node $w$ at distance $n_p-1$ from node $v_C$, and ${\tt EnsureCleanExploration}(x)$ leads it at distance $n_p^5+1$ from node $w$).

					Still in view of Proposition~\ref{proof:unk:distances}, if $p\neq q$ then in rounds $r_D$ and $r_D+1$, agent $D$ is at distance at most $4qm_q^5$ from node $v_D$, and if instruction $i_D$ belongs to a routine of $\mathcal{F}(p)$, then in rounds $r_D$ and $r_D+1$, it is at distance at most $n_p+n_p^5$ from $v_D$. In the first case, the distance between nodes $v_C$ and $v_D$ is at most $n_p+n_p^5 + 4qm_q^5$ with $p>q$ \ie at most $4pm_p^5$, while in the second case, it is at most $2(n_p+n_p^5)$ which is also at most $4pm_p^5$. Hence, node $v_D$ is at distance at most $4pm_p^5$ from node $v_C$, which concludes the proof of this claim.
				\end{proofclaim}

				First assume for the sake of contradiction that $y<x$. In view of Algorithm~\ref{alg:hyp}, instruction $i_A$ belongs to a routine of $\mathcal{F}(x)\setminus\{{\tt MoveToCentralNode}(x)\}$. In view of Claim~\ref{proof:unk:schemec}, the distance between nodes $v_A$ and $v_B$ is at most $4xm_x^5$ and every node at most at this distance from node $v_A$ has degree at most $n_x-1$. In view of Lemma~\ref{proof:unk:delay1}, this means that round $s_{B,x}$ is at most $s_{A,x} + \mathcal{S}_x$. Hence, we know that $s_{B,x}\leq s_{A,x}+\mathcal{S}_x<r_1$. Since agent $B$ executes instruction $i_B$ in round $r_B\geq r_1$, it has not declared that the gathering is achieved by round $s_{B,x}+1$ and thus, in view of Lemma~\ref{proof:unk:timebound2}, starts executing ${\tt Hypothesis}(x)$ in round $s_{B,x}$. This contradicts the assumption that in round $r_B$ which is strictly larger than $s_{B,x}$, agent $B$ executes instruction $i_B$ which belongs to ${\tt Hypothesis}(y)$ with $x>y$. This contradiction proves $y \geq x$.

				If instruction $i_B$ belongs to $\mathcal{F}(y)$, then since $y\geq x$ and instruction $i_A$ belongs to $\mathcal{F}(x)$, in view of Claim~\ref{proof:unk:schemec}, the distance between nodes $v_A$ and $v_B$ is at most $4ym_y^5$ and every node at most at this distance from node $v_B$ has degree at most $n_y-1$. Otherwise, $y\geq x$ and instruction $i_B$ belongs to none of the routines of $\mathcal{F}(y)$, which corresponds to the first bullet point in the statement of the lemma.

				Hence, we may assume that the distance between nodes $v_A$ and $v_B$ is at most $4ym_y^5$ and every node at most at this distance from node $v_B$ has degree at most $n_y-1$. In view of Lemma~\ref{proof:unk:delay1}, round $s_{A,y}$ is at most $s_{B,y}+\mathcal{S}_y$. Either $y>x$ or $y=x$. If $y>x$, then since agent $A$ executes a move instruction of ${\tt Hypothesis}(x)$ in round $r_2$, in view of Lemma~\ref{proof:unk:timebound2}, we know that $r_2<s_{A,y}\leq s_{B,y}+\mathcal{S}_y$. Still in view of Lemma~\ref{proof:unk:timebound2}, this means that in round $r_B\leq r_2$, agent $B$ has performed at most $\mathcal{S}_y$ move instructions during its execution of ${\tt Hypothesis}(y)$. In view of the waiting period at line~\ref{alg:prepro2} of Algorithm~\ref{alg:hyp}, instruction $i_B$ belongs either to the latter line or to ${\tt BallTraversal}(y)$. This case corresponds to the first bullet point in the statement of the lemma.

				Attention can be thus restricted to the case when $y=x$. In this case, it is enough to show that if instruction $i_B$ belongs to ${\tt MoveToCentralNode}(x)$, then it does not consist of taking some port. Since agent $A$ calls ${\tt MoveToCentralNode}(x)$, Lemma~\ref{proof:unk:degree} implies that the degree of every node at distance at most $4xm_x^5$ from node $v_A$ is at most $n_x-1$. In view of Lemma~\ref{proof:unk:delay1}, since the distance between nodes $v_A$ and $v_B$ is at most $4xm_x^5$, round $w_{B,x}$ is at most $w_{A,x}+\mathcal{S}_x$. In view of Algorithm~\ref{alg:hyp}, agents $A$ and $B$ spend the same number of rounds executing line~\ref{alg:prepro2} of ${\tt Hypothesis}(x)$. Then, agent $A$ spends at least $\mathcal{S}_x+n_x$ rounds executing ${\tt MoveToCentralNode}(x)$ (line~\ref{line:move:wait} of Algorithm~\ref{alg:movecentral}), and agent $B$ spends at most $n_x-1$ rounds moving in ${\tt MoveToCentralNode}(x)$ before reaching the waiting periods of lines~\ref{line:move:wait1} and~\ref{line:move:wait}. Hence, in view of Lemma~\ref{proof:unk:timebound2}, the last round in which agent $B$ executes a move instruction that consists of taking some port and belongs to ${\tt MoveToCentralNode}(x)$ is at most $w_{B,x}+\mathcal{S}_x+n_x-1$ which is at most $w_{A,x}+2\mathcal{S}_x+n_x-1$ and thus at most $r_1-1$. This proves that if the move instruction $i_B$ executed in round $r_B\geq r_1$ belongs to ${\tt MoveToCentralNode}(x)$, then it does not consist of taking some port, which concludes the proof.
			\end{proof}

Let $L_x$ be the set of labels of agents in configuration $\phi_x$.
			In view of Algorithm~\ref{alg:movecentral} and in particular of its line~\ref{line:move:labeltest}, we have the following proposition used in the proofs of the next lemmas. It reflects the fact that the agents whose label does not belong to $L_x$ ``quickly notice'' that hypothesis $x$ is not good, which allows us to restrict attention to the agents whose labels belong to $L_x$ which may not have noticed it yet when they start executing ${\tt StarCheck}(x)$.

			\begin{proposition}
				\label{proof:unk:movelabel}
				Let $x$ be a positive integer and let $A$ be an agent. Assume that $A$ executes line~\ref{line:move:wait1} or line~\ref{line:move:wait} of Algorithm~\ref{alg:movecentral} during its execution of ${\tt MoveToCentralNode}(x)$. Then its label belongs to $L_x$.
			\end{proposition}

			The next lemma is the most involved one. Assuming that an agent $A$ calls ${\tt EnsureCleanExploration}(x)$, its proof consists in showing that the dancing protocol performed in routine ${\tt StarCheck}(x)$ is not misled by the agents executing other hypotheses. Roughly speaking, Lemma~\ref{proof:unk:star} ensures that when an agent $A$ starts executing ${\tt EnsureCleanExploration}(x)$, it belongs to a group of synchronized agents that ``think'' configuration $\phi_x$ could be good. The proof of this lemma is particularly involved because of the subtlety of the dances used by function ${\tt StarCheck}(x)$, compared, e.g.,  to those of ${\tt EnsureCleanExploration}(x)$.

			\begin{lemma}
				\label{proof:unk:star}
				Let $x$ be a positive integer. Let $A$ be an agent which calls ${\tt EnsureCleanExploration}(x)$ in a round $t$ at a node $v$. There is an agent with label $l$ which calls ${\tt EnsureCleanExploration}(x)$ in round $t$ at node $v$ if and only if $l\in L_x$.
			\end{lemma}

			\begin{proof}
				In view of Proposition~\ref{proof:unk:distances} and Algorithms~\ref{alg:hyp} and~\ref{alg:unk:star}, agent $A$ calls ${\tt StarCheck}(x)$ at node $v$ in round $t-4dk_x$, where $d$ is the degree of $v$. For every agent which executes ${\tt StarCheck}(x)$, in view of Proposition~\ref{proof:unk:movelabel} and Algorithm~\ref{alg:hyp}, its label  belongs to $L_x$. To complete the proof of the lemma, it is thus enough to show that for every $l\in L_x$, the agent with label $l$ starts executing ${\tt EnsureCleanExploration}(x)$ in round $t$ at node $v$. This is verified if $k_x=1$. Thus, we assume that $k_x\geq 2$.

				First assume that $t$ is the earliest round in which an agent calls ${\tt EnsureCleanExploration}(x)$. In this case, we need the following notions of checking instruction and rank. Given an agent with label $l\in L_x$, its rank is the number of nodes in configuration $\phi_x$ with a label strictly smaller than $l$. A checking instruction is a move instruction executed in a round at least $t-4dk_x$ and at most $t-1$ consisting either in leaving or in entering node $v$. Denote by $z$ the rank of agent $A$. In view of Algorithm~\ref{alg:unk:star}, in round $t-4dk_x+2dz$, agent $A$ starts executing lines~\ref{line:star:starb} to~\ref{line:star:stare} of Algorithm~\ref{alg:unk:star}. Moreover, the execution of these lines is completed in round $t-4dk_x+2d(z+1)$. This fact is used by the following claims to show increasingly stronger properties related to the checking instructions, until proving that the lemma is verified when $t$ is the earliest round in which an agent calls ${\tt EnsureCleanExploration}(x)$.

				\begin{claim}
					\label{proof:unk:starc0}
					For every $o\in\{0, \dots, 4dk_x-1\}$, there is a checking instruction executed in round $t-4dk_x+o$ and if $o\bmod 2=0$ (resp. $o\bmod 2=1$), then it consists in leaving (resp. entering) node $v$.
				\end{claim}

				\begin{proofclaim}
					Since agent $A$ starts executing lines~\ref{line:star:starb} to~\ref{line:star:stare} of Algorithm~\ref{alg:unk:star} at node $v$ in round $t-4dk_x+2dz$ and this execution is completed in round $t-4dk_x+2d(z+1)$, the claim is verified for every $o\in\{2dz, \dots, 2d(z+1)-1\}$. Moreover, agent $A$ never assigns {\tt false} to its variable $b$ during its execution of ${\tt StarCheck}(x)$.

					For every $p\in\{0, \dots, 2k_x-1\}\setminus\{z\}$, agent $A$ starts executing lines~\ref{line:star:tokb} to~\ref{line:star:toke} of Algorithm~\ref{alg:unk:star}  in round $t-4dk_x+2dp$. For every $0<q\leq d$, in round $t-4dk_x+2dp+2q$ (resp. $t-4dk_x+2dp+2q-1$), there are $k_x$ (resp. $k_x-1$) agents at node $v$. Moreover, in view of line~\ref{line:move:card} of Algorithm~\ref{alg:movecentral}, in round $t-4dk_x$, there are $k_x$ agents at node $v$. Furthermore, in view of line~\ref{line:star:card} of Algorithm~\ref{alg:unk:star}, there are $k_x$ agents at node $v$ in round $t-4dk_x+2d(z+1)$.

					Hence, for every $o\in\{0, \dots, 2dz\}\cup\{2d(z+1), \dots, 4dk_x\}$, if $o\bmod 2=0$ (resp. $o\bmod 2=1$), then in round $t-4dk_x+o$, there are $k_x$ (resp. $k_x-1$) agents at node $v$. This means that, for every $o\in\{0, \dots, 2dz-1\}\cup\{2d(z+1), \dots, 4dk_x-1\}$, if $o\bmod 2=0$ (resp. $o\bmod 2=1$), then in round $t-4dk_x+o$, at least one checking instruction consisting in leaving (resp. entering) node $v$ is executed. This proves the claim.
				\end{proofclaim}

				 In view of Claim~\ref{proof:unk:starc0}, there are at least $4dk_x$ checking instructions. This is a key property in this proof. We are going to show that most move instructions of our algorithm cannot be checking instructions, and thus that unless the lemma is verified, there cannot be $4dk_x$ of them.

				Since agent $A$ executes move instructions belonging to ${\tt StarCheck}(x)$ both in round $t-4dk_x$ and in round $t-1$, and is at node $v$ in round $t-4dk_x$, we can use Lemma~\ref{proof:unk:scheme} to restrict the set of move instructions in which each checking instruction can be. Note that checking instructions consist in taking some port and thus do not belong to waiting periods. Moreover, since the earliest round in which an agent calls ${\tt EnsureCleanExploration}(x)$ is $t$, the checking instructions can belong neither to function ${\tt EnsureCleanExploration}(x)$ nor to function  ${\tt GraphSizeCheck}(x)$. Hence, for every checking instruction $i$, either there exists $y\geq x$ such that instruction $i$ belongs to ${\tt Hypothesis}(y)$ but to none of the routines of $\mathcal{F}(y)$, or instruction $i$ belongs to ${\tt StarCheck}(x)$.

				In the first case, instruction $i$ is executed either in line~\ref{line:hyp:back} of Algorithm~\ref{alg:hyp}, or in line~\ref{line:btrav:take1} of Algorithm~\ref{alg:btraversal}, or in line~\ref{line:btrav:take2} of Algorithm~\ref{alg:btraversal}. Whichever the case, before and after this edge traversal, the agent executing instruction $i$ waits at least $7m_x^{2m_x^5}$ rounds.

				For this reason, a checking instruction $i$ for which there exists a positive integer $y\geq x$ such that instruction $i$ belongs to ${\tt Hypothesis}(y)$ but to none of the routines of $\mathcal{F}(y)$ is called a slow checking instruction, while a checking instruction belonging to ${\tt StarCheck}(x)$ is called a fast checking instruction. Moreover, a round $r$ is called a leaving round (resp. an entering round) if and only if there exists an even (resp. odd) integer $o$ belonging to $\{0, \dots, 4dk_x-1\}$ such that $r=t-4dk_x+o$. Lastly, an agent $B$ is said to satisfy round $r$ if and only if either $r$ is a leaving round and in round $r$ agent $B$ executes a checking instruction consisting in leaving node $v$, or $r$ is an entering round and in round $r$ agent $B$ executes a checking instruction consisting in entering node $v$.

				\begin{claim}
					\label{proof:unk:starc1}
					No slow checking instruction consisting in leaving (resp. entering) node $v$ can be executed in a round strictly smaller than round $t-4dk_x+2dz$ (resp. at least round $t-4dk_x+2d(z+1)$).
				\end{claim}

				\begin{proofclaim}
					Assume for the sake of contradiction that the claim does not hold. This means that there exists an agent $B$, and either a round $r_1<t-4dk_x+2dz$ in which agent $B$ executes a slow checking instruction $i_1$ consisting in leaving node $v$, or a round $r_2\geq t-4dk_x+2d(z+1)$ in which agent $B$ executes a slow checking instruction $i_2$ consisting in entering node $v$. The proof is similar for both cases.

					There is a neighbor $u$ of node $v$, such that either agent $B$ is at node $v$ from round $t-4dk_x$ to round $r_1$ (both included) and at node $u$ from round $r_1+1$ to round $t$ (both included), or agent $B$ is at node $u$ from round $t-4dk_x$ to round $r_2$ (both included) and at node $v$ from round $r_2+1$ to round $t$ (both included). Whichever the case, agent $B$ is at node $u$ from round $r_1+1$ to round $r_2$ \ie at least from round $t-4dk_x+2dz$ to round $t-4dk_x+2d(z+1)$.

					During its execution of lines~\ref{line:star:starb} to~\ref{line:star:stare} of Algorithm~\ref{alg:unk:star} from round $t-4dk_x+2dz$ to round $t-4dk_x+2d(z+1)$, agent $A$ visits every neighbor of node $v$ including $u$, notices a cardinality different from 1 at node $u$ and assigns {\tt false} to its variable $b$, which contradicts the fact that it starts executing ${\tt EnsureCleanExploration}(x)$ in round $t$ and thus completes the proof of the claim.
				\end{proofclaim}

				In view of Claims~\ref{proof:unk:starc0} and ~\ref{proof:unk:starc1}, for every $o\in\{0, \dots, 4dk_x-1\}$, if $o<2dz$ and $o\bmod 2=0$ (resp. $o\geq 2d(z+1)$ and $o\bmod 2=1$), there is at least one checking instruction consisting in leaving (resp. entering) node $v$ but no slow checking instruction which is executed in round $t-4dk_x+o$. For every $o\in\{0, \dots, 4dk_x-1\}$, round $t-4dk_x+o$ is called a fast-only leaving (resp. fast-only entering) round if and only if $o<2dz$ and $o\bmod 2=0$ (resp. $o\geq 2d(z+1)$ and $o\bmod 2=1$).

				Thanks to the existence of fast-only rounds (which can only be satisfied by agents executing ${\tt StarCheck}(x)$), the next claim proves some properties concerning the agents executing ${\tt StarCheck}(x)$. Its statement uses the notion of {\em slice}. For every positive integer $p$ at most $2k_x$, the $p$-th slice is the set of rounds $\{t-4dk_x+2d(p-1), \dots, t-4dk_x+2dp-1\}$.

				\begin{claim}
					\label{proof:unk:starc2}
					For every positive integer $p\leq 2k_x$, there is one agent $C$, such that all fast-only rounds of the $p$-th slice are satisfied by agent $C$ and only by it. Moreover, for every agent $B\neq A$ whose label belongs to $L_x$, there exist two positive integers $p_1$ and $p_2$ both at most $2k_x$, such that agent $B$ is the only agent to satisfy the fast-only rounds of the $p_1$-th slice and the $p_2$-th slice.
				\end{claim}

				\begin{proofclaim}
					For every positive integer $p\leq 2k_x$, if $p\neq z+1$, then there are $d$ fast-only rounds in the $p$-th slice. If $p<z+1$ (resp. $p>z+1$), these are fast-only leaving (resp. entering) rounds. Thus, there are $2k_x-1$ slices with $d$ fast-only rounds, and one slice, the $(z+1)$-th, during which agent $A$ executes lines~\ref{line:star:starb} to~\ref{line:star:stare} of Algorithm~\ref{alg:unk:star}, and no round of which is fast-only.

					There are only $k_x$ agents whose labels belong to $L_x$. Each of them executes at most twice lines~\ref{line:star:starb} to~\ref{line:star:stare} of Algorithm~\ref{alg:unk:star}. Moreover, executing these lines once permits to satisfy at most $d$ fast-only rounds. Hence, the maximum number of fast-only rounds which can be satisfied is $d+2d(k-1)=d(2k_x-1)$ \ie the number of fast-only rounds. This implies that satisfying all fast-only rounds requires the three following facts. (1) Each agent whose label belongs to $L_x$ executes lines~\ref{line:star:starb} to~\ref{line:star:stare} of Algorithm~\ref{alg:unk:star} twice during its execution of ${\tt StarCheck}(x)$. (2) Each execution of lines~\ref{line:star:starb} to~\ref{line:star:stare} of Algorithm~\ref{alg:unk:star} by an agent whose label belongs to $L_x$ during its execution of ${\tt StarCheck}(x)$ {(except the first one by agent $A$)} satisfies $d$ fast-only rounds. (3) No fast-only round is satisfied by two agents.

					Consider the earliest fast-only round. There is an agent whose label belongs to $L_x$ which satisfies it by executing lines~\ref{line:star:starb} to~\ref{line:star:stare} of Algorithm~\ref{alg:unk:star} during its execution of ${\tt StarCheck}(x)$. The execution of these lines satisfies all $d$ earliest fast-only rounds \ie {all $d$ fast-only rounds of one slice which is the first one if $z\neq 0$ and the second one otherwise}. Then, similarly, for every positive integer $p\leq 2k_x$, if $p\neq z+1$, there is an agent whose label belongs to $L_x$ which satisfies the $d$ fast-only rounds of the $p$-th slice by executing lines~\ref{line:star:starb} to~\ref{line:star:stare} of Algorithm~\ref{alg:unk:star} during its execution of ${\tt StarCheck}(x)$. {Moreover, as explained in the previous paragraph, every agent whose label belongs to $L_x$ executes these lines twice during its execution of ${\tt StarCheck}(x)$ and each fast-only round is satisfied by only one agent.} This concludes the proof of the claim.
				\end{proofclaim}

				We need two other claims. The second one makes use of the properties related to fast-only rounds (especially those from the previous claim) to prove that all agents whose label belongs to $L_x$ start executing ${\tt StarCheck}(x)$ at node $v$ in round $t-4dk_x$. To prove this second claim, and build on it to show that the lemma is verified when $t$ is the earliest round in which an agent calls ${\tt EnsureCleanExploration}(x)$, we require the following claim.

				\begin{claim}
					\label{proof:unk:starc3}
					Let $B$ be an agent which starts executing ${\tt StarCheck}(x)$ at node $v$ in a round $r$. If round $r$ is at most $t-4dk_x$, then it is equal to $t-4dk_x$ and the execution by agent $B$ of ${\tt StarCheck}(x)$ returns {\tt true}.
				\end{claim}

				\begin{proofclaim}
					In view of Claim~\ref{proof:unk:starc2}, agent $B$ executes lines~\ref{line:star:starb} to~\ref{line:star:stare} of Algorithm~\ref{alg:unk:star} twice and for each execution of these lines by agent $B$, there is one slice whose $d$ fast-only rounds it satisfies. The earliest fast-only round is round $t-4dk_x$. In view of Lemma~\ref{proof:unk:star}, this implies that agent $B$ cannot start executing ${\tt StarCheck}(x)$ before round $t-4dk_x$ and thus that it does so in this round.

					Assume for the sake of contradiction that agent $B$ assigns {\tt false} to its variable $b$ in round $r_2$. In view of Algorithm~\ref{alg:unk:star}, in round $r_1$ when the second execution of these lines starts, the value of the variable $b$ of agent $B$ is still {\tt true} which implies $r_1<r_2$.

					Since agents $A$ and $B$ both start executing ${\tt StarCheck}(x)$ in round $t-4dk_x$ at node $v$, in view of Algorithm~\ref{alg:unk:star}, for every non-negative even integer $o \leq 4dk_x-1$, both agents $A$ and $B$ are at node $v$ and assign {\tt false} to their variable $b$ in round $t-4dk_x+o$ if and only if there are not $k_x$ agents at node $v$ in this round. If there exists a non-negative even integer $o$ at most $4dk_x-1$ such that $r_2=t-4dk_x+o$, then agent $A$ also assigns {\tt false} to its variable $b$ in round $t$, which contradicts the assumption that it never does. Hence, there is no non-negative even integer $o \leq 4dk_x-1$ such that $r_2=t-4dk_x+o$.

					Let $w$ be the rank of agent $B$ (the rank of agent $A$ is $z$). For every positive integer $p \leq 2k_x$ different from $(w+1)$, $(z+1)$, $(w+k_x+1)$, and $(z+k_x+1)$, and for every positive even integer $q \leq 2d$, round $r_2$ is different from the $q$-th round of the $p$-th slice \ie  from round $t-4dk_x+2d(p-1)+q-1$. Indeed, in this round, agents $A$ and $B$ are at node $v$ and assign {\tt false} to their variable $b$ if and only if there are not $k_x-1$ agents at node $v$ in this round.

					Round $r_2$ cannot belong to the $(w+1)$-th or $(z+1)$-th slices since their rounds are smaller than $t-2dk_x$ which is at most $r_2$. Moreover, for every positive even integer $q \leq 2d$, agent $B$ does not assign {\tt false} to its variable $b$ in the $q$-th round of the $(w+k_x+1)$-th slice \ie in round $t-2dk_x+2dw+q-1$, when it executes lines~\ref{line:star:one} to~\ref{line:star:back} of Algorithm~\ref{alg:unk:star} while the value of its variable $t$ is 2.

					Hence, there exists a positive even integer $q \leq 2d$ such that round $r_2$ is the $q$-th round of the $(z+k_x+1)$-th slice \ie round $t-2dk_x+2dz+q-1$. In this round, agent $A$ executes lines~\ref{line:star:one} to~\ref{line:star:back} of Algorithm~\ref{alg:unk:star} while agent $B$ executes lines~\ref{line:star:toktb} to~\ref{line:star:tokte} and line~\ref{line:star:tokw}. While agent $A$ is at a neighbor of node $v$, agent $B$ is at node $v$ and there are not $k_x-1$ agents at node $v$. Since in round $r_2-1$, there are $k_x$ agents at node $v$ and agent $A$ leaves this node in this round, other agents than $A$ and $B$ either leave or enter node $v$ in round $r_2-1$.

					Note that round $r_2-1$ (resp. $r_2$) is a leaving (resp. entering) round. In view of Claim~\ref{proof:unk:starc1}, no slow checking instruction consisting in entering node $v$ can be executed in round $r_2-1$. In view of Claim~\ref{proof:unk:starc2}, no other agent than $A$ can execute a fast checking instruction consisting in leaving (resp. entering) node $v$ in round $r_2-1$ (resp. $r_2$). Moreover, if an agent executed a fast checking instruction consisting in entering node $v$ in round $r_2-1$, then all its checking instructions consisting in leaving (resp. entering) node $v$ would be executed in an entering (resp. a leaving) round and it would not satisfy any round, which would contradict Claim~\ref{proof:unk:starc2}.

					Hence, in round $r_2-1$, an agent $C$ executes a slow checking instruction consisting in leaving node $v$. From round $r_2$, it waits at least $7m_x^{2m_x^5}$ rounds and thus it is not at node $v$ in round $r_2+1$. However, in this round, agent $A$ is back at node $v$ and notices $k_x$ agents at this node. This means that there is an agent $D\neq A$ which executes a checking instruction $i_D$ consisting in entering node $v$ in round $r_2$. In view of Claim~\ref{proof:unk:starc1}, instruction $i_D$ cannot be slow. But, if this instruction is fast, then agent $D$ executes a fast checking instruction consisting in leaving node $v$ either in round $r_2-1$ or in round $r_2+1$, like agent $A$, which contradicts Claim~\ref{proof:unk:starc2}. This shows that round $r_2$ cannot exist. The claim is proved.
				\end{proofclaim}

				\begin{claim}
					\label{proof:unk:starc4}
					For every non-negative integer $p<k_x$, the agent with rank $p$ starts executing ${\tt StarCheck}(x)$ in round $t-4dk_x$ at node $v$.
				\end{claim}

				\begin{proofclaim}
					Let $p$ be any positive integer at most $2k_x$ and different from $z+1$. Consider the $p$-th slice. Let $B$ be the agent which satisfies the fast-only rounds of the $p$-th slice, in view of Claim~\ref{proof:unk:starc2}.

					Let us first show that agent $B$ starts executing ${\tt StarCheck}(x)$ at node $v$. To this end, we assume for the sake of contradiction that agent $B$ starts executing ${\tt StarCheck}(x)$ at a neighbor $u$ of node $v$. First suppose that $d\geq 2$. In view of Algorithm~\ref{alg:unk:star}, agent $B$ performs two edge traversals between node $u$ and node $v$ and thus two fast checking instructions per execution of lines~\ref{line:star:starb} to~\ref{line:star:stare}, in two consecutive rounds. This contradicts Claim~\ref{proof:unk:starc2}, in view of which with only one execution of these lines, agent $B$ has to satisfy the $d$ fast-only rounds of the $p$-th slice, which are not consecutive (between every pair of fast-only rounds, there is at least one round which is not fast-only).

					Hence, suppose that $d=1$. Thus, node $u$ is the only neighbor of node $v$. The agents whose label belongs to $L_x$ start executing ${\tt StarCheck}(x)$ either at node $v$ or at node $u$. In view of Claim~\ref{proof:unk:starc2}, for every positive integer $q \leq 2k_x$, there is one agent such that all fast-only rounds of the $q$-th slice are satisfied by this agent and only by it.

					In view of Algorithm~\ref{alg:unk:star}, every agent $C$ which starts executing ${\tt StarCheck}(x)$ at node $v$ spends $2d(k_x-1)$ rounds \ie $k_x-1$ slices waiting between the two slices whose fast-only rounds it satisfies. Hence, if agent $C$ satisfies the fast-only rounds of the $q$-th slice, where $q$ is a positive integer at most $k_x$ (resp. at least $k_x+1$ and at most $2k_x$), then it also satisfies the fast-only rounds of the $(q+k_x)$-th slice (resp. the $(q-k_x)$-th slice). This is for instance the case of agent $A$.

					Let $e$ be the degree of the neighbor $u$ of node $v$. Similarly as for the agents which start executing ${\tt StarCheck}(x)$ at node $v$, every agent $D$ that starts executing ${\tt StarCheck}(x)$ at node $u$ spends $2e(k_x-1)$ rounds \ie $e(k_x-1)$ slices, waiting between the two slices whose fast-only rounds it satisfies. This implies that if agent $D$ satisfies the fast-only rounds of the $q$-th slice, where $q$ is a positive integer at most $k_x(2-e)$ (resp. at least $ek_x+1$ and at most $2k_x$), then it also satisfies the fast-only rounds of the $(q+ek_x)$-th slice (resp. the $(q-ek_x)$-th slice). This is for instance the case of agent $B$. However, the only value of degree $e$ for which this is possible is $1$, which means that nodes $u$ and $v$ have the same degree: 1. Otherwise, an agent which starts executing ${\tt StarCheck}(x)$ at node $u$ spends too many rounds between two executions of lines~\ref{line:star:starb} to~\ref{line:star:stare} of Algorithm~\ref{alg:unk:star} to satisfy the fast-only rounds of two distinct slices.

					Consequently, there are only two nodes in the graph: $u$ and $v$. Moreover, there are two agents in the graph: $A$ and $B$. Agent $A$ satisfies not only the fast-only rounds but all rounds of the $(z+1)$-th and $(z+k_x+1)$-th slices and does not satisfy any round of the $p$-th and $(p+k_x)$-th slices. Agent $B$ must satisfy all rounds of these two slices, which means that for each of these slices, it has to execute a checking instruction consisting in leaving node $v$ and then a checking instruction consisting in entering it. However, during its executions of lines~\ref{line:star:starb} to~\ref{line:star:stare} of Algorithm~\ref{alg:unk:star} from node $u$, it first enters $v$ and then leaves it. It cannot satisfy all rounds of the $p$-th and $(p+k_x)$-th slices. This contradicts Claim~\ref{proof:unk:starc0} and completes the proof that agent $B$ starts executing ${\tt StarCheck}(x)$ at node $v$.

					Let us now show that agent $B$ starts executing ${\tt StarCheck}(x)$ in round $t-4dk_x$. In view of Claim~\ref{proof:unk:starc3}, the execution of ${\tt StarCheck}(x)$ by agent $B$ cannot start in a round strictly smaller than $t-4dk_x$. This implies that agent $B$ starts executing ${\tt StarCheck}(x)$ at the earliest in round $t-4dk_x$.

					Finally, in view of Algorithm~\ref{alg:unk:star}, since agent $B$ leaves (resp. enters) node $v$ in round $t-4dk_x+2dp$ (resp. $t-4dk_x+2dp+1$), it does so while executing line~\ref{line:star:move} (resp. line~\ref{line:star:back}), and thus in view of line~\ref{line:star:test}, its rank is $p$, and it starts executing ${\tt StarCheck}(x)$ in round $t-4dk_x$. This completes the proof of the claim.
				\end{proofclaim}

				In view of Claim~\ref{proof:unk:starc4}, every agent $B$ whose label belongs to $L_x$ starts executing ${\tt StarCheck}(x)$ in round $t-4dk_x$ at node $v$. In view of Proposition~\ref{proof:unk:distances}, agent $B$ completes this execution in round $t$ at node $v$. In view of Claim~\ref{proof:unk:starc3}, the execution by agent of ${\tt StarCheck}(x)$ returns {\tt true}. This completes the proof of the lemma when round $t$ is the earliest one in which an agent calls ${\tt EnsureCleanExploration}(x)$.

				Now assume that round $t$ is not the earliest one in which an agent calls ${\tt EnsureCleanExploration}(x)$. As explained at the very beginning of this proof, in view of Proposition~\ref{proof:unk:movelabel}, agent $A$ that calls routine ${\tt EnsureCleanExploration}(x)$ in round $t$ has a label belonging to $L_x$. Let $r<t$ be the earliest round in which an agent calls ${\tt EnsureCleanExploration}(x)$. The part of this proof related to the case when $t$ is the earliest round applies to round $r$. This means that there is an agent with label $l$ that calls ${\tt EnsureCleanExploration}(x)$ in round $r$ at node $v$ if and only if $l\in L_x$. This contradicts the fact that agent $A$, whose label belongs to $L_x$ calls ${\tt EnsureCleanExploration}(x)$ in round $t>r$, and completes the proof.
			\end{proof}

			The following lemmas focus on routine ${\tt EnsureCleanExploration}$. It is the second dancing protocol. It is executed after ${\tt StarCheck}$ to ensure a clean exploration by ${\tt GraphSizeCheck}$. The next lemma is the most technical one left. The other remaining lemmas mostly consist in showing that this one can be applied. Since all agents whose labels belong to $L_x$ start a common execution of ${\tt EnsureCleanExploration}$, in view of Lemma~\ref{proof:unk:star}, the next lemma roughly states that this routine allows them to detect agents whose label does not belong to $L_x$ and which would disturb the execution of ${\tt GraphSizeCheck}$.

			\begin{lemma}
				\label{proof:unk:ensure}
				Let $x$ be a positive integer. Let $A$ and $B$ be two agents. Assume that agent $A$ starts executing ${\tt EnsureCleanExploration}(x)$ at node $v$ in round $r_1$. Assume that in round $r_2$, agent $A$ executes a move instruction belonging either to ${\tt EnsureCleanExploration}(x)$ or to ${\tt GraphSizeCheck}(x)$, and that no agent declares that the gathering is achieved before round $r_2+1$. Also assume that there exists a node $u$ at distance at most $n_x^5$ from node $v$ that agents $A$ and $B$ visit in rounds belonging to $\{r_1, \dots, r_2+1\}$. If the label of agent $B$ does not belong to $L_x$, then the execution by agent $A$ of ${\tt EnsureCleanExploration}(x)$ returns {\tt false}.
			\end{lemma}

			\begin{proof}
				In view of Lemma~\ref{proof:unk:scheme}, for every move instruction $i_B$ executed by agent $B$ in a round of $\{r_1, \dots, r_2\}$, either there exists $y\geq x$ such that instruction $i_B$ belongs to ${\tt Hypothesis}(y)$ but to none of the routines of $\mathcal{F}(y)$, or instruction $i_B$ belongs to a waiting period of ${\tt MoveToCentralNode}(x)$, or instruction $i_B$ belongs to a routine of $\mathcal{F}(x)\setminus\{{\tt MoveToCentralNode}(x)\}$. Since the label of agent $B$ does not belong to $L_x$, in view of Proposition~\ref{proof:unk:movelabel}, agent $B$ neither starts executing the waiting periods of ${\tt MoveToCentralNode}(x)$ nor the routines of $\mathcal{F}(x)\setminus\{{\tt MoveToCentralNode}(x)\}$. Thus, for every move instruction $i_B$ executed by agent $B$ in a round of $\{r_1, \dots, r_2\}$, there exists $y\geq x$ such that instruction $i_B$ belongs to ${\tt Hypothesis}(y)$ but to none of the routines of $\mathcal{F}(y)$.

				In view of Algorithms~\ref{alg:hyp} and~\ref{alg:btraversal}, every edge traversal performed while executing ${\tt Hypothesis}(y)$ but none of the routines of $\mathcal{F}(y)$ (either line~\ref{line:hyp:back} of Algorithm~\ref{alg:hyp}, or line~\ref{line:btrav:take1} of Algorithm~\ref{alg:btraversal}, or line~\ref{line:btrav:take2} of Algorithm~\ref{alg:btraversal}) is preceded and followed by a waiting period of at least $7m_y^{2m_y^5}\geq 7n_x^{2n_x^5}$ rounds. In view of Lemma~\ref{proof:unk:timebound1}, $7n_x^{2n_x^5}$ upper bounds the number of rounds required by any agent to execute lines~\ref{line:hyp:star} to~\ref{alg:check2} of Algorithm~\ref{alg:hyp} (\ie routines ${\tt StarCheck}(x)$, ${\tt EnsureCleanExploration}(x)$, and ${\tt GraphSizeCheck}(x)$). Hence, between rounds $r_1$ and $r_2$ agent $B$ performs at most one edge traversal.

				There are at most two nodes in which agent $B$ can be in every round of $\{r_1, \dots, r_2\}$. Among these at most two nodes, one is $u$, and if there is a second one, call it $w$, then it is a neighbor of node $u$. In view of Proposition~\ref{proof:unk:distances}, the distance between nodes $v$ and $v_A$ is at most $n_x-1$. In view of Lemma~\ref{proof:unk:degree}, the degree of every node at distance at most $4xm_x^5$ from $v_A$ is at most $n_x-1$. This implies the following. There is a path $\pi$ of length $n_x^5+1$ such that $\mathcal{N}(\pi, v)$ contains nodes $u$ and $w$. Moreover, all nodes of $\mathcal{N}(\pi, v)$ are at distance at most $4xm_x^5$ from node $v_A$. The degree of nodes $v$ and $u$, like every node of $\mathcal{N}(\pi, v)$, is at most $n_x-1$.

				Assume for the sake of contradiction that the execution by agent $A$ of ${\tt EnsureCleanExploration}(x)$ returns {\tt true}. Consider Algorithm~\ref{alg:ensure}, executed from node $v$  by all $k_x$ agents whose labels belong to $L_x$, starting in round $r_1$. In view of line~\ref{line:clean:back}, at the beginning of each iteration of the for loop of line~\ref{line:clean:enum}, agent $A$ is at node $v$. Each iteration of this loop consists in considering a path of length $n_x^5+1$ whose elements belong to $\{0, \dots, n_x-2\}$ and following it. Since the execution by agent $A$ of ${\tt EnsureCleanExploration}(x)$ does not return {\tt false}, it is not interrupted, and among all the paths enumerated, there is the path $\pi$. In view of line~\ref{line:clean:two}, during their common execution of ${\tt EnsureCleanExploration}(x)$ which is not interrupted before returning {\tt true}, all $k_x$ agents whose labels belong to $L_x$ follow path $\pi$ from node $v$ twice. Since agent $B$ performs at most one edge traversal between $u$ and $w$ in the meantime, there is a round in which all $k_x$ agents whose labels belong to $L_x$ execute line~\ref{line:clean:card} while at the same node (either $u$ or $w$) as agent $B$. In this round, the execution of ${\tt EnsureCleanExploration}(x)$ by each of the $k_x$ agents whose labels belong to $L_x$ (to which agent $A$ belongs) returns {\tt false}. This is a contradiction and concludes the proof.
			\end{proof}

			The intuitive notion of ``clean'' exploration presented in Section~\ref{subsec:intuition} comes into the picture with the following lemma. This notion is formally defined below. Let $x$ be a positive integer, and let $A$ be an agent that calls ${\tt GraphSizeCheck}(x)$ at a node $v$ in some round. In view of Algorithm~\ref{alg:SizeCheck}, during its execution of ${\tt GraphSizeCheck}(x)$, agent $A$ executes ${\tt EST^+}(n_x)$ from node $v$. Denote by $t_2$ and $t_3$ the rounds in which the execution by agent $A$ of ${\tt EST^+}(n_x)$ is started and completed, respectively. The execution of ${\tt GraphSizeCheck}(x)$ by agent $A$ is said to be clean if for every round $t_2\leq t\leq t_3$, agent $A$ is at node $v$ in round $t$ if and only if there is at least one other agent at this node in round $t$.

			\begin{lemma}
				\label{proof:unk:clean}
				Let $x$ be a positive integer. Let $A$ be an agent which calls ${\tt GraphSizeCheck}(x)$ in a round $t$. If no agent declares that the gathering is achieved before round $t+2k_x{\tt T}({\tt EST^+}(n_x))$, then the execution of ${\tt GraphSizeCheck}(x)$ by agent $A$ is clean.
			\end{lemma}

			\begin{proof}
				In view of Lemma~\ref{proof:unk:star}, there is a node $v$ and a round $t_1$ such that for each label $l\in L_x$ the execution of ${\tt StarCheck}(x)$ by the agent with label $l$  is completed at node $v$ in round $t_1$. In other words, agent $A$ starts executing ${\tt EnsureCleanExploration}(x)$ at node $v$ in round $t_1$, like the $k_x-1$ other agents whose labels belong to $L_x$. Moreover, since no instruction of Algorithm~\ref{alg:ensure} results in different behaviors depending on the executing agent, in every round in which at least one agent executes ${\tt EnsureCleanExploration}(x)$, all $k_x$ agents whose labels belong to $L_x$ are at the same node and execute the same instructions, until round $t_2$ in which the execution by each of them of this function returns {\tt true}. Then, in view of Proposition~\ref{proof:unk:sizecheck}, the execution of ${\tt GraphSizeCheck}(x)$ by each of the $k_x$ agents whose labels belong to $L_x$ is completed in round $t_2+2k_x{\tt T}({\tt EST^+}(n_x))$ at node $v$. In view of Algorithm~\ref{alg:SizeCheck}, the $k_x$ agents whose labels belong to $L_x$ take turns in executing ${\tt EST^+}(n_x)$ while the $k_x-1$ others wait at node $v$. Since $k_x-1\geq 1$, when agent $A$ is executing ${\tt EST^+}(n_x)$ and is at node $v$, there is at least one other agent at this node.

				Assume for the sake of contradiction that there exists a round $t_2\leq r\leq t_2+2k_x{\tt T}({\tt EST^+}(n_x))$ in which an agent $B$ occupies the same node $u\neq v$ as agent $A$. In view of Proposition~\ref{proof:unk:distances}, the distance between nodes $u$ and $v$ is at most $n_x^5$. In view of Algorithm~\ref{alg:SizeCheck}, during the execution of ${\tt GraphSizeCheck}(x)$ that all $k_x$ agents whose labels belong to $L_x$ start in round $t_2$ at node $v$, they take turns in executing ${\tt EST^+}(n_x)$ while the $k_x-1$ other agents wait at node $v$. Hence, in round $r$, when agent $A$ is at node $u\neq v$, all $k_x-1$ other agents whose label belongs to $L_x$ are at node $v$. This means that the label of agent $B$ does not belong to $L_x$. Moreover, no agent declares that the gathering is achieved before round $t_2+2k_x{\tt T}({\tt EST^+}(n_x)$. Thus, we can apply Lemma~\ref{proof:unk:ensure}, which contradicts the assumption that agent $A$ calls ${\tt GraphSizeCheck}(x)$. This concludes the proof.
			\end{proof}

			The two following lemmas state important high-level properties of our algorithm: its safety and liveness.

			\begin{lemma}
				\label{proof:unk:safety}
				Assume that an agent $A$ ends up declaring that the gathering is achieved at a node $v$ in a round $t$. Every agent declares that the gathering is achieved at node $v$ in round $t$. Moreover, in round $t$, every agent assigns to its variable $leader$ (resp. $size$) the value of the smallest label of an agent (resp. the size of the graph).
			\end{lemma}

			\begin{proof}
				First assume that $t$ is the earliest round in which an agent declares that the gathering is achieved.

				In view of Lemma~\ref{proof:unk:star}, there is a node $v$ and a round $t_1$ such that for each label $l\in L_x$ the execution by the agent with label $l$ of ${\tt StarCheck}(x)$ is completed at node $v$ in round $t_1$. In other words, agent $A$ starts executing ${\tt EnsureCleanExploration}(x)$ at node $v$ in round $t_1$, like the $k_x-1$ other agents whose labels belong to $L_x$. Moreover, since no instruction of Algorithm~\ref{alg:ensure} results in different behaviors depending on the executing agent, in every round in which at least one agent executes ${\tt EnsureCleanExploration}(x)$, all $k_x$ agents whose labels belong to $L_x$ are at the same node and execute the same instructions, until round $t_2$ in which the execution by each of them of this function returns {\tt true}. Then, in view of Proposition~\ref{proof:unk:sizecheck}, the execution of ${\tt GraphSizeCheck}(x)$ by each of the $k_x$ agents whose labels belong to $L_x$ is completed in round $t=t_2+2k_x{\tt T}({\tt EST^+}(n_x))$ at node $v$. In view of Lemma~\ref{proof:unk:clean}, all these executions of ${\tt GraphSizeCheck}(x)$ are clean. In view of the description of procedure ${\tt EST^+}$ in Section~\ref{sec:ninconnu}, agent $A$ declares that the gathering is achieved because the size $n$ of the graph is equal to $n_x$, which means that all $k_x$ agents whose labels belong to $L_x$ also declare that the gathering is achieved in round $t$ at node $v$.

				Note that since $n_x=n$, in view of line~\ref{line:gath:size} of Algorithm~\ref{alg:mainnoupper}, in round $t$, every agent whose label belongs to $L_x$ assigns to its variable $size$ the size of the graph. Moreover, if there is no agent whose label does not belong to $L_x$, then in view of line~\ref{line:gath:lead} of Algorithm~\ref{alg:mainnoupper}, in round $t$, every agent whose label belongs to $L_x$ assigns to its variable $leader$ the smallest label of an agent. Hence, if there is no agent whose label does not belong to $L_x$ and $t$ is the earliest round in which an agent declares that the gathering is achieved, then the lemma is verified.

				Still assuming that $t$ is the earliest round in which an agent declares that the gathering is achieved, let us now assume that there exists an agent $B$ whose label does not belong to $L_x$. Denote by $u$ the node (possibly different from $v$) at which agent $B$ is in round $t$.

				Let $\pi$ be a shortest path from node $v$ to node $u$. Since $n_x$ is the size of the graph, all elements of $\pi$ belong to $\{0, \dots, n_x-2\}$ and the length of $\pi$ is at most $n_x-1$. There exists a path $\zeta$ of length $n_x^5+1$ whose elements belong to $\{0, \dots, n_x-2\}$ which is prefixed by $\pi$. In view of line~\ref{line:clean:back} of Algorithm~\ref{alg:ensure}, at the beginning of each iteration of the for loop of line~\ref{line:clean:enum} of Algorithm~\ref{alg:ensure}, agent $A$ is at node $v$. Each iteration of this loop consists of considering a path of length $n_x^5+1$ whose elements belong to $\{0, \dots, n_x-2\}$ and following it. During the execution of ${\tt EnsureCleanExploration}(x)$ by agent $A$, there is an iteration of the for loop of line~\ref{line:clean:enum}, during which the value of variable $x$ is $\zeta$, and thus during which agent $A$ follows $\pi$ from node $v$, and visits node $u$, in some round $r_A$.

				Note that agent $A$ starts executing ${\tt EnsureCleanExploration}(x)$ at node $v$ in round $t_1$, and executes a move instruction belonging to ${\tt GraphSizeCheck}(x)$ in round $t-1$. Also note that agents $A$ and $B$ visit node $u$ in rounds of $\{t_1, \dots, t\}$ and the label of agent $B$ does not belong to $L_x$. In view of this and of the fact that $t$ is the earliest round in which an agent declares that the gathering is achieved, it follows from Lemma~\ref{proof:unk:ensure} that agent $A$ does not call ${\tt GraphSizeCheck}(x)$, which is a contradiction and concludes the analysis in the case when $t$ is the earliest round in which an agent declares that the gathering is achieved.

				Now assume that round $t$ is not the earliest one in which an agent declares that the gathering is achieved. Let $r<t$ be the earliest round in which an agent declares that the gathering is achieved. The part of this proof related to the case when $t$ is the earliest round applies to round $r$. This means that all agents declare that the gathering is achieved in round $r$. This contradicts the fact that agent $A$ declares that the gathering is achieved in round $t>r$, and completes the proof.
			\end{proof}

			\begin{lemma}
				\label{proof:unk:convergence}
				There is an agent that ends up declaring that the gathering is achieved.
			\end{lemma}

			\begin{proof}
				For the sake of contradiction, assume that no agent declares that the gathering is achieved. Let $x$ be a positive integer such that $\phi_x$ is the real initial configuration $\phi$. In particular, $n_x=n$, $L_x$ is the set of all the labels of the agents in the graph, and there is a single node $u_x$ such that for every $l$ in $L_x$, ${\tt path}_x(l)$ is a shortest path from the initial node of the agent with label $l$ to $u_x$. Let $F$ be the earliest agent (or one of the earliest agents) to wake-up and $l_F$ be its label. Since it never declares that the gathering is achieved, in view of Lemma~\ref{proof:unk:timebound2}, agent $F$ calls ${\tt Hypothesis}(x)$ in round $s_{F,x}$. Moreover, still in view of Lemma~\ref{proof:unk:timebound2}, every agent $A\neq F$, starts executing ${\tt Hypothesis}(x)$ in round $s_{A,x}$ which is at least $s_{F,x}$.

				In view of Algorithm~\ref{alg:hyp}, in order to show that agent $F$ ends up declaring that the gathering is achieved, it is enough to show that the fact that $\phi_x$ is the initial configuration implies that the calls of $F$ to ${\tt BallTraversal}(x)$, ${\tt MoveToCentralNode}(x)$, ${\tt StarCheck}(x)$, ${\tt EnsureCleanExploration}(x)$, and ${\tt GraphSizeCheck}(x)$ all return {\tt true}.

				In view of Lemma~\ref{proof:unk:degree} and Algorithm~\ref{alg:btraversal}, the execution by every agent of function ${\tt BallTraversal}(x)$ would return {\tt false} only if, while executing it, the agent visited a node with degree at least $n_x$, which cannot happen in a graph of $n_x=n$ nodes.

				In view of Algorithm~\ref{alg:movecentral}, the execution by every agent $A$ of function ${\tt MoveToCentralNode}(x)$ would return {\tt false} only if either there were no node in $\phi_x$ with label $l_A$ or ${\tt path}_x(l_A)$ were not a valid path in the graph. However, since $\phi_x$ is the initial configuration, neither of these conditions is satisfied.

				To prove that the call of agent $F$ to ${\tt StarCheck}(x)$ returns {\tt true}, we need the following claim.

				\begin{claim}
					\label{proof:unk:convc1}
					There exists a round $r$ in which all agents start executing ${\tt StarCheck}(x)$ at node $u_x$.
				\end{claim}

				\begin{proofclaim}
					In view of Lemmas~\ref{proof:unk:timebound2} and~\ref{proof:unk:delay1}, during their execution of ${\tt Hypothesis}(x)$, all agents start executing line~\ref{alg:prepro2} of Algorithm~\ref{alg:hyp} in round $w_{F,x}+\mathcal{S}_x$ at the latest. Then, they spend $\mathcal{S}_x$ rounds executing line~\ref{alg:prepro2} of Algorithm~\ref{alg:hyp}, and start executing ${\tt MoveToCentralNode}(x)$ in a round $r_1$ of $\{w_{F,x}+\mathcal{S}_x, \dots, w_{F,x}+2\mathcal{S}_x\}$. This means that when an agent starts executing ${\tt MoveToCentralNode}(x)$, the other agents have finished executing ${\tt BallTraversal}(x)$, and are either waiting or also starting the execution of ${\tt MoveToCentralNode}(x)$.

					In view of Lemma~\ref{proof:unk:initnode} and Proposition~\ref{proof:unk:distances}, every agent $A$ is at node $v$ when its execution of ${\tt BallTraversal}(x)$ is completed and when it starts executing ${\tt MoveToCentralNode}(x)$. Then, all agents spend at most $n_x-1$ rounds moving, and since $\phi_x$ is the initial configuration, there is a round $r_3$ of $\{w_{F,x}+\mathcal{S}_x, \dots, s_{F,x}+2\mathcal{S}_x+n_x-1\}$ in which they are all at node $u_x$.

					Note that since every agent $A$ follows a shortest path from its initial node $v_A$ to node $u_x$, there is no round in which it is at node $u_x$ and executes line~\ref{line:move:move} of Algorithm~\ref{alg:movecentral}. In other words, there is no round $r_1\leq r<r_3$ in which there are $k_x$ agents at node $u_x$. Hence, the earliest agents to reach node $u_x$ do not complete the execution of lines~\ref{line:move:whileb} to~\ref{line:move:whilee} before round $r_3$. More precisely, let $r_2$ be the earliest round in which an agent executes line~\ref{line:move:whileb}. In each round $r_2\leq r<r_3$, there are fewer than $k_x$ agents at node $v$. Round $r_2$ belongs to $\{w_{F,x}+\mathcal{S}_x, \dots, s_{F,x}+2\mathcal{S}_x+n_x-1\}$, which means that $r_3-r_2\leq \mathcal{S}_x+n_x$.

					Hence, all agents start executing line~\ref{line:move:wait} in round $r_3$. Then, in round $r_3+\mathcal{S}_x+n_x$, they all complete the execution of this line and start executing ${\tt StarCheck}(x)$ in the same round at node $u_x$. This concludes the proof of this claim.
				\end{proofclaim}

				In view of Algorithm~\ref{alg:unk:star}, the agents take turns in visiting the neighbors of node $u_x$ while the $k_x-1$ other agents are waiting. During its visits of the neighbors of node $u_x$, each agent alternatively leaves node $v$ and enters it back. In the rounds after leaving node $v$, each agent is alone at its node and does not notice a cardinality different from 1 while the $k_x-1$ other agents are at node $v$ and do not notice a cardinality different from $k_x-1$. In the next round, the agent whose turn it is to visit the neighbors is back at node $v$ and no agent notices a cardinality different from $k_x$. More precisely, if we denote by $r_1$ the round in which all agents start executing ${\tt StarCheck}(x)$ and by $d$ the degree of node $u_x$, then for every non-negative even integer (resp. non-negative odd integer) $o \leq 4dk_x$, in round $r_1+o$, there are $k_x$ agents at node $u_x$ checking whether there are $k_x$ agents at their current node (resp. $k_x-1$ agents at node $u_x$ checking whether there are $k_x-1$ agents at their current node, and one agent at a neighbor of node $u_x$ possibly checking whether there is 1 agent at its current node). Hence, no execution of ${\tt StarCheck}(x)$ returns {\tt false}. All agents start executing ${\tt EnsureCleanExploration}(x)$ in round $r_1+4dk_x$ at node $u_x$.

				In view of Algorithm~\ref{alg:SizeCheck}, in each round of their execution of ${\tt EnsureCleanExploration}(x)$, all agents execute the same instructions at the same node. Their execution of this function returns {\tt false} if and only if there is a round in which there are not $k_x$ agents at their current node, which does not occur since they stay together all along the execution.

				Finally, all agents start executing ${\tt GraphSizeCheck}(x)$ in the same round $r_2$ at the same node. In view of Proposition~\ref{proof:unk:distances}, this node is node $u_x$ where they started executing ${\tt StarCheck}(x)$. In view of Algorithm~\ref{alg:SizeCheck}, all agents take turns in executing ${\tt EST^+}(n_x)$ while the $k_x-1$ other agents wait. In view of Lemma~\ref{proof:unk:clean}, the execution of ${\tt GraphSizeCheck}(x)$ by each of them is clean. In view of the description of procedure ${\tt EST^+}$ in Section~\ref{sec:ninconnu}, {the execution by every agent of ${\tt GraphSizeCheck}(x)$ returns {\tt true} because the size $n$ of the graph is equal to $n_x$. In view of Proposition~\ref{proof:unk:sizecheck}, it does so in round $r_2+2k_x{\tt T}({\tt EST^+}(n_x))$. This concludes the proof of the lemma.}
			\end{proof}

			From Lemmas~\ref{proof:unk:safety} and~\ref{proof:unk:convergence}, we get the following theorem.

			\begin{theorem}
				\label{theo:unknown}
				Algorithm ${\tt GatherUnknownUpperBound}$ solves the gathering and leader election problems. Moreover, at the end of its execution of ${\tt GatherUnknownUpperBound}$, each agent knows the size of the graph.
			\end{theorem}

	\section{Consequences for the gossiping problem} \label{sec:elecgoss}

		In the previous sections, we have shown that the lack of direct means of communication does not preclude the gathering and leader election. In what follows, we show that the lack of direct means of communication also does not preclude communication itself, because it turns out that the feasibility of the gathering problem implies the feasibility of the gossiping problem. In the gossiping problem, each agent has a message to transmit that is a binary string (not necessarily different from other messages), and the goal for all agents is to learn all messages. Without loss of generality, we assume that each message $M$ is equal to ${\tt code}(M')$, for some binary string $M'$. The reason why we make such an assumption is given below.

		Algorithm~\ref{alg:gossip} gives the pseudocode of procedure~${\tt Gossip}$ that permits to solve the gossiping problem provided the following two conditions are satisfied: $(1)$ all agents know a common upperbound $N$ on the graph size, and $(2)$ all agents start executing the procedure in the same round and from the same node. Note that these conditions can be reached starting from any initial configuration, if the agents first apply ${\tt GatherKnownUpperBound}$ (resp. ${\tt GatherUnknownUpperBound}$), according to Theorem~\ref{theo:known} (resp. Theorem~\ref{theo:unknown}).

		\begin{algorithm}
			\begin{footnotesize}
			\caption{Algorithm~${\tt Gossip}$\label{alg:gossip}}
			\label{algo:gossip}
			\begin{algorithmic}[1]
				\Begin
					\State let $M$ be the message to be transmitted
					\State $a\gets{\tt CurCard}$; $i\gets 0$; $j\gets 2$;
						$b\gets{\tt true}$; $S\gets \varnothing$
					\While{$i\neq a$} \label{line:gossip:2}
						\State $(m, k)\gets{\tt Communicate}(j, M, b)$ \label{line:gossip:tool}
						\If{$m$ is suffixed by $01$} \label{line:gossip:4}
							\State $S\gets S\cup\{(m, k)\}$; $i\gets i+k$; $j\gets 2$
							\If{there is a couple $(M,*)\in S$}
								\State $b\gets{\tt false}$
							\EndIf
						\Else
							\State $j\gets j+2$ \label{line:gossip:fail}
						\EndIf
					\EndWhile \label{line:gossip:3}
				\End
			\end{algorithmic}
			\end{footnotesize}
		\end{algorithm}

		When executing procedure~${\tt Gossip}$, each agent handles three variables: $S,i$ and $b$. The variable $S$ is used to store couples $(m,k)$ where $m$ is a message that has been received by the agent and $k$ is a positive integer corresponding to the number of agents for which $m$ is the message. Those couples are learned via successive calls to function ${\tt Communicate}$. The variable $i$ is meant to store the number of agents whose message has been received, while the variable $b$ is a boolean whose purpose is to indicate whether the message of the executing agent has been transmitted to the members of the team or not ($b$ is initialized to {\tt true}, and then it is assigned the value {\tt false} when the message of the agent has been transmitted). The procedure is completed when variable $i$ becomes equal to the total number of agents. As it is explained below, this event means that each agent has transmitted its message.

		To be more specific, assume that all the agents know a common upper bound $N$ on the graph size and start executing procedure~${\tt Gossip}$ in some round $t$ at some node $v$. At the beginning of the execution of every agent we have $S=\emptyset$, $b=$ {\tt true} and $i=0$. Denote by $\sigma$ (resp. $k_{\sigma}$) the lexicographically smallest message among the shortest messages to be transmitted (resp. the number of agents for which $\sigma$ is the message to be transmitted). By the while loop of Algorithm~\ref{alg:gossip} and Lemma~\ref{lem:com} (which can be applied  because each message is an image of a binary string under function ${\tt code}$), it follows that for every positive even integer $j \leq |\sigma|$, the execution of function ${\tt Communicate}(j,M_A,b)$ by each agent $A$ (where ${\tt M_A}$ is the message of even length of agent $A$) is started (resp. completed) at node $v$ in round $t+\sum_{s=1}^{s=\frac{j}{2}}10(s-1){\tt T}({\tt EXPLO}(N))$ (resp. in round $t+\sum_{s=1}^{s=\frac{j}{2}}10s{\tt T}({\tt EXPLO}(N))$). Its return value is $(\sigma,k_{\sigma})$ if $|\sigma|=j$, and it is a couple whose first element contains no bit $0$, otherwise. Hence,  upon completion of round $t+\sum_{s=1}^{s=\frac{|\sigma|}{2}}10s{\tt T}({\tt EXPLO}(N))$, every agent is at node $v$ and for each agent we have: $S=\{(\sigma,k_{\sigma})\}$, $i=k_{\sigma}$, and $b=$ {\tt false} (resp. {\tt true})  if the message of the agent is $\sigma$ (resp. is not $\sigma$). At this point, if $i=k_{\sigma}$ is equal to the number of agents in the team, all agents know that all the messages have been transmitted (in fact, all agents had the same message) and stop the execution of procedure~${\tt Gossip}$. Otherwise, the agents continue their execution. Denote by $\sigma'$ (resp. $k_{\sigma'}$) the lexicographically smallest message among the shortest messages different from $\sigma$ (resp. the number of agents for which $\sigma'$ is the message to be transmitted). Using similar arguments as above, we can show that upon completion of round $t+\sum_{s=1}^{s=\frac{|\sigma|}{2}}10s{\tt T}({\tt EXPLO}(N))+\sum_{s=1}^{s=\frac{|\sigma'|}{2}}10s{\tt T}({\tt EXPLO}(N))$, every agent is at node $v$ and ifor each agent we have: $S=\{(\sigma,k_{\sigma}),(\sigma',k_{\sigma'})\}$, $i=k_{\sigma}+k_{\sigma'}$, and $b=$ {\tt false} (resp. {\tt true}) if the message of the agent is $\sigma$ or $\sigma'$ (resp. is neither $\sigma$ nor $\sigma'$). At this point, if $i=k_{\sigma}+k_{\sigma'}$ is equal to the number of agents in the team, all agents know that all the messages have been transmitted and stop the execution of procedure~${\tt Gossip}$. Otherwise, the agents still continue their execution.

		By induction on the number of messages to be transmitted, we can show that there exists a time $T$ that is polynomial in $N$ and in the length of the largest message, such that the execution of procedure~${\tt Gossip}$ by each agent is completed in round $t+T$: by then, for each agent, the set $S$ is such that $(m,k)\in S$ iff $k>0$ and there are exactly $k$ agents for which $m$ is the message to be transmitted.

		Hence from Theorem~\ref{theo:known} and Theorem~\ref{theo:unknown}, we get the following result about Algorithm~${\tt GossipKnownUpperbound}$ and Algorithm~${\tt GossipUnknownUpperBound}$: the former (resp. latter) algorithm simply consists in applying Algorithm~${\tt GatherKnownUpperBound}$ (resp. Algorithm~${\tt GatherUnknownUpperBound}$) and then applying Algorithm~${\tt Gossip}$.

		\begin{theorem}
			Assuming that the agents initially know a common upper bound $N$ on the graph size (resp. initially do not know any upper bound on the graph size), Algorithm~${\tt GossipKnownUpperbound}$ (resp. Algorithm~${\tt GossipUnknownUpperBound}$ ) solves the gossiping problem. Moreover,  the time complexity of Algorithm~${\tt GossipKnownUpperbound}$ is polynomial in the known upper bound $N$, in the length $\ell$ of the smallest label among the agents and in the length of the largest message to be transmitted.
		\end{theorem}

	\section{Conclusion}

		We designed deterministic algorithms for fundamental problems, such as gathering, leader election and gossiping in the synchronous scenario, under a model much weaker than the traditional one, in which the ability to talk among co-located agents is replaced by the mere information of how many agents are currently co-located with a given agent. It is clear that this assumption cannot be entirely removed: if agents do not know anything about the number of co-located agents, even a team of two agents would never be aware that they met.

		Our algorithm assuming the knowledge of an upper bound $N$ on the size of the network has complexity polynomial in $N$, in the length of the smallest label and, in case of gossiping, in the length of the largest message. We did not try to actually optimize this complexity: this is a long-standing open problem even in the traditional model, and even for only two agents.

		The purpose of our algorithm working without any {\em a priori} knowledge about the network is to show feasibility of gathering under this harsher scenario. The algorithm, which emulates a solution with direct means of communication in the scenario where agents are deprived of it, has time complexity exponential in the size of the network and in the labels of agents. The natural open problem yielded by our results is whether deterministic gathering, leader election and gossiping can be performed without any {\em a priori} knowledge, in time polynomial in the size of the network, in the length of the smallest label and, in case of gossiping, in the length of the largest message. Note that, at first glance, one might think that we could easily get such a complexity here by simply choosing to emulate a more efficient algorithm. However, some of the tools that we designed to enable such an emulation are heavily exponential themselves: using them with a more efficient but more technical algorithm (putting aside the question of whether such a change would be even feasible), would lead to an additional burden without improving the complexity to cross the polynomial border.
		A first step towards a polynomial solution of gathering and gossiping without direct communication and without any {\em a priori} knowledge would be to add the possibility of randomization, and design a randomized algorithm for these tasks working in polynomial time with high probability.
		
%		It is not even clear if polynomial gathering without any {\em a priori} knowledge would be possible in our model without direct communication, if randomization were allowed. Of course, meeting of many random walks
%		after polynomial time is well known, so agents could easily gather, but recall that we require simultaneous declaring that gathering is over and stopping. It is also not clear if randomization could help in achieving gossiping in polynomial time, if agents know nothing about the network.

		In this paper we considered the synchronous scenario. In the asynchronous scenario, gathering, leader election and gossiping were solved in \cite{DPV}, using the traditional  assumption that co-located agents can talk. It would be interesting to investigate if this assumption could be similarly weakened in the asynchronous case. This is not clear, as then agents can meet also inside edges (if meeting is allowed only at nodes and agents crossing each other in an edge do not notice it, asynchronous gathering is impossible). Would knowing how many agents are in every such meeting point be enough to solve the above problems, when any chatter is forbidden? Clearly, the solution presented in this paper could not be used, as it heavily relies on waiting times that cannot be controlled in the asynchronous scenario.

	\pagebreak

	%%%%%%%%%%%%%%%%%%%%%%%%%%%%%%%%%%%%%%%%%%%%%%%%%%%%%%%%%%%
	\bibliographystyle{plain}
	
	%%%%%%%%%%%%%%%%%%%%%%%%%%%%%%%%%%%%%%%%%%%%%%%%%%%%%%%%%%%

\end{document}